\documentclass[12pt]{article}
 
\usepackage{amsmath,amssymb,amstext,mathtools,amsthm,mathdots}
\usepackage[affil-it]{authblk}
\usepackage[pdftex]{graphicx} 
\usepackage{mathtools} 
\usepackage{titleps}
\usepackage{multirow}
\usepackage{relsize}

\newcommand{\TITLE}{Instantons with continuous conformal symmetries}
\newcommand{\SUBTITLE}{Hyperbolic and singular monopoles and more, oh my!}
\newcommand{\AUTHOR}{C. J. Lang}

\usepackage[pdftex,pagebackref=true]{hyperref}

\hypersetup{
    plainpages=false,       
    unicode=false,         
    pdftoolbar=true,      
    pdfmenubar=true,     
    pdffitwindow=false,   
    pdfstartview={FitH},    
    pdftitle={\TITLE: \SUBTITLE},    
    pdfauthor={\AUTHOR},    
    pdfsubject={Mathematics}, 
    pdfkeywords={instanton} {Lie theory} {moduli space} {representation theory}, 
    pdfnewwindow=true,      
    colorlinks=true,       
    linkcolor=blue,        
    citecolor=green,        
    filecolor=magenta,      
    urlcolor=cyan           
}

\usepackage{bookmark}

\providecommand{\keywords}[1]
{
  {\small	
  \textbf{\textit{Keywords---}} #1}
}

\usepackage{cite}
\makeatletter
\newcommand{\citecomment}[2][]{\citen{#2}#1\citevar}
\newcommand{\citeone}[1]{\citecomment{#1}}
\newcommand{\citetwo}[2][]{\citecomment[,~#1]{#2}}
\newcommand{\citevar}{\@ifnextchar\bgroup{;~\citeone}{\@ifnextchar[{;~\citetwo}{]}}}
\newcommand{\citefirst}{\@ifnextchar\bgroup{\citeone}{\@ifnextchar[{\citetwo}{]}}}
\newcommand{\cites}{[\citefirst}
\makeatother

\renewpagestyle{plain}{
\setfoot[][][]{}{\thepage}{}
}

\newcommand{\mf}[1]{\mathfrak{#1}}

\newtheorem{theorem}{Theorem}
\numberwithin{theorem}{section}
\newtheorem{lemma}[theorem]{Lemma}
\newtheorem{note}[theorem]{Note}
\newtheorem{prop}[theorem]{Proposition}
\newtheorem{cor}[theorem]{Corollary}

\newtheorem{definition}[theorem]{Definition}
\newtheorem{conj}[theorem]{Conjecture}

\usepackage[margin=1in]{geometry}

\makeatletter
\@addtoreset{table}{section}
\makeatother

\begin{document}

\pagestyle{plain}
\pagenumbering{arabic}

\title{\TITLE\\[0.2em]\smaller{}\SUBTITLE}
\author{\AUTHOR\thanks{E-mail address: clang@mun.ca}}
\affil{Department of Mathematics and Statistics \\ Memorial University of Newfoundland \\ St. John's, NL, Canada, A1C 5S7}
\date{October 14, 2025} 
\maketitle
\begin{abstract}
Throughout this paper, we comprehensively study instantons with every kind of continuous conformal symmetry. Examples of these objects are hard to come by due to non-linear constraints. However, by applying previous work on moduli spaces, we introduce a linear constraint, whose solution greatly simplifies these non-linear constraints. This simplification not only allows us to easily find a plethora of novel instantons with various continuous conformal symmetries and higher rank structure groups, it also provides a framework for classifying such symmetric objects. We also prove that the basic instanton is essentially the only instanton with two particular kinds of conformal symmetry. Additionally, we discuss the connections between instantons with continuous symmetries and other gauge-theoretic objects: hyperbolic and singular monopoles as well as hyperbolic analogues to Higgs bundles and Nahm data.
\end{abstract}

\cleardoublepage

\renewcommand\contentsname{Table of Contents}
\tableofcontents
\phantomsection    

\listoftables
\cleardoublepage
\phantomsection	

\section{Introduction}\label{sec:intro}
In this paper, we comprehensively examine instantons with every kind of continuous conformal symmetry. This examination is done by carefully applying two main results, Theorem~\ref{thm:mainthm} and Proposition~\ref{prop:R}, proven in previous work~\cite[Theorem~1.1 \& Proposition~1.2]{lang_moduli_2024}. For any continuous symmetry, these results produce a linear equation that determines if an object possesses that symmetry. These results not only allow us to easily compute examples of instantons, they also provide a framework for classifying them. 

The results listed below differ from their original source as in this paper, we deal with a mix of right and left actions. In particular, when dealing with compact groups, we use the following result.
\begin{theorem}{\cite[Theorem~1.1]{lang_moduli_2024}}
Let $\mathcal{X}$ be a smooth manifold, $\mathcal{G}$ a compact Lie group, and $\mathcal{S}$ a compact, connected Lie group. Suppose that $\mathcal{G}$ and $\mathcal{S}$ act smoothly on $\mathcal{X}$ on the right and left, respectively and the two actions commute. We have that $[A]\in\mathcal{X}/\mathcal{G}$ is fixed by $\mathcal{S}$ if and only if there is some Lie algebra homomorphism $\rho\colon\mathrm{Lie}(\mathcal{S})\rightarrow \mathrm{Lie}(\mathcal{G})$ such that, for all $x\in\mathrm{Lie}(\mathcal{S})$, \label{thm:mainthm}
\begin{equation}
x.A-\rho(x).A=0.\label{eq:maineq}
\end{equation}
\end{theorem}
When $\mathcal{S}$ is one-dimensional, we use the following result.
\begin{prop}{\cite[Proposition~1.2]{lang_moduli_2024}}
Let $\mathcal{X}$ be a smooth manifold, $\mathcal{G}$ a Lie group, and $\mathcal{S}$ a connected, one-dimensional Lie group (isomorphic to either $S^1$ or $\mathbb{R}$). Suppose that $\mathcal{G}$ and $\mathcal{S}$ act smoothly on $\mathcal{X}$ on the right and left, respectively, and the two actions commute. We have that $[A]\in\mathcal{X}/\mathcal{G}$ is fixed by $\mathcal{S}$ if and only if there is some $\rho\in\mathrm{Lie}(\mathcal{G})$ such that, for all $t\in\mathbb{R}$, \label{prop:R}
\begin{equation}
t.A-t\rho.A=0.\label{eq:R}
\end{equation}
\end{prop}

In Theorem~\ref{thm:mainthm}, we require $\mathcal{G}$ and $\mathcal{S}$ to be compact. In this paper, the group $\mathcal{G}$ is not compact; however, we prove some lemmas in Section~\ref{subsubsec:herdoflemmas} that allow us to focus on compact subgroups. As these results only apply to the elements in $\mathcal{X}$ that interest us and do not hold in general, we cannot use Theorem~\ref{thm:mainthm} directly, but we can use the same techniques to prove our results. As we do not require compactness for Proposition~\ref{prop:R}, we use this result directly to study instantons with circular symmetries.

This paper serves two purposes: to study symmetric instantons and showcase the power of the aforementioned results. In particular, as opposed to previous work studying instantons with finite or abelian symmetries, in this paper, we study instantons with various continuous conformal symmetries. We also identify many novel instantons with various continuous symmetries and higher rank structure groups. Additionally, we examine the relationships between instantons with continuous symmetries and various other gauge-theoretic objects. This work first appeared in my thesis, though it has been improved upon here. In particular, Section~\ref{subsec:CircularSymmetry} has been expanded to cover the equivalence of certain circular symmetries and toral symmetry~\cite[\S1.4 \& \S3]{lang_thesis_2024}.

In the rest of this section, we position this work in the instanton literature. In particular, in Section~\ref{subsec:constructions}, we examine various methods for constructing instantons. In Section~\ref{subsec:symmetric}, we investigate the many works studying symmetric instantons. In Section~\ref{subsec:connections}, we explore the connections between symmetric instantons and other gauge-theoretic objects. In Section~\ref{subsec:outline}, the main results and an outline of the paper are provided.

\subsection{Constructing instantons}\label{subsec:constructions}
Let $M$ be a 4-manifold, $E\rightarrow M$ a vector bundle over $M$ with structure group $\mathrm{Sp}(n)$, equipped with a connection $\mathbb{A}$ of curvature $F_\mathbb{A}$. Instantons are solutions to the self-dual equations $F_\mathbb{A}=\star F_\mathbb{A}$ with finite action $\frac{1}{2\pi}\int_M|F_\mathbb{A}|^2\mathrm{vol}_M$. Due to the integrability of Chern numbers, the action of an instanton, which is related to the second Chern number, is always an integer multiple of some positive constant. This multiple $N$ is known as the instanton number or charge of the instanton. Note that instantons and anti-instantons, solutions to the anti-self-dual equations $F_\mathbb{A}=-\star F_\mathbb{A}$ with finite action, differ by a choice of orientation, so there is no appreciable difference between the objects. 

Of particular interest is when $M=\mathbb{R}^4$. Due to the conformal invariance of the Hodge star acting on two-forms in four dimensions, instantons on $M=\mathbb{R}^4$ are equivalent to those on a punctured four-sphere. Uhlenbeck proved that, in fact, instantons on $M=\mathbb{R}^4$ are equivalent to those on $M=S^4$~\cite{uhlenbeck_removable_1982}.

As the self-dual equations are a system of coupled, non-linear partial differential equations, instantons are extremely difficult to obtain. There are several methods for constructing instantons. However, each method comes with its pitfalls. 

The 't Hooft ansatz is a way of constructing $\mathrm{SU}(2)$ instantons by way of harmonic maps. Using this method, the construction of an $\mathrm{SU}(2)$ instanton of charge $N$ involves choosing $N$ distinct poles and $N$ corresponding weights. Varying these choices produces a $(5N+3)$-dimensional family in the configuration space of $\mathrm{SU}(2)$ instantons with charge $N$, as a constant choice of gauge adds an additional three degrees of freedom. The JNR ansatz is a generalization of 't Hooft's work, using the conformal group to extend the choice of poles and weights by one each~\cite{jackiw_conformal_1977}. This method generates a $(5N+7)$-dimensional family in the configuration space of $\mathrm{SU}(2)$ instantons with charge $N$; only the ratios of the weights in the 't Hooft ansatz relative to this extra weight matter. For a deeper discussion on these ansatzes, see Whitehead's work~\cite[Construction B.1.2 \& Construction B.1.3]{whitehead_integrality_2022}. 

Both of the aforementioned ansatzes are straightforward methods for constructing instantons. However, they do not construct all instantons. Indeed, the moduli space of $\mathrm{SU}(2)$ instantons with charge $N$ is $8N$-dimensional. For $N=1$ and $N=2$, every $\mathrm{SU}(2)$ instanton actually can be constructed using the JNR ansatz. Note that for these cases, the extra dimensions of the family of JNR instantons can be shown to be degenerate. However, for $N\geq 3$, these ansatzes do not construct every charge $N$ instanton.

The ADHM transform is a powerful result relating instantons on $\mathbb{R}^4$ with ADHM data, matrices satisfying a non-linear constraint~\cite{atiyah_construction_1978}. We introduce this transform in more detail in Section~\ref{sec:ADHM}. Unlike the previous ansatzes, the ADHM transform constructs every instanton, even for higher rank structure groups. While they are much easier to find than instantons, finding examples of ADHM data is still challenging, due to the non-linear constraints on the data. 

There are few known examples of instantons. As such, we would like to apply Theorem~\ref{thm:mainthm} and Proposition~\ref{prop:R} to find examples of and provide a framework for classifying symmetric instantons. However, we encounter a problem: the space of all instantons and the corresponding gauge group are infinite-dimensional. Theorem~\ref{thm:mainthm} and Proposition~\ref{prop:R} rely on finite dimensional smooth manifolds. Thankfully, the ADHM transform helps us. This transform allows us to replace our infinite-dimensional spaces with finite-dimensional replacements: we take our set of objects $\mathcal{X}$ to be pairs of quaternionic matrices and our gauge group $\mathcal{G}$ to be a matrix Lie group, acting on these matrices.

Previous work finding examples of ADHM data used representation theory to simplify the ADHM data. However, this work was only done for finite or abelian groups. The current work uses representation theory and Lie theory to derive linear algebraic equations that one must first solve when examining continuous symmetries (in fact, we solve these equations in most cases). Solving these linear equations is not only very simple, it also greatly simplifies the non-linear constraints. 

\subsection{Studying symmetric instantons}\label{subsec:symmetric}
The original motivation for finding symmetric instantons was due to the relationship between instantons and skyrmions. Skyrmions are a three-dimensional topological soliton describing nuclei. For more information regarding these objects, see Manton's recent book~\cite{manton_skyrmion_2022}. A skyrmion has a topological charge $B$ representing the baryon number of the nucleus. For a given charge $B$, minimal energy skyrmions often have a great deal of symmetry. In particular, the $B=1,2,3,4$ minimal skyrmions have spherical, axial, tetrahedral, and cubic symmetry, respectively. 

Atiyah--Manton created a method for approximating the minimal energy skyrmions with charge $B$ by computing the holonomy of a $\mathrm{SU}(2)$ instanton on $\mathbb{R}^4$ with charge $n$ along lines parallel to a chosen axis~\cite{atiyah_skyrmions_1989}. This approach is quite effective, resulting in Skyrme fields with less than $1\%$ more energy than the numerical minimal energy skyrmions. This approach only applied to massless pions. However, the pair later used a trick similar to the one Atiyah used to relate hyperbolic monopoles and circle-invariant instantons in order to approximate massive pions. Indeed, Atiyah--Manton approximated Euclidean skyrmions with massive pions by first approximating them by hyperbolic skyrmions with massless pions and then approximating these hyperbolic skyrmions by computing the holonomy of a $\mathrm{SU}(2)$ instanton about circles~\cite{atiyah_skyrmions_2005}. 

The computations required to use Atiyah--Manton's method have recently been simplified using the theory of parallel transport of an induced connection, reducing the problem from solving a differential equation to a finite-difference approximation~\cite{cork_model_2021,harland_approximating_2023}. This method also bypasses issues related to gauge singularities. Harland--Sutcliffe refined this method to approximate skyrmions by using an ultra-discrete approach~\cite{harland_rational_2023}. Whereas Atiyah--Manton's method required solving a differential equation, this new method requires only evaluating the ADHM data. In particular, Harland--Sutcliffe produced Skyrme fields with charge 1 or 2, finding that their fields have even less excess energy than those generated by Atiyah--Manton.  

The existence of these symmetric minimal energy skyrmions and Atiyah--Manton's method of approximating skyrmions via instantons inspired a search for symmetric instantons. In the same work as above, Atiyah--Manton found an axial instanton with charge 2~\cite{atiyah_skyrmions_1989}. Later, Leese--Manton found a charge 3 tetrahedral instanton and a charge 4 cubic instanton~\cite{leese_stable_1994}. Singer--Sutcliffe realized that for finite symmetry groups leaving one axis invariant, the symmetry of an instanton is always generated by representations of the symmetry group~\cite{singer_symmetric_1999}. Using representation theory, Singer--Sutcliffe constructed an icosahedral instanton with charge 7 and computed a Skyrme field via the holonomy of the instanton. Additionally, they found a family of ADHM data approximating the scattering of seven skyrmions. 

Whereas Leese--Manton found their tetrahedral instanton by placing four poles of equal weight on the vertices of a regular tetrahedron and using the JNR approach, Houghton was prompted by Singer--Sutcliffe's work on symmetric instantons using ADHM data and found ADHM data with tetrahedral symmetry~\cite{houghton_instanton_1999}. Then, by varying the ADHM parameters, Houghton examined the 3-skyrmion vibration modes. Motivated by the existence of a charge 17 skyrmion with icosahedral symmetry, Sutcliffe found an instanton with the same symmetry and charge, predicting a whole range of fullerene instantons~\cite{sutcliffe_instantons_2004}. Sutcliffe also outlined a procedure for constructing instantons with Platonic symmetry, where these symmetries leave one axis invariant~\cite{sutcliffe_platonic_2005}. As a demonstration, Sutcliffe found a charge 17 instanton with the structure of a truncated icosahedron. 

Inspired by the charge five, six, and eight minimal energy skyrmions, which have the symmetry of $D_{2d}$, $D_{4d}$, and $D_{6d}$, respectively, Cork--Halcrow found instantons with the same charge and the same symmetry~\cite{cork_adhm_2022}. All of the Skyrme fields generated by the above symmetric instantons have less than $2\%$ excess energy of the corresponding numerical skyrmion~\cite{cork_adhm_2022}. Cork--Halcrow also found instantons with toral symmetry and instantons which looked like spinning tops, interpolating well-separated clusters and highly-symmetric configurations~\cite{cork_adhm_2022}.

Moving beyond symmetric instantons inspired by symmetric skyrmions, Allen--Sutcliffe constructed instantons with the symmetry of various polytopes~\cite{allen_adhm_2013}. Specifically, they found $\mathrm{SU}(2)$ instantons with the structure of a pentatope, hyperoctahedron, and the 24-cell with charges 4, 7, and 23, respectively. Just like the Platonic symmetry case, the JNR approach can be used to place equal weights at the vertices of a polytope to force the symmetry. However, there may be an instanton with the same symmetry but a lower charge, not of the JNR type, as this ansatz does not describe all instantons. Thus, ADHM data is needed. Allen--Sutcliffe used the ADHM method and representation theory to prove that their symmetric instantons have the lowest charges, excluding the basic charge one instanton. 

The preceding work relied on the finiteness of the symmetry groups to reduce the problem to representation theory. In the current work, the ability to reduce the problem to representation theory uses a completely different approach: Lie theory. Moreover, previous work focused on instantons with isometric symmetries, but the conformal equivalence of the self-dual equation begs for us to search for instantons symmetric under conformal symmetries, like Manton--Sutcliffe's conformal circle action~\cite{manton_platonic_2014}. The Lie theory approach used here allows us to consider instantons with conformal symmetries, thereby allowing us to fully study instantons symmetric under Manton--Sutcliffe's conformal circle action (which correspond to hyperbolic monopoles). We can study instantons symmetric under other conformal symmetries, however, we see in Theorem~\ref{thm:knowconfsphersym} and Corollary~\ref{cor:knowfullsym} that every instanton symmetric under other kinds of continuous conformal symmetries is constructed from the basic charge one instanton, leaving those symmetric under Manton--Sutcliffe's conformal circle action as the only ones of note.

Beckett took a different approach, using the Equivariant Index Theorem and the representation theory of $S^1$ to examine instantons symmetric under a circle action, those which correspond to singular monopoles~\cite{beckett_equivariant_2020}. Whereas previous work relied on the finiteness of symmetry groups, Beckett was able to prove that for this specific circle action, the symmetry of a circle-invariant instanton is generated by representations of $S^1$. Using this result, Beckett produced several examples of circle-invariant $\mathrm{SU}(2)$ instantons and examined the moduli spaces of such instantons for low charges. In Proposition~\ref{prop:singularmonoconst}, we answer Beckett's call for a formula for constructing a singular monopole given circle-invariant ADHM data.

Whitehead used the same approach as Beckett, returning to finite symmetries~\cite{whitehead_integrality_2022}. Doing so, Whitehead was able to prove the non-existence of instantons with certain charges and finite symmetry groups and construct novel symmetric instantons. In particular, Whitehead found a dodecahedral and an icosidodecahedral $\mathrm{SU}(2)$ instanton with charge 13 and 23, respectively. Additionally, Whitehead proved that there is no $\mathrm{SU}(2)$ instanton with charge in $(1,119)$ and the symmetry of the 600-cell, answering a specific case of the conjecture of Allen--Sutcliffe: the lowest charge $\mathrm{SU}(2)$ instanton with the symmetry of a polytope, other than the basic charge one instanton, is the one obtained using the JNR ansatz if and only if the polytope has triangles as two-faces~\cite{allen_adhm_2013}.  

As noted above, previous work on symmetric instantons focused mostly on $\mathrm{Sp}(1)\simeq\mathrm{SU}(2)$ instantons. However, there were a few works that focused on finding symmetric instantons with higher rank structure groups. Conversely to the previous work going from symmetric skyrmions to symmetric instantons, Gat searched for spherically symmetric $\mathrm{SU}(3)$ Skyrme fields by directly finding spherically symmetric $\mathrm{SU}(3)$ instantons and using Atiyah--Manton's holonomy method~\cite{gat_su_1991}. On the other hand, later work found spherically symmetric $\mathrm{SU}(N)$ skyrmions using harmonic maps~\cite{ioannidou_sun_1999}. This discovery inspired a search for the corresponding spherically symmetric $\mathrm{SU}(N)$ instantons, which was also done using harmonic maps, as opposed to ADHM data, which is used in this work~\cite{ioannidou_sun_2000}. Not only does the current work focus on a different family of structure groups, it deals with all kinds of continuous symmetries and investigates all instantons with these symmetries. Note that while the work presented here could be extended to more general structure groups, we examine instantons with the structure group $\mathrm{Sp}(n)$ because this allows for a simpler conformal action on the ADHM data.

While Singer--Sutcliffe's work using representation theory utilized two representations of the given group, Beckett and Whitehead's work only used a single representation. While this single representation still allowed for powerful results when dealing with symmetric $\mathrm{SU}(2)$ instantons, the lack of a second representation makes obtaining results with higher rank structure groups more challenging. In this work, just as Singer--Sutcliffe, we obtain and use two representations. Whereas Singer--Sutcliffe and their contemporaries focused on $\mathrm{SU}(2)$ instantons, we show how the second representation is extremely useful when consider instantons with higher rank structure groups. 

When investigating hyperbolic monopoles, Manton and Sutcliffe realized that many of the Platonic instantons previously identified possessed more symmetry than originally realized. Indeed, in addition to their Platonic symmetry, these instantons are symmetric under Manton--Sutcliffe's conformal circle action~\cite{manton_platonic_2014}. Because of these examples, one could combine the methods for finite symmetries and the methods used here to simplify the search for examples of instantons with finite symmetries and higher rank structure groups. Additionally, like Ioannidou, we could apply Atiyah--Manton's holonomy method of approximating skyrmions, using the simplifications provided by Cork, Harland, and Winyard, to the continuously symmetric instantons generated in this work~\cite{ioannidou_sun_2000}. 

Whether dealing with finite or abelian symmetries, previous work studying symmetric instantons relied on methods very specific to their cases. This work differs in that it applies general results---Theorem~\ref{thm:mainthm} and Proposition~\ref{prop:R}---to study continuous conformal,  generally nonabelian, symmetries, showcasing the power of these results. In particular, we use these results to study instantons with higher rank structure groups. By studying continuous symmetries, we are able to easily identify many novel instantons with various continuous conformal symmetries and higher rank structure groups.

\subsection{Connections with other objects}\label{subsec:connections}
In addition to finding examples of instantons, symmetric instantons are important because of their relationship with other objects. Monopoles are topological solitons over a 3-manifold $\tilde{M}$. In addition to hyperbolic monopoles, for which $\tilde{M}=H^3$, we have singular monopoles, for which there are some $p_1,\ldots,p_l\in\mathbb{R}^3$ such that $\tilde{M}=\mathbb{R}^3\setminus\{p_1,\ldots,p_l\}$. 

Using the conformal equivalence of $S^4\setminus S^2\equiv H^3\times S^1$, Atiyah realized that hyperbolic monopoles with integral mass correspond to instantons symmetric under a certain circle action~\cite{atiyah_instantons_1984,atiyah_magnetic_1984}. The integral mass condition is required so that we can extend an instanton on $S^4\setminus S^2$ to one on $S^4$~\cite{atiyah_magnetic_1984}. More detail about this correspondence is given in Section~\ref{subsubsec:HyperbolicMonopoles}. 

The Nahm transform provides a connection between Euclidean monopoles and Nahm data, solutions to the Nahm equations. Braam--Austin proved a connection between $\mathrm{SU}(2)$ instantons symmetric under a circle action, hyperbolic monopoles, and solutions to a one-dimensional integrable lattice system: a discretization of the Nahm equations~\cite{braam_boundary_1990}. Due to this connection, Ward looked at $\mathrm{SU}(2)$ instantons symmetric under the action of a torus, connecting them with axially symmetric hyperbolic monopoles~\cite{ward_symmetric_2016}. Ward found such symmetric instantons correspond to solutions of an integrable, two-dimensional lattice version of Hitchin's equations.

The relationship between instantons and singular monopoles was first described by Kronheimer, who set up a mini-twistor approach to study the moduli space in a similar fashion to Hitchin's work on Euclidean monopoles~\cite{kronheimer_monopoles_1985,hitchin_monopoles_1982}. This work was explored further by Pauly, in the case of singular monopoles on Euclidean space as well as compact spaces, where all non-flat monopoles must be singular~\cite{pauly_gauge_1996,pauly_monopole_1998}. In particular, Pauly investigated the moduli spaces of these singular monopoles with structure group $\mathrm{SU}(2)$. Additionally, when studying monopoles on $S^3$, Pauly found that associated to each gauge equivalence class of monopoles defined on an open subset $\Omega\subseteq S^3$ is a holomorphic function on the complex two-dimensional space of geodesics in $\Omega$~\cite{pauly_spherical_2001}. Royston et al. determined the dimension of the moduli space of singular monopoles with any simple Lie group as a structure group and maximal symmetry breaking~\cite{moore_parameter_2014}. 

The relationship between singular monopoles and instantons utilizes the Hopf fibration $\pi\colon S^3\rightarrow S^2$, which describes $S^3$ as a non-trivial principal bundle over $S^2$ with fibre $S^1$. That is, for every $p\in S^2$, $\pi^{-1}(p)\simeq S^1$. We can extend the fibration linearly along rays emanating from the origin to a map from $\mathbb{R}^4\setminus \{0\}$ to $\mathbb{R}^3\setminus\{0\}$. Doing so, for all $p\in\mathbb{R}^3\setminus\{0\}$ we have that its preimage under the fibration is a circle. Singular monopoles are in a one-to-one correspondence with circle-invariant instantons on $\mathbb{R}^4\setminus\{0\}$~\cite{pauly_gauge_1996}. Of particular interest are those with Dirac type singularities: the singular monopoles whose corresponding instantons on $\mathbb{R}^4\setminus\{0\}$ can be extended smoothly over the origin in some gauge, giving us an instanton on $\mathbb{R}^4$. This property allows us to describe these monopoles using ADHM data. 

Mirroring the relationships between Euclidean monopoles and the Nahm equations on a finite interval as well as hyperbolic monopoles and the discretized Nahm equation, Cherkis--Kapustin described singular monopoles in terms of solutions to the Nahm equations on a semi-infinite interval~\cite{cherkis_singular_1998,cherkis_singular_1999}. Later, Cherkis used the bow formalism to construct instanton configurations on the Taub-NUT space~\cite{cherkis_moduli_2009}. In particular, Cheshire bows correspond to singular monopoles~\cite{blair_cheshire_2011}. While all singular monopoles are created using the Nahm formalism, this method is difficult to use beyond two singularities. In contrast, the Cheshire bow formalism does not encounter this challenge. 

As previously mentioned, Manton and Sutcliffe realized that many of the Platonic instantons previously identified are also symmetric under Manton--Sutcliffe's conformal circle action~\cite{manton_platonic_2014}. Because of this additional symmetry, these instantons correspond to Platonic hyperbolic monopoles. I previously generalized the set of ADHM data considered by Manton--Sutcliffe for higher rank structure groups~\cite[Definition~1]{lang_hyperbolic_2023}. While Manton--Sutcliffe knew that their data did not describe every hyperbolic monopole with integral mass, they were unable to identify the exact constraints satisfied by all such data. In Section~\ref{subsubsec:conformalsupersphericalsym}, we identify the exact constraints. Moreover, we study all spherically symmetric hyperbolic monopoles, generalizing my previous work~\cite{lang_hyperbolic_2023}. While my previous work done on hyperbolic monopoles falls under the umbrella of the work done here, any significant overlaps have been avoided.

The connections between symmetric instantons and hyperbolic as well as singular monopoles are well-known. In this work, we showcase these connections as well as how these connections arise as opposite circular symmetries, in some sense. Additionally, we see how these connections overlap when dealing with toral symmetry. Moreover, in Sections~\ref{subsubsec:HyperbolicHiggsBundles} and \ref{subsubsec:HyperbolicNahm} we demonstrate how instantons with simple or isoclinic spherical symmetry correspond to hyperbolic analogues to other gauge-theoretic objects: Higgs bundles and Nahm data, respectively.

\subsection{Main results and outline}\label{subsec:outline}
The main results of this work are the theorems where we apply Theorem~\ref{thm:mainthm} or Proposition~\ref{prop:R} to find linear equations completely describing various symmetric topological solitons as well as the theorems where we solve these equations in general. In doing so, we provide a framework for classifying all these symmetric topological solitons. Using these results, we identify multiple novel instantons with various symmetries and higher rank structure groups in Propositions~\ref{prop:IsoEx}, \ref{prop:notinMSset}, and \ref{prop:RotInstEx}. Finally, in Theorem~\ref{thm:knowconfsphersym} and Corollary~\ref{cor:knowfullsym}, we prove that we already know every instanton with two particular kinds of conformal symmetry.  

In Section~\ref{sec:ADHMdata}, we introduce what it means for ADHM data to be symmetric under conformal transformations, proving several lemmas important to later sections. We thoroughly examine this topic in order to provide a comprehensive source on the relationship between symmetric instantons and their corresponding ADHM data. In Section~\ref{subsec:CircularSymmetry}, we investigate instantons with various kinds of circular symmetry. In particular, we look at the connection between such instantons and hyperbolic and singular monopoles. We also examine the equivalence of certain circular symmetries and toral symmetry. In Section~\ref{subsec:ToralSymmetry}, we investigate instantons with toral symmetry. In particular, we look at the connection between such instantons and hyperbolic and singular monopoles with axial symmetry. In Section~\ref{subsec:SphericalSymmetry}, we investigate instantons with the three kinds of spherical symmetry: simple, isoclinic, and conformal. We also discuss the relationship between these symmetric instantons and hyperbolic analogues to Higgs bundles and Nahm data. Additionally, we prove the aforementioned Theorem~\ref{thm:knowconfsphersym} and Corollary~\ref{cor:knowfullsym}. In Section~\ref{subsec:SuperSphericalSymmetry}, we study instantons with isoclinic and conformal superspherical symmetry. In particular, we also look at the connection between such instantons and hyperbolic analogues to Higgs bundles with axial symmetry as well as spherically symmetric hyperbolic and singular monopoles. In Section~\ref{subsec:RotationalSymmetry}, we study instantons with rotational symmetry.

In Appendix~\ref{appendix:contsubgroups}, we identify the different symmetries that we can study. However, we also introduce each symmetry in the section that examines instantons with said symmetry. As such, Appendix~\ref{appendix:contsubgroups} is not required to understand the various sections. Readers interested in the identification of continuous symmetries, however, should read through this appendix.

\section{ADHM data and conformal transformations}\label{sec:ADHMdata}

In this section, we introduce ADHM data and what it means for ADHM data to be symmetric under different conformal transformations. We thoroughly examine this topic in order to provide a comprehensive source on the relationship between symmetric instantons and their corresponding ADHM data.

\subsection{ADHM data and the ADHM transform}\label{sec:ADHM}

In this section, we introduce ADHM data and the ADHM transform.
\begin{definition}
Let $n,k\in\mathbb{N}_+$. Let $(a,b)\in\mathrm{Mat}(n+k,k,\mathbb{H})^{\oplus 2}$ be a pair of quaternionic matrices and for all $x\in\mathbb{H}$, let $\Delta(x):=a-bx$. A pair $(a,b)$ is \textbf{ADHM data} if $b$ has rank $k$ and $\Delta(x)^\dagger\Delta(x)$ is a real, positive definite matrix for all $x\in\mathbb{H}$. Let $\mathcal{N}_{n,k}$ be the set of ADHM data for given $n$ and $k$.\label{def:ADHMdata} 
\end{definition}

\begin{note}
The condition $\Delta(x)^\dagger \Delta(x)$ being positive definite for all $x\in\mathbb{H}$ does not imply that the rank of $b$ is $k$. Indeed, consider 
\begin{equation*}
\Delta(x):=a-bx=\begin{bmatrix}
1 & 0 \\
0 & 1 \\
0 & -x
\end{bmatrix}.
\end{equation*}
We see that $\Delta(x)^\dagger \Delta(x)=\begin{bmatrix}
1 & 0 \\
0 & 1+|x|^2
\end{bmatrix}$ is real and positive definite for all $x\in\mathbb{H}$. However, we see that the rank of $b$ is $1$, not $k=2$. 

Given an $\mathrm{Sp}(n)$ instanton $\mathbb{A}$, the instanton number $I$ is given by $I=\frac{1}{4\pi^2}\int_{\mathbb{R}^4}\mathrm{Tr}(F_\mathbb{A}\wedge F_\mathbb{A})\mathrm{vol}_{\mathbb{R}^4}$, the second Chern number. The ADHM transform provides a correspondence between ADHM data and instantons. Thus, some ADHM data $(a,b)$ corresponds to $\mathbb{A}$. Imposing $b$ has rank $k$ ensures that $I=k$. Suppose $(a,b)$ is $(n+k)\times k$ ADHM data with $k':=\mathrm{rank}(b)<k$. Then the corresponding instanton can also be made from $(n+k')\times k'$ ADHM data $(a',b')$ where $\mathrm{rank}(b')=k'$. Indeed, the above example generates an instanton with $I=1=k-1$, but this instanton can also be generated by $\Delta'(x)=\begin{bmatrix}
1 \\ -x
\end{bmatrix}$. 
\end{note}

We now outline some properties of ADHM data~\cite[\S 2, \S 5.1]{Corrigan_1978}.
\begin{lemma}
Let $(a,b)\in\mathcal{N}_{n,k}$. Then $a^\dagger a$ and $b^\dagger b$ are real, symmetric, and positive definite matrices. Also, $a^\dagger b$ and $b^\dagger a$ are symmetric matrices. \label{lemma:ADHMproperties}
\end{lemma}

\begin{proof}
Suppose $(a,b)\in\mathcal{N}_{n,k}$. Note that $\Delta(0)^\dagger \Delta(0)=a^\dagger a$ is real, symmetric, and positive definite. Let $x\in\mathbb{H}\setminus\{0\}$. Expanding, we see
\begin{equation}
\Delta(x)^\dagger\Delta(x)=a^\dagger b-x^\dagger b^\dagger a-a^\dagger bx+x^\dagger b^\dagger bx.\label{eq:expandDelta}
\end{equation}
Thus, we see that
\begin{equation*}
\Delta(x)^\dagger\Delta(x)+\Delta(-x)^\dagger\Delta(-x)=2(a^\dagger a+x^\dagger b^\dagger bx).
\end{equation*}
As $\Delta(x)^\dagger \Delta(x)$, $\Delta(-x)^\dagger\Delta(-x)$, and $a^\dagger a$ are real, then $x^\dagger b^\dagger bx$ is real. Thus, $|x|^4b^\dagger b$ is real, so $b^\dagger b$ is real, hence symmetric.

Furthermore, we have that the ranks of $b^\dagger b$ and $b$ are equal, so $b^\dagger b$ is a $k\times k$ matrix with rank $k$. That is, positive definite. Similarly, $a^\dagger a$ and $a$ have the same rank, so $a$ must have rank $k$.

Finally, from \eqref{eq:expandDelta}, we see that $x^\dagger b^\dagger a+a^\dagger bx$ is symmetric for all $x\in\mathbb{H}$. Taking $x=1,i,j,k$ we see that $b^\dagger a+a^\dagger b$, $-ib^\dagger a+a^\dagger b i$, $-jb^\dagger a+a^\dagger bj$, and $-kb^\dagger a+a^\dagger bk$ are all symmetric. Multiplying $b^\dagger a+a^\dagger b$ on the left by $i,j,k$ we see that $ib^\dagger a+ia^\dagger b$, $jb^\dagger a+ja^\dagger b$, and $kb^\dagger a+ka^\dagger b$ are all symmetric. Adding symmetric matrices, we see that $ia^\dagger b+a^\dagger bi$, $ja^\dagger b+a^\dagger bj$, and $ka^\dagger b+a^\dagger bk$ are symmetric. 

As $ia^\dagger b+a^\dagger bi$ is symmetric, we know for all $l,m\in\{1,\ldots,k\}$ we have
\begin{equation*}
i(a^\dagger b)_{lm}+(a^\dagger b)_{lm}i=i(a^\dagger b)_{ml}+(a^\dagger b)_{ml}i.
\end{equation*}
Rearranging, we see $i$ and $(a^\dagger b)_{ml}-(a^\dagger b)_{lm}$ anti-commute. Thus, $(a^\dagger b)_{ml}-(a^\dagger b)_{lm}$ must be in the span of $j$ and $k$. We can do the same thing with $ja^\dagger b+a^\dagger bj$ and $ka^\dagger b+a^\dagger bk$ to find that $(a^\dagger b)_{ml}-(a^\dagger b)_{lm}=0$. That is, $a^\dagger b$ is symmetric, so $b^\dagger a$ is as well.
\end{proof}

The ADHM transform from ADHM data to $\mathrm{Sp}(n)$ instantons is given as follows.
\begin{theorem}[ADHM data to instantons]
Let $(a,b)\in\mathcal{N}_{n,k}$ and $x\in\mathbb{H}$. The kernel of $\Delta(x)^\dagger$ is $n$-dimensional, so there is a $V(x)\in\mathrm{Mat}(n+k,n,\mathbb{H})$ such that $V(x)^\dagger \Delta(x)=0$ and $V(x)^\dagger V(x)=I_n$. Moreover, $V$ can be chosen to be smooth. The connection defined by $\mathbb{A}_\mu(x):=V(x)^\dagger \partial_\mu V(x)$ is an instanton. 
\end{theorem}

I omit the proof of this theorem, the converse transform, and the completeness of the transform, as they are covered in great detail by others~\cite{Corrigan_1978,corrigan_construction_1984,christ_general_1978}. 

\begin{note}
There is a freedom in the choice of $V(x)$ above. Indeed, given a smooth function $g\colon\mathbb{H}\rightarrow\mathrm{Sp}(n)$, consider $V'(x):=V(x)g(x)$. Then $V'(x)$ spans the kernel of $\Delta(x)^\dagger$. The corresponding instanton is gauge equivalent to the instanton generated using $V$. Thus, the choice of $V(x)$ does not matter to the final instanton.
\end{note}

\subsection{Conformal transformations}
In this section, we examine the conformal transformations on $\mathbb{R}^4$, equivalently $S^4$. Moreover, we examine their effect on ADHM data and their corresponding instantons. For more information on conformal transformations on $\mathbb{R}^4$, see Wilker's work~\cite{wilker_quaternion_1993}.
\begin{definition}
As in Wilker's work on M\"{o}bius groups, we adopt the convention that conjugation by $0$ is the identity~\cite[\S3]{wilker_quaternion_1993}. The \textbf{conformal group} $\mathrm{Conf}(\mathbb{R}^4)$ is isomorphic to $\mathrm{S	L}(2,\mathbb{H})$, the group of two by two quaternionic matrices given by
\begin{equation}
\mathrm{SL}(2,\mathbb{H})=\left\{\begin{bmatrix}
A & B \\ C & D
\end{bmatrix}\Bigl\vert |AC^{-1} DC-BC|=1\right\}.
\end{equation}
Note that $C$ can be zero, in which case the condition is $|AD|=1$. This group acts on $S^4\simeq \mathbb{R}^4\cup\{\infty\}\simeq\mathbb{H}\cup\{\infty\}$ as
\begin{equation}
\begin{bmatrix}
A & B \\
C & D
\end{bmatrix}.x:=\left\{\begin{array}{cl}
(Ax+B)(Cx+D)^{-1} & \quad\textrm{if }Cx+D\neq 0\\
AC^{-1} & \quad\textrm{if }x=\infty\textrm{ and }C\neq 0\\
\infty & \quad\textrm{otherwise}
\end{array}\right..\label{eq:conformalS4}
\end{equation}\label{def:conformalactionR4}
\end{definition}
 
We need to understand how this group acts on ADHM data, given by Definition~\ref{def:ADHMdata}.
\begin{definition}
Let $(a,b)\in\mathcal{N}_{n,k}$ and $\begin{bmatrix}
A & B \\ C & D
\end{bmatrix}\in\mathrm{SL}(2,\mathbb{H})$. There is a smooth right Lie group action 
\begin{equation}
\begin{bmatrix}
A & B \\ C & D
\end{bmatrix}.(a,b):=(aD-bB,bA-aC).
\end{equation}
This action is called the \textbf{conformal action} on $\mathcal{N}_{n,k}$.
\end{definition}

We have a second action acting on the ADHM data.
\begin{definition}
The gauge group $\mathrm{Sp}(n+k)\times \mathrm{GL}(k,\mathbb{R})$ acts on $\mathcal{N}_{n,k}$ as
\begin{equation}
(Q,K).(a,b):=(QaK^{-1},QbK^{-1}).
\end{equation}
This action is called the \textbf{gauge action} and is a smooth left Lie group action.
\end{definition}
 
These actions are named to reflect their effect on the corresponding instantons.
\begin{lemma} 
Let $(a,b)\in\mathcal{N}_{n,k}$. 
\begin{itemize}
\item[(1)] For $\begin{bmatrix}
A & B \\ C & D
\end{bmatrix}\in\mathrm{SL}(2,\mathbb{H})$, let $(a',b')\in\mathcal{N}_{n,k}$ be given by $(a',b'):=\begin{bmatrix}
A & B \\ C & D
\end{bmatrix}.(a,b)$. The corresponding instantons $\mathbb{A}$ and $\mathbb{A}'$ are related via the same conformal transformation: 
\begin{equation}
\mathbb{A}'=\begin{bmatrix}
A & B \\ C & D
\end{bmatrix}^*\mathbb{A}.
\end{equation}

\item[(2)] For $(Q,K)\in\mathrm{Sp}(n+k)\times\mathrm{GL}(k,\mathbb{R})$, the instantons corresponding to $(a,b)$ and $(Q,K).(a,b)$ are identical. 

\item[(3)] The gauge and conformal actions on $\mathcal{N}_{n,k}$ commute.
\end{itemize}
\end{lemma}
 
\begin{proof}
\begin{itemize}
\item[(1)] Let $\Delta(x):=a-bx$ and $\Delta'(x):=a'-b'x$. Let $V(x)\in\mathrm{Mat}(n+k,n,\mathbb{H})$ such that $V(x)^\dagger \Delta(x)=0$ and $V(x)^\dagger V(x)=I_n$. Then define $V'(x):=V\left(\begin{bmatrix}
A & B \\ C & D
\end{bmatrix}.x\right)$. As $\Delta'(x)=\Delta\left(\begin{bmatrix}
A & B \\ C & D
\end{bmatrix}.x\right)(Cx+D)$, we see that $V'(x)^\dagger \Delta'(x)=0$ and $V'(x)^\dagger V'(x)=I_n$. 
 
Therefore, we have
\begin{equation*}
\mathbb{A}_\mu'(x)=V\left(\begin{bmatrix}
A & B \\ C & D
\end{bmatrix}.x\right)^\dagger \partial_\mu V\left(\begin{bmatrix}
A & B \\ C & D
\end{bmatrix}.x\right).
\end{equation*}
Hence, just as in previous work, we have our result~\cite[Proposition~2]{lang_hyperbolic_2023}.

\item[(2)] Let $(a,b)\in\mathcal{N}_{n,k}$ and for some $(Q,K)\in\mathrm{Sp}(n+k)\times\mathrm{GL}(k,\mathbb{R})$, let $(a',b'):=(Q,K).(a,b)$. Let $V(x)\in\mathrm{Mat}(n+k,n,\mathbb{H})$ such that $V(x)^\dagger \Delta(x)=0$ and $V(x)^\dagger V(x)=I_n$. Then define $V'(x):=QV(x)$. We see that
\begin{equation*}
V'(x)^\dagger \Delta'(x)=V(x)^\dagger Q^\dagger (a'-b'x)=V(x)^\dagger \Delta(x)K^{-1}=0.
\end{equation*}
Furthermore, we see that $V'(x)^\dagger V'(x)=I_n$. Therefore, we can use $V$ and $V'$ to construct the corresponding instantons.

Let $\mathbb{A}$ and $\mathbb{A}'$ be the instantons corresponding to $(a,b)$ and $(a',b')$, respectively. Then
\begin{equation*}
\mathbb{A}'_\mu(x)=V(x)^\dagger Q^\dagger \partial_\mu QV(x)=\mathbb{A}_\mu(x).
\end{equation*}
Hence, $\mathbb{A}'=\mathbb{A}$. 

\item[(3)] This fact follows from the definition of the actions.\qedhere
\end{itemize}
\end{proof}
 
We are only interested in instantons up to gauge. That is, we are interested in equivalence classes of instantons in the moduli space of instantons, given by $\mathcal{N}_{n,k}/(\mathrm{Sp}(n+k)\times \mathrm{GL}(k,\mathbb{R}))$. As the actions commute, the conformal action descends to an action on the moduli space. When searching for symmetric instantons, we are looking at fixed points of this action.

Every instanton has a standard form, induced by the gauge action.
\begin{lemma}
To each $(a,b)\in\mathcal{N}_{n,k}$, there is a unique pair $(\hat{M},U)\in\mathcal{N}_{n,k}$ given by $U:=\begin{bmatrix}
0 \\ I_k
\end{bmatrix}$ and $\hat{M}:=\begin{bmatrix}
L \\ M
\end{bmatrix}$ whose lower block $M$ is symmetric, such that $[(\hat{M},U)]=[(a,b)]$. \label{lemma:stdform}
\end{lemma}

\begin{definition}
The collection of $\hat{M}$ such that $(\hat{M},U)\in\mathcal{N}_{n,k}$ is denoted $\tilde{\mathcal{N}}_{n,k}$ and instantons of this form are said to be in \textbf{standard form}.
\end{definition}
 
\begin{proof}
Let $(a,b)\in\mathcal{N}_{n,k}$. Given $(Q,K)\in\mathrm{Sp}(n+k)\times\mathrm{GL}(k,\mathbb{R})$, consider $(a',b'):=(Q,K).(a,b)$. Under this transformation, we see that $(b')^\dagger b'=K^T b^\dagger b K$. By Lemma~\ref{lemma:ADHMproperties}, we know that $b^\dagger b$ is a real, symmetric, and positive definite matrix. Thus, it is orthogonally diagonalizeable with real eigenvalues. That is, there is some $J\in\mathrm{O}(k)$ and $\lambda_1,\ldots,\lambda_k>0$ such that $J^Tb^\dagger bJ=\mathrm{diag}(\lambda_1,\ldots,\lambda_k)$. Let $D:=\mathrm{diag}(1/\sqrt{\lambda_1},\ldots,1/\sqrt{\lambda_k})$. Letting $K:=JD\in\mathrm{GL}(k,\mathbb{R})$ and $Q=I_{n+k}$, we note that $K^Tb^\dagger bK=I_k$. Thus, $(a,b)$ is gauge equivalent to some ADHM data with $(b')^\dagger b'=I_k$. 

Now, suppose that $(a,b)\in\mathcal{N}_{n,k}$ such that $b^\dagger b=I_k$. Let $\tilde{Q}$ be a $n\times (n+k)$ quaternionic matrix satisfying $\tilde{Q}\tilde{Q}^\dagger=I_n$ and $\tilde{Q}b=0$. These two equations impose $n(n+k)$ conditions on $\tilde{Q}$, so such a $\tilde{Q}$ exists. Then let $Q:=\begin{bmatrix}
\tilde{Q} \\ b^\dagger
\end{bmatrix}\in\mathrm{Sp}(n+k)$ and note that $Qb=\begin{bmatrix}
0 \\ I_k
\end{bmatrix}$. Letting $K=I_k$, we see that $(a,b)$ is gauge equivalent to ADHM data with $b=\begin{bmatrix}
0 \\ I_k
\end{bmatrix}$. 

Finally, suppose that $(a,b)\in\mathcal{N}_{n,k}$ such that $b=\begin{bmatrix}
0 \\ I_k
\end{bmatrix}$. Write $a=\begin{bmatrix}
L \\ M
\end{bmatrix}$. By Lemma~\ref{lemma:ADHMproperties}, we know that $M=b^\dagger a$ is symmetric. Combining results, we have that any ADHM is equivalent to an element of $\tilde{\mathcal{N}}_{n,k}$.
\end{proof}
 
The gauge action does not physically alter the instanton, so if a conformal action can be undone by a gauge action, then the conformal action physically does not alter the instanton. 
\begin{definition}
Given $\tilde{A}\in\mathrm{SL}(2,\mathbb{H})$, ADHM data $\hat{M}\in\tilde{\mathcal{N}}_{n,k}$ is said to be $\tilde{A}$-\textbf{equivariant} if $\tilde{A}.[(\hat{M},U)]=[(\hat{M},U)]$, that is, if there exists some $(Q,K)\in\mathrm{Sp}(n+k)\times\mathrm{GL}(k,\mathbb{R})$ such that $(Q,K).(\hat{M},U)=\tilde{A}.(\hat{M},U)$. Note that we can extend this concept to $\mathcal{N}_{n,k}$.\label{def:instantonequivar}
\end{definition}
 
Note that, since the gauge and conformal actions commute, if two instantons are related by a gauge transformation, then they are both equivariant under the same conformal transformations.

In order to use Theorem~\ref{thm:mainthm}, we need the group of symmetries to be compact. Right now, our symmetry group is $\mathrm{SL}(2,\mathbb{H})$, which is not compact. Just as we focus on orthogonal transformations when discussing symmetries of monopoles, we wish to reduce our symmetry group to a compact one. To that end, we investigate the symmetries of the basic instanton.
\begin{prop}
Consider the basic instanton $\hat{M}=\begin{bmatrix}
1 & 0 
\end{bmatrix}^T\in\tilde{\mathcal{N}}_{1,1}$. The subgroup $H_{\hat{M}}\subseteq \mathrm{SL}(2,\mathbb{H})$ of symmetries of the instanton is exactly $\mathrm{Sp}(2)$. \label{prop:basicsymmetries}
\end{prop}

\begin{proof}
Suppose that $(\hat{M},U)$ is equivariant under $\begin{bmatrix}
A & B \\ C & D
\end{bmatrix}\in\mathrm{SL}(2,\mathbb{H})$. Then there exists $Q\in\mathrm{Sp}(2)$ and $K\in\mathrm{GL}(1,\mathbb{R})$ such that $(Q,K).\begin{bmatrix}
A & B \\ C & D
\end{bmatrix}.(\hat{M},U)=(\hat{M},U)$. Thus,
\begin{align*}
\begin{bmatrix}
D \\ -B
\end{bmatrix}K^{-1}&=Q^\dagger \begin{bmatrix}
1 \\ 0
\end{bmatrix};\\
\begin{bmatrix}
-C \\ A
\end{bmatrix}K^{-1} &=Q^\dagger \begin{bmatrix}
0 \\ 1
\end{bmatrix}.
\end{align*}
Solving these equations, we see that $Q=\frac{1}{K}\begin{bmatrix}
D^\dagger & -B^\dagger \\ -C^\dagger & A^\dagger
\end{bmatrix}$. As $Q\in\mathrm{Sp}(2)$, we see that 
\begin{equation}
I_2=QQ^\dagger=\frac{1}{K^2}\begin{bmatrix}
|D|^2+|B|^2 & -D^\dagger C-B^\dagger A \\ -C^\dagger D-A^\dagger B & |A|^2+|C|^2
\end{bmatrix}.\label{eq:basicsymeqs}
\end{equation}

As $\begin{bmatrix}
A & B \\ C & D
\end{bmatrix}\in\mathrm{SL}(2,\mathbb{H})$, we have that $|AC^{-1}DC-BC|=1$. If $C=0$, we have that $|AD|=1$. Furthermore, we have from the above identities that $A^\dagger B=0$ and $|A|^2=K^2$. Thus, $B=0$, so $|D|^2=K^2$. Hence, $1=|AD|=K^2$.

If $C\neq 0$, using \eqref{eq:basicsymeqs}, we see that as $|AC^{-1}DC-BC|=1$,
\begin{equation*}
1=\frac{|A|^2+|C|^2}{|C|}|B|=K^2\frac{|B|}{|C|}.
\end{equation*}
That is, $K^2|B|=|C|$. Hence, 
\begin{equation*}
I_2=Q^\dagger Q=\frac{1}{K^2}\begin{bmatrix}
|D|^2+|C|^2 & -DB^\dagger -CA^\dagger \\ -BD^\dagger -AC^\dagger & |A|^2+|B|^2
\end{bmatrix}.
\end{equation*}
Thus, we see that $|A|^2+|B|^2=K^2=|A|^2+|C|^2$, so $|B|=|C|$. Thus, $K^2=1$.

In either case, $K^2=1$. Thus, we see that
\begin{equation*}
\begin{bmatrix}
A & B \\ C & D
\end{bmatrix}\begin{bmatrix}
A & B \\ C & D
\end{bmatrix}^\dagger=\begin{bmatrix}
|A|^2+|B|^2 & AC^\dagger+BD^\dagger \\ CA^\dagger+DB^\dagger & |C|^2+|D|^2
\end{bmatrix}=I_2.
\end{equation*}
That is, $\begin{bmatrix}
A & B \\ C & D
\end{bmatrix}\in\mathrm{Sp}(2)$.

Conversely, given $\begin{bmatrix}
A & B \\ C & D
\end{bmatrix}\in\mathrm{Sp}(2)$, we define $Q:=\begin{bmatrix}
D^\dagger & -B^\dagger \\ -C^\dagger & A^\dagger
\end{bmatrix}\in\mathrm{Sp}(2)$. Then 
\begin{equation*}
(Q,1).\begin{bmatrix}
A & B \\ C & D
\end{bmatrix}.(\hat{M},U)=(\hat{M},U).
\end{equation*} 
As the transformation was arbitrary, $H_{\hat{M}}=\mathrm{Sp}(2)$.
\end{proof}

Proposition~\ref{prop:basicsymmetries} tells us that the symmetry group of the basic instanton is $\mathrm{Sp}(2)$. Thus, instantons related by a conformal transformation $\tilde{A}\in\mathrm{SL}(2,\mathbb{H})$ have a symmetry group of the form $\tilde{A}^{-1}\mathrm{Sp}(2)\tilde{A}$. It is not known if similar results hold for all non-flat instantons. However, as the basic instanton is the $\mathrm{Sp}(1)$ instanton with the most symmetry, we make the following conjecture. 
\begin{conj}
Suppose that $\mathbb{A}$ is a non-flat instanton and define $H_{\mathbb{A}}\subseteq\mathrm{SL}(2,\mathbb{H})$ to be its group of symmetries. Then there is some $\tilde{A}\in\mathrm{SL}(2,\mathbb{H})$ such that $\tilde{A}H_{\mathbb{A}}\tilde{A}^{-1}\subseteq \mathrm{Sp}(2)$.\label{conj}
\end{conj}

Assuming this conjecture, when discussing symmetric instantons, we need only examine subgroups of $\mathrm{Sp}(2)$. In Appendix~\ref{appendix:contsubgroups}, we classify all the connected, Lie subgroups of $\mathrm{Sp}(2)$ up to conjugacy.

\begin{note}
The validity of this conjecture has no impact on the following results outside of Lemma~\ref{lemma:hypermonoR1}, on which nothing else relies. If it is false, then there are other kinds of symmetric instantons and perhaps there are other kinds of circle actions relating instantons and hyperbolic monopoles.
\end{note}

In an effort to ignore uninteresting instantons, we focus on the following subset of ADHM data.
\begin{definition}
Let $\mathcal{M}_{n,k}$ be the set of $\hat{M}\in\tilde{\mathcal{N}}_{n,k}$ such that
\begin{itemize}
\item $M$ (a $k\times k$ quaternionic matrix) is symmetric,
\item $L$ (a $n\times k$ quaternionic matrix) is such that $LL^\dagger$ is a positive definite matrix (we already know it is positive semi-definite),
\item let $\Delta(x):=\begin{bmatrix}
L \\ M-I_kx
\end{bmatrix}$, for all $x\in\mathbb{H}$, then $R:=\Delta(0)^\dagger \Delta(0)=L^\dagger L+M^\dagger M$ is real and non-singular, 
\item $\Delta(x)^\dagger \Delta(x)$ is non-singular for all $x\in\mathbb{H}$.
\end{itemize}\label{def:Mstd}
\end{definition}

\begin{note}
Given the form of $R$, it is real and non-singular if and only if it is symmetric and non-singular. Additionally, if $R$ is real, the $\Delta(x)^\dagger \Delta(x)$ is automatically real.
\end{note}

\begin{note}
Throughout this work, $k$ indicates both the instanton number and the standard quaternion ($k=ij$) and its use at any given time is contextual.\label{note:koverload}
\end{note}

The condition $LL^\dagger$ is positive definite forces $n\leq k$ and is a natural condition for finding novel instantons. Indeed, otherwise, we have an instanton embedded trivially as an instanton with a higher rank structure group. We ignore these as they are reducible.
\begin{prop}
Suppose that $\hat{M}$ satisfies all of the above conditions of $\mathcal{M}_{n,k}$ in Definition~\ref{def:Mstd}, except $LL^\dagger$ is not positive definite. Then there is some $L'$ such that, up to gauge, $L=\begin{bmatrix}
L' \\ 0
\end{bmatrix}$ and $L'L'^\dagger$ is a positive definite $n'\times n'$ matrix. Then $\hat{M}':=\begin{bmatrix}
L' \\ M
\end{bmatrix}\in\mathcal{M}_{n',k}$. Letting $\mathbb{A}'$ and $\mathbb{A}$ be the instantons corresponding to $\hat{M}'$ and $\hat{M}$, respectively, we have that $\mathbb{A}_\mu=\mathbb{A}'_\mu\oplus 0$.
\end{prop}

\begin{proof}
By hypothesis, there is a vector $v\neq 0$ such that $|L^\dagger v|^2=v^\dagger LL^\dagger v=0$. Choose an orthonormal basis $\{e_1,\ldots,e_{n-1},v/|v|\}$. Let $q^\dagger=\begin{bmatrix}
e_1 & \cdots & e_{n-1} & v/|v|
\end{bmatrix}$. Thus, $q\in\mathrm{Sp}(n)$ and as $v^\dagger L=0$,
\begin{equation*}
qL=\begin{bmatrix}
e_1^\dagger L \\ \vdots \\ e_{n-1}^\dagger L \\ 0
\end{bmatrix}.
\end{equation*}
Thus, we have a zero row in $qL$. We can do this for more rows by repeating this process until we obtain some matrix $L'$ such that $L'L'^\dagger$ is positive definite $n' \times n'$ matrix and $L=\begin{bmatrix}
L' \\ 0
\end{bmatrix}$. Hence, $\hat{M}':=\begin{bmatrix}
L' \\ M
\end{bmatrix}\in\mathcal{M}_{n',k}$.

Suppose that $\tilde{V}(x)=\begin{bmatrix}
\psi'(x) \\ v'(x)
\end{bmatrix}$ is the data that gives us $\mathbb{A}'_\mu$ for $\hat{M}'$. Then let $V(x):=\begin{bmatrix}
\psi'(x) & 0 \\
0 & I_m \\
v'(x) & 0
\end{bmatrix}$. We see that $V(x)$ is normalized and satisfies $V(x)^\dagger \Delta(x)=0$. Therefore, we see that $\mathbb{A}_\mu(x)=\mathbb{A}'_\mu(x)\oplus 0$. Thus, the instanton is indeed just an embedding of one with a smaller rank structure group.
\end{proof}

There are a variety of connected Lie subgroups of $\mathrm{Sp}(2)$. In Appendix~\ref{appendix:contsubgroups}, we classify all the different subgroups $H$ of interest. In particular, they are all compact or one-dimensional.
\begin{definition} 
Following Conjecture~\ref{conj}, for $\hat{M}\in\mathcal{M}_{n,k}$, let $H_{\hat{M}}\subseteq\mathrm{Sp}(2)$ be the set of conformal transformations under which $\hat{M}$ is equivariant. As we are dealing with group actions, $H_{\hat{M}}$ is a subgroup. We are interested in the case that $H_{\hat{M}}$ is a connected Lie subgroup. 

When $H_{\hat{M}}$ contains a subgroup conjugate to one of the groups listed in Table~\ref{table:conformalsubgroups}, we say that $\hat{M}$ has the corresponding symmetry given in the table.
\end{definition}

\begin{table}[t]
\centering
\begin{tabular}{|c|c|c|}
\hline
$G$ & Name of symmetry & Lie algebra\\\hline
$R_t$, for $t\in[0,1]$ & Circular $t$ & $\left\langle\mathrm{diag}(i,it)\right\rangle$\\
$S^1\times S^1$ & Toral & $\left\langle\mathrm{diag}(i,0),\mathrm{diag}(0,i)\right\rangle$\\
\multirow{3}{*}{$\mathrm{Sp}(1)$} & Simple spherical & $\mf{h}_{3,1,1}$\\
& Isoclinic spherical & $\mf{h}_{4,1}$\\
& Conformal spherical & $\mf{h}_5$\\
$(S^1\times\mathrm{Sp}(1))/\{\pm (1,1)\}$ & Conformal superspherical & $\mf{p}_{3,1,1}$\\
$S^1\times\mathrm{Sp}(1)$ & Isoclinic superspherical & $\mf{p}_{4,1}$\\
$\mathrm{Sp}(1)\times\mathrm{Sp}(1)$ & Rotational & $\mf{sp}(1)\oplus\mf{sp}(1)$\\
$\mathrm{Sp}(2)$ & Full & $\mf{sp}(2)$\\\hline
\end{tabular}
\caption[Connected Lie subgroups of $\mathrm{Sp}(2)$]{List of connected Lie subgroups $G$ of $\mathrm{Sp}(2)$, of non-zero dimension, with corresponding Lie algebras needed to classify all symmetric instantons. For the definitions of $\mf{h}_{3,1,1}$, $\mf{h}_{4,1}$, and $\mf{h}_5$, see Proposition~\ref{prop:sp1subalgebras}. For the definitions of $\mf{p}_{3,1,1}$ and $\mf{p}_{4,1}$, see Proposition~\ref{prop:sp1s1subalgebras}.}\label{table:conformalsubgroups}
\end{table} 

\subsection{A herd of lemmas}\label{subsubsec:herdoflemmas}

In this section, we prove some lemmas that prove to be quite useful in our investigation of symmetric instantons.

As we are interested in Lie group actions, we obtain induced Lie algebra actions. In the following result, we determine exactly how ADHM data in standard form transforms under the action of elements of $\mathfrak{sp}(2)$.
\begin{lemma}
Given $\hat{M}\in\mathcal{M}_{n,k}$ and $A\in\mathrm{Sp}(2)$, let $(\hat{M}_A,U_A):=A.(\hat{M},U)$. We have that $\hat{M}_A$ and $U_A$ satisfy
\begin{equation}
\begin{bmatrix}
\hat{M}_A & U_A
\end{bmatrix}=\begin{bmatrix}
\hat{M} & U
\end{bmatrix}I_k\otimes \left(\begin{bmatrix}
0 & 1 \\ -1 & 0
\end{bmatrix}A\begin{bmatrix}
0 & -1 \\ 1 & 0
\end{bmatrix}\right).\label{eq:conformalaction}
\end{equation}\label{lemma:conformalaction}

Given $\Upsilon\in\mathfrak{sp}(2)$, consider $A(\theta):=e^{\theta\Upsilon}$. Let $\hat{M}_\Upsilon:=\hat{M}_{A}'(0)$ and $U_\Upsilon:=U_{A}'(0)$. Then
\begin{equation}
\begin{bmatrix}
\hat{M}_\Upsilon & U_\Upsilon 
\end{bmatrix}=\begin{bmatrix}
\hat{M} & U
\end{bmatrix} I_k\otimes \left(\begin{bmatrix}
0 & 1 \\ -1 & 0
\end{bmatrix}\Upsilon\begin{bmatrix}
0 & -1 \\ 1 & 0
\end{bmatrix}\right).\label{eq:defLiealgaction}
\end{equation}
Note that $\hat{M}_{-\Upsilon}=-\hat{M}_\Upsilon$ and similarly for $U$.
\end{lemma}

\begin{proof}
Expanding the right-hand side of \eqref{eq:conformalaction}, we see that it exactly matches the definitions of $\hat{M}_A$ and $U_A$. The second claim follows by differentiating \eqref{eq:conformalaction}. The final claim follows from comparing \eqref{eq:defLiealgaction} for $-\Upsilon$ and $\Upsilon$.
\end{proof}

In order to use Theorem~\ref{thm:mainthm}, we need the gauge group to be compact. Our gauge group is $\mathrm{Sp}(n+k)\times \mathrm{GL}(k,\mathbb{R})$, which is not compact. However, it turns out that when investigating $\tilde{A}$-equivariance, $K$ is determined uniquely by $Q$ and $\tilde{A}$. Thus, we can ignore the non-compact part of the gauge group, leaving us with a compact one.
\begin{lemma}
Let $\hat{M}\in\mathcal{M}_{n,k}$ be $\begin{bmatrix}A & B \\ C & D\end{bmatrix}$-equivariant. More specifically, let $(Q,K)\in\mathrm{Sp}(n+k)\times\mathrm{GL}(k,\mathbb{R})$ be a gauge transformation realizing this equivariance, that is 
\begin{equation*}
\begin{bmatrix}A & B \\ C & D\end{bmatrix}.(Q,K).(\hat{M},U)=(\hat{M},U).
\end{equation*} 
Then $K$ is determined by \label{lemma:KfromQ}
\begin{equation}
K=U^TQ(UA-\hat{M}C).
\end{equation}
\end{lemma}

\begin{proof}
The equivariance condition implies $Q(UA-\hat{M}C)K^{-1}=U$. The result follows.
\end{proof}

\begin{note}
We see that $K$ depends smoothly on $Q$ and $\begin{bmatrix}
A & B \\ C & D
\end{bmatrix}$.
\end{note}

Above, we reduced the gauge group from $\mathrm{Sp}(n+k)\times\mathrm{GL}(k,\mathbb{R})$ to $\mathrm{Sp}(n+k)$, which is compact. The following result tells us that when dealing with isometries, we may choose to reduce the gauge group in a different way, reducing it to $\mathrm{O}(k)$, which is much smaller than $\mathrm{Sp}(n+k)$.

Given that we focus on symmetries contained in $\mathrm{Sp}(2)$ and Definition~\ref{def:conformalactionR4} tells us how this group acts on $S^4$, we see that the only isometries on $\mathbb{R}^4$ in this group belong to the subgroup $\mathrm{Sp}(1)\times\mathrm{Sp}(1)\subseteq \mathrm{Sp}(2)$. Note that $\mathrm{Sp}(1)\times\mathrm{Sp}(1)\simeq\mathrm{Spin}(4)$, which is the double cover of $\mathrm{SO}(4)$.
\begin{lemma}
Let $\hat{M}\in\mathcal{M}_{n,k}$. Let $\mathrm{diag}(a,b)\in\mathrm{Sp}(1)\times\mathrm{Sp}(1)\subseteq \mathrm{Sp}(2)$ such that $\hat{M}$ is $\mathrm{diag}(a,b)$-equivariant. Specifically, let $(Q,K)\in \mathrm{Sp}(n+k)\times\mathrm{GL}(k,\mathbb{R})$ be such that 
\begin{equation*}
\mathrm{diag}(a,b).(Q,K).(\hat{M},U)=(\hat{M},U).
\end{equation*} 
Then, $K\in\mathrm{O}(k)$ and there is some $q\in\mathrm{Sp}(n)$ such that $Q=\mathrm{diag}(q,a^\dagger K)$.\label{lemma:isomsimplify}
\end{lemma}

\begin{proof}
Suppose that $\mathrm{diag}(a,b).(Q,K).(\hat{M},U)=(\hat{M},U)$. Then $QUaK^{-1}=U$. Simplifying, we have that $Q$ must have the form given in the statement. That $K\in\mathrm{O}(k)$ is because $Q\in\mathrm{Sp}(n+k)$. 
\end{proof}

Lemma~\ref{lemma:KfromQ} tells us that $Q$ and the conformal transformation in $\mathrm{Sp}(2)$ determine $K$. Furthermore, Lemma~\ref{lemma:isomsimplify} tells us that if the conformal transformation is actually an isometry, then $K\in\mathrm{O}(k)$ and $Q$ must be block diagonal. For such transformations, $K$ determines $Q$. This result is similar to previous work dealing with hyperbolic monopoles~\cite[Lemma~4]{lang_hyperbolic_2023}.
\begin{lemma}
Suppose $\hat{M}\in\mathcal{M}_{n,k}$, $\mathrm{diag}(a,b)\in\mathrm{Sp}(2)$, and suppose there is some $K\in\mathrm{O}(k)$ such that $a^\dagger KMK^T=Mb^\dagger$. If $[K,R]=0$, then there exists a unique $q\in\mathrm{Sp}(n)$ such that $q LK^T=Lb^\dagger$. Specifically,\label{lemma:focusonY}
\begin{equation}
\left(\mathrm{diag}(
Lb^\dagger KL^\dagger (LL^\dagger)^{-1}, a^\dagger K
),K\right).\mathrm{diag}(a,b).(\hat{M},U)=(\hat{M},U).
\end{equation}
That is, $\hat{M}$ is equivariant under $\mathrm{diag}(a,b)$.
\end{lemma}

\begin{proof}
If $q$ and $\tilde{q}$ satisfy the constraint, then $(q-\tilde{q})LK^T=0$. As $LL^\dagger$ is positive definite, then $q=\tilde{q}$, so $q$ is unique.

Consider $q:=Lb^\dagger KL^\dagger (LL^\dagger)^{-1}$. We see that
\begin{align*}
q^\dagger q&=(LL^\dagger)^{-1} LK^T bL^\dagger Lb^\dagger KL^\dagger(LL^\dagger)^{-1}\\
&=(LL^\dagger)^{-1} LK^T b(R-M^\dagger M)b^\dagger KL^\dagger(LL^\dagger)^{-1}\\
&=(LL^\dagger)^{-1} LK^T (R-bM^\dagger a^\dagger a Mb^\dagger) KL^\dagger(LL^\dagger)^{-1}\\
&=(LL^\dagger)^{-1} LK^T (R-KM^\dagger MK^T) KL^\dagger(LL^\dagger)^{-1}\\
&=(LL^\dagger)^{-1}L(R-M^\dagger M)L^\dagger (LL^\dagger)^{-1}\\
&=(LL^\dagger)^{-1}LL^\dagger LL^\dagger (LL^\dagger)^{-1}\\
&=I_n.
\end{align*}
Hence, $q\in\mathrm{Sp}(n)$. 

We see that
\begin{align*}
Lb^\dagger KL^\dagger LL^\dagger &=Lb^\dagger K(R-M^\dagger M)L^\dagger\\
&=L(Rb^\dagger K-b^\dagger KM^\dagger K^TKMK^TK)L^\dagger\\
&=L(Rb^\dagger -b^\dagger b M^\dagger a^\dagger aM b^\dagger)KL^\dagger\\
&=L(R-M^\dagger M)b^\dagger KL^\dagger\\
&=LL^\dagger Lb^\dagger KL^\dagger,
\end{align*}
so $[Lb^\dagger KL^\dagger, LL^\dagger]=0$. Therefore, $q=(LL^\dagger)^{-1}Lb^\dagger KL^\dagger$, so as $a^\dagger KMK^T=Mb^\dagger$ and $[K,R]=0$, we have
\begin{align*}
qLK^T&=(LL^\dagger)^{-1}Lb^\dagger KL^\dagger LK^T\\
&=(LL^\dagger)^{-1}Lb^\dagger K(R-M^\dagger M)K^T\\
&=(LL^\dagger)^{-1}Lb^\dagger (R-bM^\dagger a^\dagger aMb^\dagger)\\
&=(LL^\dagger)^{-1}L(R-M^\dagger M)b^\dagger\\
&=(LL^\dagger)^{-1}LL^\dagger Lb^\dagger\\
&=Lb^\dagger.
\end{align*}
Taking $Q:=\mathrm{diag}(q,a^\dagger K)\in\mathrm{Sp}(2)$, we have that $(Q,K).\mathrm{diag}(a,b).(\hat{M},U)=(\hat{M},U)$. That is, $\hat{M}$ is $\mathrm{diag}(a,b)$-equivariant.
\end{proof}

\begin{note}
We see that $q$ is smooth in $K$ and the element in $\mathrm{Sp}(1)\times\mathrm{Sp}(1)$.
\end{note}

The following result allows us to narrow the search for instantons equivariant under isometries.
\begin{cor}
Suppose $\hat{M}\in\mathcal{M}_{n,k}$ is $\mathrm{diag}(a,b)$-equivariant. That is, there is some $(Q,K)\in\mathrm{Sp}(n+k)\times\mathrm{GL}(k,\mathbb{R})$ such that $(Q,K).\mathrm{diag}(a,b).(\hat{M},U)=(\hat{M}.U)$. Then $K\in\mathrm{O}(k)$, $[R,K]=0$, and $Q=\mathrm{diag}(Lb^\dagger KL^\dagger (LL^\dagger)^{-1},a^\dagger K)$.\label{cor:IsomS}
\end{cor}

\begin{proof}
Suppose that $(Q,K)\in\mathrm{Sp}(n+k)\times\mathrm{GL}(k,\mathbb{R})$ such that $(Q,K).\mathrm{diag}(a,b).(\hat{M},U)=(\hat{M},U)$. By Lemma~\ref{lemma:isomsimplify}, $K\in\mathrm{O}(k)$ and there is some $q\in\mathrm{Sp}(n)$ such that $Q=\mathrm{diag}(q,a^\dagger K)$.

We see that
\begin{equation*}
KRK^T=KL^\dagger q^\dagger qLK^T+KM^\dagger K^Taa^\dagger KMK^T.
\end{equation*}
Since $(Q,K).(\hat{M},U)=\mathrm{diag}(a^\dagger,b^\dagger).(\hat{M},U)$ and $R$ is real, we have
\begin{equation*}
KRK^T=bRb^\dagger=R.
\end{equation*}
Hence, $[K,R]=0$. Then, by Lemma~\ref{lemma:focusonY}, we have that $q=Lb^\dagger KL^\dagger (LL^\dagger)^{-1}$.
\end{proof}

In the subsequent sections, we investigate equivariant instantons. In doing so, we find that the continuous symmetries of instantons are generated by Lie algebra representations. We see that instantons generated by isomorphic representations are gauge equivalent, so we need only investigate a single representation from every isomorphism class. Moreover, it does not matter which representative we choose. 
\begin{lemma}
Let $\mathfrak{g}\subseteq \mathfrak{sp}(2)$ be a Lie subalgebra and $\hat{M}\in\mathcal{M}_{n,k}$. Let $\rho,\tilde{\rho}\colon\mathfrak{g}\rightarrow\mathfrak{sp}(n+k)$ and $\lambda,\tilde{\lambda}\colon\mathfrak{g}\rightarrow\mathfrak{so}(k)$ be pairs of isomorphic representations of $\mathfrak{g}$. That is, there is some $P\in\mathrm{Sp}(n+k)$ and $\Lambda\in\mathrm{O}(k)$ such that for all $x\in\mathfrak{g}$, $\tilde{\rho}(x)=P^\dagger \rho(x) P$ and $\tilde{\lambda}(x)=\Lambda^T\lambda(x) \Lambda$. 

Let $(\hat{M}',U'):=(P,\Lambda).(\hat{M},U)\in\mathcal{N}_{n,k}$. Then, for all $x\in\mathfrak{g}$, we have that
\begin{equation}
\left(e^{\tilde{\rho}(x)},e^{\tilde{\lambda}(x)}\right).e^x.(\hat{M},U)=(\hat{M},U),\label{eq:GeneratingProp1}
\end{equation}
if and only if\label{lemma:Generating}
\begin{equation}
\left(e^{{\rho}(x)},e^{{\lambda}(x)}\right).e^x.(\hat{M}',U')=(\hat{M}',U').\label{eq:GeneratingProp2}
\end{equation}
\end{lemma}

\begin{proof}
Consider \eqref{eq:GeneratingProp1}. Using the relationship between the representations, we see that 
\begin{equation*}
\left(P^\dagger e^{{\rho}(x)}P,\Lambda^T e^{{\lambda}(x)}\Lambda\right).e^x.(\hat{M},U)=(\hat{M},U).
\end{equation*}
As the actions commute, we see that upon simplification we obtain \eqref{eq:GeneratingProp2}. Thus, the equations are equivalent.
\end{proof}

\section{Circular symmetry}\label{subsec:CircularSymmetry}

In this section, we find an equation describing all instantons with circular $t$-symmetry. We then discuss the connections between such instantons and hyperbolic and singular monopoles. 

First, we introduce the notion of circular $t$-symmetry, as given in Table~\ref{table:conformalsubgroups}. For $t\in[0,1]$, let $R_t:=\{\mathrm{diag}(e^{i\theta},e^{ti\theta})\mid \theta\in\mathbb{R}\}$. An instanton is said to have circular $t$-symmetry if it is equivariant under every element of $R_t$. Note that for $t\in\mathbb{Q}\cap[0,1]$, $R_t\simeq S^1$ and for $t\in[0,1]\setminus\mathbb{Q}$, $R_t\simeq\mathbb{R}$. 

As $R_t$ is always one-dimensional, we can use Proposition~\ref{prop:R}.
\begin{theorem}
Let $\hat{M}\in\mathcal{M}_{n,k}$ and let $t\in[0,1]$. Then $\hat{M}$ has circular $t$-symmetry if and only if there exists $\rho\in\mathfrak{so}(k)$ such that
\label{thm:circularsym}
\begin{align}
tMi-iM+[\rho,M]&=0;
\label{eq:circularsymmetry}\\
[\rho,R]&=0.\label{eq:circularsymmetryextra}
\end{align}
\end{theorem}

\begin{definition}
The matrix $\rho$ is called the \textbf{generator} of the circular $t$-symmetry of $\hat{M}$.
\end{definition}

\begin{note}
All connected Lie subgroups of $\mathrm{Sp}(2)$ with Lie algebra $\mathbb{R}$ are conjugate to some $R_t$. That is, of the form $AR_tA^\dagger$, for some $A\in\mathrm{Sp}(2)$. Instantons equivariant under this group are of the form $A^\dagger.(\hat{M},U)$, with $\hat{M}$ equivariant under $R_t$.
\end{note}

\begin{proof}
Let $\mathcal{X}$ be the smooth manifold comprised of all pairs of $(n+k)\times k$ quaternionic matrices. The gauge and conformal actions can easily be expanded from $\mathcal{M}_{n,k}$ to smooth actions on $\mathcal{X}$. Just as with $\mathcal{M}_{n,k}$, the conformal action on $\mathcal{X}$ descends to an action on $\mathcal{X}/(\mathrm{Sp}(n+k)\times\mathrm{GL}(k,\mathbb{R}))$.

In particular, we focus on the action of $R_t$ on $\mathcal{X}$. Regardless of the value of $t$, $R_t$ is a connected, one-dimensional Lie group. As such, Proposition~\ref{prop:R} tells us that $\hat{M}$ has circular $t$-symmetry if and only if there is $\rho\in\mathfrak{gl}(k)$ and $\varsigma\in\mathfrak{sp}(n+k)$ such that for all $\theta\in\mathbb{R}$, we have
\begin{equation}
\theta\left(t\hat{M}i,Ui\right)-\theta\left(\varsigma\hat{M}-\hat{M}\rho,\varsigma U-U\rho\right)=(0,0).\label{eq:circularsymearly}
\end{equation}
Indeed, the two terms correspond to the actions of $\mathrm{diag}(i,ti)$ and $(\varsigma,\rho)$ on $(\hat{M},U)$, respectively.

In particular, from the proof of Proposition~\ref{prop:R}, we know that for all $\theta\in\mathbb{R}$, we have $\mathrm{diag}(e^{i\theta},e^{ti\theta}).(\hat{M},U)=(e^{\varsigma\theta},e^{\rho\theta}).(\hat{M},U)$~\cite[Proposition~1.2]{lang_moduli_2024}. By Corollary~\ref{cor:IsomS}, we know that for all $\theta\in\mathbb{R}$, $e^{\rho\theta}\in\mathrm{O}(k)$, $[R,e^{\rho\theta}]=0$, and $e^{\varsigma\theta}=\mathrm{diag}(Le^{-ti\theta}e^{\rho\theta}L^\dagger(LL^\dagger)^{-1},e^{-i\theta}e^{\rho\theta})$. 

From the above, we can differentiate and evaluate at $\theta=0$ to find $\rho\in\mathfrak{so}(k)$, $[R,\rho]=0$, and $\varsigma=\mathrm{diag}(L(-ti+\rho)L^\dagger(LL^\dagger)^{-1},\rho-i)$. Substituting into \eqref{eq:circularsymearly}, we find that $\hat{M}$ has circular $t$-symmetry if and only if there is $\rho\in\mathfrak{so}(k)$ such that for all $\theta\in\mathbb{R}$, we have
\begin{equation}
\theta \begin{bmatrix}
-tLiL^\dagger (LL^\dagger)^{-1}L+L\rho L^\dagger (LL^\dagger)^{-1}L-L\rho +tLi \\
tMi-iM+[\rho ,M]
\end{bmatrix}=0.\label{eq:circularsymconditions}
\end{equation}

Suppose $\hat{M}$ has circular $t$-symmetry. From above, $[R,\rho]=0$. Additionally, evaluating \eqref{eq:circularsymconditions} at $\theta=1$ and focusing on the bottom row, we see that we get \eqref{eq:circularsymmetry}. 

Conversely, similar to the proof of Proposition~\ref{prop:R}, suppose that \eqref{eq:circularsymmetry} and \eqref{eq:circularsymmetryextra} hold for some $\rho\in\mathfrak{so}(k)$~\cite[Proposition~1.2]{lang_moduli_2024}. Let $\mathrm{diag}(e^{i\theta_0},e^{ti\theta_0})\in R_t$. Let $K(\theta):=e^{\rho\theta}\in\mathrm{O}(k)$. We show that $A(\theta):=e^{-i\theta}K(\theta)MK(\theta)^Te^{ti\theta}$ is constant. Indeed,
\begin{equation*}
A'(\theta)=e^{-i\theta}K(\theta)(-iM+[\rho,M]+tMi)K(\theta)^Te^{ti\theta}=0.
\end{equation*}
Hence, $A$ is constant, so 
\begin{equation*}
A(\theta_0)=A(0)=M.
\end{equation*}
As $[\rho,R]=0$, $[K(\theta),R]=0$, so by Lemma~\ref{lemma:focusonY}, we have that $\hat{M}$ is $\mathrm{diag}(e^{i\theta_0},e^{ti\theta_0})$-equivariant. As $\theta_0$ was arbitrary, we have circular $t$-symmetry.
\end{proof}

As given in Table~\ref{table:conformalsubgroups}, an instanton is said to have toral symmetry if it is equivariant under $\mathrm{diag}(e^{i\phi_1},e^{i\phi_2})$ for all $\phi_1,\phi_2\in\mathbb{R}$. It turns out that we are only interested in circular $t$-symmetry when $t$ is rational; otherwise, it is equivalent to toral symmetry. 
\begin{prop}
Let $\hat{M}\in\mathcal{M}_{n,k}$ and $t\in[0,1]\setminus\mathbb{Q}$. Then $\hat{M}$ has circular $t$-symmetry if and only if it has toral symmetry.\label{prop:irrationaltorus}
\end{prop}

\begin{proof}
Let $\hat{M}\in\mathcal{M}_{n,k}$ and $t\in[0,1]\setminus\mathbb{Q}$. If $\hat{M}$ has toral symmetry, then we see that it has circular $t$-symmetry. Conversely, suppose that $\hat{M}$ has circular $t$-symmetry. 

Let $H:=\{\mathrm{diag}(e^{i\phi_1},e^{i\phi_2})\mid \mathrm{diag}(e^{i\phi_1},e^{i\phi_2}).[(\hat{M},U)]=[(\hat{M},U)]\}$. That is, $H$ is the group of symmetries of $\hat{M}$, restricted to toral symmetries. As $\hat{M}$ has circular $t$-symmetry, $R_t\subseteq H$. 

In previous work, I prove that if the gauge group acts on a smooth manifold properly, then the group of symmetries is closed~\cite[Proposition~2.4]{lang_moduli_2024}. In particular, this result holds if the gauge group is compact. Here, we do not have a compact gauge group; however, we obtain the same result, as we can reduce our gauge group to a compact one, as per Lemma~\ref{lemma:isomsimplify}. 

Let $S\subseteq S^1\times S^1\times\mathrm{O}(k)$ be the stabilizer group of $(\hat{M},U)$ restricted to toral rotations. That is
\begin{multline*}
S:=\{(\mathrm{diag}(e^{i\theta_1},e^{i\theta_2}),K)\mid (\mathrm{diag}(Le^{-i\theta_2}KL^\dagger (LL^\dagger)^{-1},e^{-i\theta_1} K),K).\\
\mathrm{diag}(e^{i\theta_1},e^{i\theta_2}).(\hat{M},U)=(\hat{M},U)\}.
\end{multline*}
Indeed, if $(Q,K)\in\mathrm{Sp}(n+k)\times\mathrm{GL}(k,\mathbb{R})$ such that $\mathrm{diag}(e^{i\theta_1},e^{i\theta_2}).(Q,K).(\hat{M},U)=(\hat{M},U)$, then Corollary~\ref{cor:IsomS} tells us that $K\in\mathrm{O}(k)$, $[R,K]=0$, and $Q$ is given by the diagonal quaternionic matrix above. Therefore, the pairs in $S$ encapsulate all the toral symmetry of the instanton. That $S$ is a group follows from this fact. For this proof, only the fact that all toral symmetry is captured by $S$ matters. We examine this stabilizer group more in the proof of Theorem~\ref{thm:toralsym}.

Following the proof of Theorem~\ref{thm:mainthm}, we have that $S$ is closed. Indeed, it is the preimage of $(\hat{M},U)$ under the smooth map $f\colon S^1\times S^1\times \mathrm{O}(k)\rightarrow \mathrm{Mat}(n+k,k,\mathbb{H})^{\oplus 2}$ given by 
\begin{equation*}
f(\mathrm{diag}(e^{i\theta_1},e^{i\theta_2}),K):=(\mathrm{diag}(Le^{-i\theta_2}KL^\dagger (LL^\dagger)^{-1},e^{-i\theta_1} K),K).\\
\mathrm{diag}(e^{i\theta_1},e^{i\theta_2}).(\hat{M},U).
\end{equation*} 

We see that $H=\pi_1(S)$, where $\pi_1\colon S^1\times S^1\times\mathrm{O}(k)\rightarrow S^1\times S^1$. As $\mathrm{O}(k)$ is compact, $\pi_1$ is a closed map, so $H=\pi_1(S)$ is closed.

Note that as $t$ is irrational, $R_t$ is dense in $S^1\times S^1$. Then, as $R_t\subseteq H$, we have that 
\begin{equation*}
S^1\times S^1=\overline{R_t}\subseteq \overline{H}=H\subseteq S^1\times S^1.
\end{equation*}
Therefore, $\hat{M}$ has toral symmetry.
\end{proof}

When searching for instantons with circular $t$-symmetry, we need not check the final condition of $\mathcal{M}_{n,k}$ in Definition~\ref{def:Mstd} everywhere, as Lemma~\ref{lemma:finalconditioncirc} establishes. Note the similarity to previous work dealing with hyperbolic monopoles~\cite[Lemma~5]{lang_hyperbolic_2023}.
\begin{lemma}
Suppose that $\hat{M}$ satisfies \eqref{eq:circularsymmetry} and \eqref{eq:circularsymmetryextra} for some $t\in[0,1]$ and $\rho\in\mathfrak{so}(k)$ as well as the first three conditions of Definition~\ref{def:Mstd}. If the final condition is satisfied at all $x=x_0+x_1i+x_2j\in\mathbb{H}$ with $x_2\geq 0$, then $\hat{M}\in\mathcal{M}_{n,k}$. \label{lemma:finalconditioncirc}
\end{lemma}

\begin{proof}
Let $z=z_0+z_1i+z_2j+z_3k\in \mathbb{H}$. Note that $\mathrm{diag}(e^{i\theta},e^{ti\theta})$ acts on $z$ as $z\mapsto e^{i\theta}ze^{-ti\theta}$. Specifically, it takes $z\mapsto (z_0+z_1i)e^{(1-t)i\theta}+(z_2+z_3i)e^{(1+t)i\theta}j$. Thus, there is some $\theta\in\mathbb{R}$ such that $x:=e^{i\theta}ze^{-ti\theta}$ has no $k$ component and the $j$ component is non-negative.

As $\hat{M}$ has circular $t$-symmetry, there is a pair $(Q,K)\in\mathrm{Sp}(n+k)\times \mathrm{O}(k)$ such that $(Q,K).\mathrm{diag}(e^{i\theta},e^{ti\theta}).(\hat{M},U)=(\hat{M},U)$. Recall that $K$ is necessarily an orthogonal matrix by Corollary~\ref{cor:IsomS}. Thus,
\begin{equation*}
\Delta(x)=(\hat{M}e^{ti\theta}-Ue^{i\theta}z)e^{-ti\theta}=Q^\dagger \Delta(z)Ke^{-ti\theta}.
\end{equation*}
Hence, $\Delta(x)^\dagger \Delta(x)=e^{ti\theta}K^T\Delta(z)^\dagger \Delta(z) Ke^{-ti\theta}$. As $\Delta(x)^\dagger \Delta(x)$ is real, we have
\begin{equation*}
\Delta(z)^\dagger \Delta(z) = K\Delta(x)^\dagger \Delta(x) K^T.
\end{equation*} 
As $\Delta(x)^\dagger\Delta(x)$ is non-singular, it follows that so too is $\Delta(z)^\dagger\Delta(z)$. 
\end{proof}

\subsection{Structure of circular symmetry}

In light of Theorem~\ref{thm:circularsym}, we know that we can find instantons with circular $t$-symmetry given a matrix $\rho\in\mathfrak{so}(k)$. Such a matrix generates a real representation of $\mathbb{R}$. Starting with such a real representation, we can narrow down the possible $M$, so we are only left with finding $L$ such that we have an instanton, which necessarily has circular $t$-symmetry. 

By proving that the group action of $S^1$ on $\mathbb{H}$ given by $e^{i\theta}.x:=xe^{i\theta}$, which is equivalent to the action of $R_0$ on $\mathbb{H}$, lifts to an action on the spinor bundle, Beckett proved that for instantons with circular $0$-symmetry, their symmetry is generated by representations of $S^1$~\cite[Proposition~4.2.1]{beckett_equivariant_2020}. Using Lie theory, we prove this result for circular $t$-symmetry, for all $t\in\mathbb{Q}\cap [0,1]$. That is, for all circular $t$-symmetry not equivalent to toral symmetry.

Note that a representation $\rho\colon\mathbb{R}\rightarrow\mathfrak{so}(k)$ corresponds to a representation of $S^1$ if and only if $\mathrm{exp}(\rho(2\pi))=I_k$. 
\begin{prop}
Let $t=\frac{a}{b}\in\mathbb{Q}\cap [0,1]$ and let $\hat{M}\in\mathcal{M}_{n,k}$. Then $\hat{M}$ has circular $t$-symmetry if and only if there exists a real representation $\rho\colon\mathbb{R}\rightarrow\mathfrak{so}(k)$, corresponding to a representation of $S^1$, such that for all $\theta\in\mathbb{R}$,
\label{prop:circulartsymrep}
\begin{align}
a\theta Mi-b\theta iM+\frac{1}{2}[\rho(\theta),M]&=0;
\label{eq:circulartsymmetryrepversion}\\
[\rho(\theta),R]&=0.\label{eq:circulartsymmetryextrarepversion}
\end{align}
\end{prop}

\begin{definition}
We call $\rho$ the \textbf{generating representation} of the circular $t$-symmetry of $\hat{M}$.
\end{definition}

\begin{proof}
If such a $\rho$ exists, by Theorem~\ref{thm:circularsym}, taking $\tilde{\rho}:=\frac{1}{2b}\rho(1)$, we have circular $t$-symmetry. Conversely, suppose that $\hat{M}$ has circular $t$-symmetry. By Theorem~\ref{thm:circularsym}, there exists $\rho\in\mathfrak{so}(k)$ satisfying \eqref{eq:circularsymmetry} and \eqref{eq:circularsymmetryextra}.

As $\rho\in\mathfrak{so}(k)$, there exists $U\in\mathrm{O}(k)$ such that 
\begin{equation*}
\rho=U\mathrm{diag}(\rho_1,\ldots,\rho_m)U^T,
\end{equation*}
with each $\rho_l$ being either zero or $\begin{bmatrix}
0 & a_l \\ -a_l & 0
\end{bmatrix}$ for some $a_l\in\mathbb{R}$. We can write $M$ according to the decomposition of $\rho$ as 
\begin{equation*}
M=U\begin{bmatrix}
M_{11} & \cdots & M_{1m} \\
\vdots & \ddots & \vdots \\
M_{1m}^T & \cdots & M_{mm}
\end{bmatrix}U^T.
\end{equation*}

Consider \eqref{eq:circularsymmetry}. We can multiply on the right by $U$ and the left by $U^T$ to get
\begin{equation*}
\left[\mathrm{diag}(\rho_1,\ldots,\rho_m),\begin{bmatrix}
M_{11} & \cdots & M_{1m} \\
\vdots & \ddots & \vdots \\
M_{1m}^T & \cdots & M_{mm}
\end{bmatrix}\right]=\begin{bmatrix}
iM_{11}-tM_{11}i & \cdots & iM_{1m}-tM_{1m}i \\
\vdots & \ddots & \vdots \\
iM_{1m}^T-tM_{1m}^Ti & \cdots & iM_{mm}-tM_{mm}i
\end{bmatrix}.
\end{equation*}

The on-diagonal components, which are symmetric, must satisfy $[\rho_l,M_{ll}]=iM_{ll}-tM_{ll}i$. If $t\neq 1$ and $|2a_l\pm t|\neq 1$, then $M_{ll}=0$. If $t=1$ and if $a_l\notin \{-1,0,1\}$, then $M_{ll}=\alpha_l I_2$ for some $\alpha_l\in\mathbb{C}$. However, note that such a $M_{ll}$ satisfies the required equation for any value of $a_l$. The off-diagonal components must satisfy $\rho_{l}M_{lp}-M_{lp}\rho_p=iM_{lp}-tM_{lp}i$. If $|a_l\pm a_p|\neq 1\pm t$, then $M_{lp}=0$. It is important to note that the off-diagonal $M$ components do not depend on the exact values of the $a_l$ and $a_p$, but rather the relationships between them. 

Using the above relationships, we proceed to construct a new matrix $\rho'$ of the same form as $\rho$. Specifically, we will find $a_l'\in\mathbb{R}$ and define $\rho_l'$ to have the same form as the corresponding $\rho_l$, swapping $a_l$ with $a_l'$. These $a_l'$ will be chosen such that $\rho':=U\mathrm{diag}(\rho_1',\ldots,\rho_m')U^T$ still satisfies \eqref{eq:circularsymmetry} and \eqref{eq:circularsymmetryextra}, but a multiple of it generates a representation of $\mathbb{R}$ corresponding to a representation of $S^1$.

To begin constructing $\rho'$, define the relation $\sim_0$ on $\{1,\ldots,m\}$ by $l\sim_0 p$ if and only if $M_{lp}\neq 0$. This relation generates an equivalence relation, which we denote by $\sim$. Consider $[l]\in\{1,\ldots,m\}/\sim$. For such a class, if $M_{ll}=0$, then we are free to set $a_l'$ to anything. In particular, we may choose $a_l'=0$. If $t\neq 1$ and $M_{ll}\neq 0$, then from above, we know that $|2a_l\pm t|=1$, so we define $a_l':=a_l$. If $t=1$ and $M_{ll}$ is proportional to the identity, then we may choose $a_l'=0$. In any case, we have that $2ba_l'$ is an integer. 

Regardless of the value of $t$ or if $M_{ll}$ vanishes, for $p\in[l]$, the relationship between $a_p$ and $a_l$ is determined by the form of $M_{lp}$ if $M_{lp}\neq 0$. From above, we see that if $M_{lp}\neq 0$, then $|a_l\pm a_p|=1\pm t$. As $2ba_l'\in\mathbb{Z}$, we have that if $a_p'$ and $a_l'$ satisfy the same relationship as $a_p$ and $a_l$, then $a_p'\in\mathbb{Q}$. Specifically, we have $2ba_p'\in\mathbb{Z}$. Finally, as $p\in[l]$, if $M_{lp}=0$, there is some $q\in[l]$ such that $M_{lq}\neq 0$. Then there is some string of elements in $[l]$ connecting $p$ and $l$ by non-zero off-diagonal components. From above, we see that for each element in this string, $2ba_q'\in\mathbb{Z}$ and the same is true for $p$ in particular. Therefore, $2ba_p'\in\mathbb{Z}$ for all $p\in[l]$, including $l$. Note that by construction, $\rho'$ satisfies \eqref{eq:circularsymmetry}. That $\rho'$ satisfies \eqref{eq:circularsymmetryextra} follows from Corollary~\ref{cor:IsomS}.

Doing the above for each equivalence class, we define $\tilde{\rho}\colon\mathbb{R}\rightarrow \mathfrak{so}(k)$ by $\tilde{\rho}(\theta):=2b\theta \rho'$. From the work above, we note that $\mathrm{exp}(\tilde{\rho}(2\pi))=I_k$. Furthermore, by construction, $\tilde{\rho}$ satisfies \eqref{eq:circulartsymmetryrepversion} and \eqref{eq:circulartsymmetryextrarepversion}.
\end{proof}

Proposition~\ref{prop:circulartsymrep} tells us that we can use representation theory for $S^1$ to better understand instantons with circular $t$-symmetry, for $t\in\mathbb{Q}$. The generating representation was obtained by examining the bottom of \eqref{eq:circularsymconditions}. By focusing instead on the top of that equation, we obtain an induced representation.
\begin{lemma}
Suppose that $\hat{M}\in\mathcal{M}_{n,k}$ has circular $t$-symmetry generated by $\rho\in\mathfrak{so}(k)$. That is, $\rho$ satisfies \eqref{eq:circularsymmetry} and \eqref{eq:circularsymmetryextra}. Then $y:=L(\rho-ti)L^\dagger (LL^\dagger)^{-1}\in\mathfrak{sp}(n)$. 
\end{lemma}

\begin{proof}
By Theorem~\ref{thm:circularsym}, we know that $tMi-iM+[\rho,M]=0$. Using this equation, we see
\begin{equation*}
\begin{aligned}
M^\dagger M(\rho-ti)&=M^\dagger[M,\rho]+M^\dagger \rho M-tM^\dagger Mi\\
&=M^\dagger \rho M-M^\dagger iM\\
&=[M^\dagger,\rho]M+\rho M^\dagger M-M^\dagger iM\\
&=(\rho-ti)M^\dagger M.
\end{aligned}
\end{equation*}
Hence, $[M^\dagger M,\rho-ti]=0$. Then, as $[\rho,R]=0$, by Theorem~\ref{thm:circularsym}, and $R$ is real, $[\rho-ti,R]=0$, so $[L^\dagger L,\rho-ti]=0$. Thus, we see that $[LL^\dagger ,L(\rho-ti)L^\dagger]=0$. Therefore, $y\in\mathfrak{sp}(n)$. 
\end{proof}

However, just as with $\rho$, as we are free to focus on rational values of $t$, we can improve on this result. 
\begin{prop}
Suppose that $\hat{M}\in\mathcal{M}_{n,k}$ has circular $t$-symmetry, with $t=\frac{a}{b}\in\mathbb{Q}\cap[0,1]$, generated by a Lie algebra representation $(\mathbb{R}^k,\rho)$, corresponding to a representation of $S^1$. That is, $\rho$ satisfies \eqref{eq:circulartsymmetryrepversion} and \eqref{eq:circulartsymmetryextrarepversion}. Let $\lambda(\theta):=L(\rho(\theta)-2ai\theta)L^\dagger(LL^\dagger)^{-1}\in\mathfrak{sp}(n)$. We have $(\mathbb{H}^n,\lambda)$ is a quaternionic representation of $\mathbb{R}$ corresponding to a representation of $S^1$.
\end{prop}

\begin{proof}
From the work above, as $\lambda(\theta)=2b\theta y$, we have $\lambda(\theta)\in\mathfrak{sp}(n)$. It remains to show that $e^{\lambda(2\pi)}=I_n$. Indeed, due to the commutativity of $LL^\dagger$ and $L(\rho(\theta)-2a\theta i)L^\dagger$, we have that 
\begin{equation*}
e^{\lambda(\theta)}=Le^{\rho(\theta)-2a\theta i}L^\dagger (LL^\dagger)^{-1}.
\end{equation*}
Taking $\theta=2\pi$, we see that as $\rho$ and $i$ commute,
\begin{equation*}
e^{\lambda(2\pi)}=Le^{\rho(2\pi)}e^{-4a\pi i}L^\dagger(LL^\dagger)^{-1}=I_n.
\end{equation*}
Therefore, $(\mathbb{H}^n,\lambda)$ corresponds to a representation of $S^1$.
\end{proof}

We know that this induced representation satisfies the following equation.
\begin{cor}
Suppose that $\hat{M}\in\mathcal{M}_{n,k}$ has circular $t$-symmetry, with $t=\frac{a}{b}\in\mathbb{Q}\cap[0,1]$, generated by $(\mathbb{R}^k,\rho)$, which corresponds to a representation of $S^1$. Then $\lambda(\theta) L-L\rho(\theta)+2a\theta Li=0$.
\end{cor}

\begin{proof}
This result follows from multiplying the top of \eqref{eq:circularsymconditions} by $2b$.
\end{proof}

\subsection{Hyperbolic monopoles}\label{subsubsec:HyperbolicMonopoles}

In this section, we discuss the connection between hyperbolic monopoles and circle-invariant instantons. In particular, we discuss how hyperbolic monopoles with integral weight correspond to circular $1$-symmetric ADHM data and how to transform this ADHM data into a hyperbolic monopole. This section is an extension of my previous work on hyperbolic monopoles to all hyperbolic monopoles with integral weight~\cite{lang_hyperbolic_2023}.

Let $E\rightarrow H^3$ be a vector bundle over $H^3$ with structure group $\mathrm{Sp}(n)$, equipped with a connection $A$, of curvature $F_A$, and a section $\Phi$ of $\mathrm{End}(E)$, called the Higgs field. Let $\star$ denote the Hodge star. Hyperbolic monopoles are solutions to the Bogomolny equations $F_A=\star D_A\Phi$, with finite energy $\varepsilon=\frac{1}{2\pi}\int_{H^3} |F_A|^2\mathrm{vol}_{H^3}$.

It is well known that the Bogomolny equations are a dimensional reduction of the self-dual equations. That is, monopoles are related to instantons that are independent of one coordinate. When this coordinate is one of the standard Cartesian coordinates, the instanton becomes translation invariant and either has no action (and is therefore flat and of no interest) or has infinite action. Atiyah realized that if this coordinate corresponded to that of a circle, then the compactness of the circle implies that the instanton has finite action. Then, by utilizing the conformal equivalence $S^4\setminus S^2\equiv H^3\times S^1$, he deduced that circle-invariant instantons correspond to hyperbolic monopoles~\cite{atiyah_instantons_1984,atiyah_magnetic_1984}. While Atiyah used a particular circle action, we see below that this conclusion is not true for every circle action.

Atiyah's idea is that given a circle action on $S^4$, remove some $S^2\subseteq S^4$ invariant under the action and use coordinates $(x,y,z,\theta)$ to describe $S^4\setminus S^2$. The coordinates $(x,y,z)$ correspond to coordinates on $H^3$ and $\theta$ corresponds to the coordinate on $S^1$. Specifically, $\theta$ is generated by the circle action. Then, an instanton that is equivariant under the circle action is independent of the $\theta$ coordinate and can be written as a hyperbolic monopole. Conversely, given a hyperbolic monopole, we can construct a circle-invariant instanton on $S^4\setminus S^2$. In order to extend the obtained instanton to all of $S^4$, the monopole must have integral mass~\cite{atiyah_magnetic_1984}. 

It turns out that only one conjugacy class of circle actions gives rise to hyperbolic monopoles.
\begin{lemma}
A circle action relates hyperbolic monopoles with integral mass and instantons invariant under this action if and only if the circle action is conjugate to $R_1$.\label{lemma:hypermonoR1}
\end{lemma}

\begin{note}
Both Braam--Austin and Manton--Sutcliffe's conformal circle actions are conjugate to $R_1$. Thus, they are equivalent. However, they lead to different models of hyperbolic space, allowing us to view monopoles differently.
\end{note}

\begin{proof}
The isometry group of $H^3$ is six-dimensional. Thus, adding rotation on $S^1$, we have a group with dimension at least seven acting on $H^3\times S^1$ as isometries. As $S^4\setminus S^2\equiv H^3\times S^1$, these isometries correspond to conformal maps on $S^4\setminus S^2$, which we can complete, obtaining conformal maps on $S^4$. 

Suppose that we have a circle action relating hyperbolic monopoles with integral mass and instantons invariant under this action. By Conjecture~\ref{conj}, we may assume that the circle action is conjugate to an action in $\mathrm{Sp}(2)$. This action is generated by an element in $\mathfrak{sp}(2)$. As this Lie algebra is compact, the adjoint orbit of an element of the Lie algebra has a non-empty intersection with any Cartan subalgebra. In particular, this fact means that up to conjugation, the circle action is generated by an element in $\{\mathrm{diag}(ai,bi)\mid a,b\in\mathbb{R}\}$. In particular, up to conjugation, we may assume that the circle action is generated by $\mathrm{diag}(i,ti)\in\mathfrak{sp}(2)$ for some $t\in[0,1]$. 

If a circle action generates this conformal equivalence, then there is a seven-dimensional Lie subalgebra $\mathfrak{l}$ of $\mathrm{Lie}(\mathrm{SL}(2,\mathbb{H}))$ that commutes with $\mathrm{diag}(i,ti)$. This is because the circle action corresponds to rotating $S^1$ in $H^3\times S^1$, thus commutes with isometries of $H^3$. Note that $\mathrm{Lie}(\mathrm{SL}(2,\mathbb{H}))$ is the subset of $\mathrm{Mat}(2,2,\mathbb{H})$ with purely imaginary trace. When $t=0$ or $t\in(0,1)$, the subalgebra of elements that commute with $\mathrm{diag}(i,ti)$ is five-dimensional or three-dimensional, respectively. In either case, the space is not large enough to contain $\mathfrak{l}$. However, when $t=1$, this space is seven-dimensional, as required. 

Finally, conversely, we know that $R_1$ is conjugate to Manton--Sutcliffe and Braam--Austin's circle actions, which relate hyperbolic monopoles with integral mass and instantons invariant under these actions. 
\end{proof}

\begin{note}
That only one conjugacy class corresponds to hyperbolic monopoles is to be expected. All hyperbolic monopoles with integral mass arise from this construction. However, the circle action used to generate $S^4\setminus S^2\equiv H^3\times S^1$ should always generate all hyperbolic monopoles with integral mass. As it stands, only circle actions conjugate to $R_1$ generate hyperbolic monopoles. Being related by conjugation, the hyperbolic monopoles created by each action are all related. If multiple conjugacy classes generated hyperbolic monopoles, then there would be no guarantee that the monopoles would be related.
\end{note}

In later sections, we investigate symmetric hyperbolic monopoles more deeply. However, we see that such monopoles must all arise from instantons equivariant under $R_1$, or a group conjugate to it.

The model of hyperbolic 3-space obtained from the action of $R_1$ is the half-space model. It is obtained as follows. The fixed two-sphere of the $R_1$ action is $\mathbb{C}\cup\{\infty\}$. On $S^4\setminus(\mathbb{C}\cup\{\infty\})$, equivalently $\mathbb{H}\setminus\mathbb{C}$, we use coordinates $(x_0,x_1,r,\theta)$, with $r>0$, which corresponds to $x_0+x_1i+r\cos\theta j+r\sin\theta k$. Note that the action of $e^{i\phi/2}$ sends $(x_0,x_1,r,\theta)$ to $(x_0,x_1,r,\theta+\phi)$. Thus, $(x_0,x_1,r)$ are coordinates for $H^3$ and $\theta$ is the coordinate for $S^1$. 

In these new coordinates, the metric for $\mathbb{H}\setminus\mathbb{C}$ is $dx_0^2+dx_1^2+dr^2+r^2d\theta^2$. As $r>0$, we can divide by $r^2$ to get the metric
\begin{equation*}
\frac{dx_0^2+dx_1^2+dr^2}{r^2}+d\theta^2.
\end{equation*}
This metric is exactly the metric for $H^3\times S^1$, using the half-space model.

The ADHM transform relates instantons and ADHM data. Recall that given ADHM data $(a,b)\in\mathcal{M}$, for all $x\in\mathbb{H}$, we define $\Delta(x):=a-bx$. We then find $V(x)\in\mathrm{Mat}(n+k,n,\mathbb{H})$ such that $V(x)^\dagger\Delta(x)=0$ and $V(x)^\dagger V(x)=I_n$. The instanton is obtained via $\mathbb{A}_\mu(x):=V(x)^\dagger \partial_\mu V(x)$. By changing coordinates to $X$ and $\theta$, where $X$ is the coordinate in $H^3$ and $\theta$ is the $S^1$ coordinate, as the instanton has circular $1$-symmetry, there is some gauge in which the instanton is independent of $\theta$. Thus, the instanton is a function of $X$. In these coordinates, $\mathbb{A}=\mathbb{A}_0 dx_0+\mathbb{A}_1dx_1+\mathbb{A}_rdr+\mathbb{A}_\theta d\theta$. Letting $\Phi:=\mathbb{A}_\theta$ and $A$ be the remainder of $\mathbb{A}$, the pair $(\Phi,A)$ is the corresponding hyperbolic monopole. 

This process is outlined in the case of Manton--Sutcliffe's conformal circle action~\cite[\S4]{manton_platonic_2014}. We now complete this process for our $R_1$ circle action. For other $R_1$ circle actions, similar results can be found.
\begin{prop}
Suppose that $\hat{M}\in\mathcal{M}_{n,k}$ has circular $1$-symmetry, with respect to usual $R_1$ group, generated by $\rho\in\mathfrak{so}(k)$. Let $X$ denote the coordinates of $H^3$. Suppose that we complete the ADHM procedure for $H^3$, finding $V(X)$ satisfying $V(X)^\dagger \Delta(X)=0$ and $V(X)^\dagger V(X)=I_n$. The corresponding hyperbolic monopole $(\Phi,A)$ is given, up to gauge, by 
\begin{align}
\Phi(X)&=\frac{1}{2}V(X)^\dagger \mathrm{diag}\bigl(L(i-\rho)L^\dagger (LL^\dagger)^{-1},i-\rho\bigr) V(X);\\
A_l(X)&=V(X)^\dagger \partial_lV(X).
\end{align}
The $\partial_l$ denotes partial differentiation with respect to the $l$th coordinate of $H^3$.\label{prop:makehypermono}
\end{prop}

\begin{proof}
As $\mathbb{H}\setminus\mathbb{C}\simeq H^3\times S^1$, any $x\in\mathbb{H}\setminus\mathbb{C}$ can be uniquely written as 
\begin{equation*}
x=\mathrm{diag}(e^{i\theta/2},e^{i\theta/2}).X=e^{i\theta/2}(x_0+ix_1+rj)e^{-i\theta/2}.
\end{equation*} 
Note the factor of two means that given $x$, we get a unique $\theta\in[0,2\pi]$ and $X\in H^3$.

Let $Q(\theta):=\mathrm{diag}(Le^{-i\theta/2}e^{\rho\theta/2}L^\dagger (LL^\dagger)^{-1},e^{-i\theta/2}e^{\rho\theta/2})\in\mathrm{Sp}(n+k)$ and $K(\theta):=e^{\rho\theta/2}\in\mathrm{SO}(k)$. Note that by the proof of Lemma~\ref{lemma:focusonY}, we have that $[LL^\dagger,Le^{-i\theta/2}e^{\rho\theta/2}L^\dagger]$. Then, per the proof of Theorem~\ref{thm:circularsym}, we have
\begin{equation}
(Q(\theta),K(\theta)).\mathrm{diag}(e^{i\theta/2},e^{i\theta/2}).(\hat{M},U)=(\hat{M},U).\label{eq:circularsymSagain}
\end{equation}
Therefore,
\begin{align*}
\Delta(x)&=\Delta(e^{i\theta/2}(x_0+ix_1+rj)e^{-i\theta/2})\\
&=\hat{M}-Ue^{i\theta/2}(x_0+ix_1+rj)e^{-i\theta/2}\\
&=(\hat{M}e^{i\theta/2}-Ue^{i\theta/2} X)e^{-i\theta/2}\\
&=Q(\theta)^\dagger \Delta(X) K(\theta)e^{-i\theta/2}.
\end{align*}

Extending $V$ to all of $\mathbb{H}\setminus\mathbb{C}$ by $V(x):=Q(\theta)^\dagger V(X)$, we note that $V(x)^\dagger \Delta(x)=0$ and $V(x)^\dagger V(x)=I_n$. Thus, we can use it to create our instanton and thus our monopole.

The $\theta$ component of the instanton $\mathbb{A}$ corresponding to $(\hat{M},U)$ is given by
\begin{equation*}
\mathbb{A}_\theta(x)=V(x)^\dagger\partial_\theta V(x)= V(X)^\dagger Q(\theta)\partial_\theta Q(\theta)^\dagger V(X).
\end{equation*}
Simplifying, we find that
\begin{equation*}
\mathbb{A}_\theta(x)=\frac{1}{2}V(X)^\dagger \begin{bmatrix}
(LL^\dagger)^{-1}Le^{-i\theta/2}e^{\rho\theta/2} L^\dagger L e^{-\rho\theta/2}e^{i\theta/2}(i-\rho)L^\dagger (LL^\dagger)^{-1} & 0 \\
0 & i-\rho
\end{bmatrix}V(X).
\end{equation*}
Recall that $R:=L^\dagger L+M^\dagger M$ commutes with $\rho$, by Theorem~\ref{thm:circularsym}, and thus $e^{\rho\theta}$. Furthermore, as $R$ is real, it commutes with quaternions. By \eqref{eq:circularsymSagain}, we have that 
\begin{equation*}
e^{-i\theta/2}e^{\rho\theta/2}Me^{i\theta/2}e^{-\rho\theta/2}=M.
\end{equation*} 
From this equation, we see that
\begin{align*}
e^{-i\theta/2}e^{\rho\theta/2}L^\dagger Le^{-\rho\theta/2}e^{i\theta/2}&=e^{-i\theta/2}e^{\rho\theta/2}(R-M^\dagger M)e^{-\rho\theta/2}e^{i\theta/2}\\
&=R-M^\dagger e^{\rho\theta/2}e^{-i\theta/2}M e^{-\rho\theta/2}e^{i\theta/2}=R-M^\dagger M=L^\dagger L.
\end{align*}
Returning to the computation of $\mathbb{A}_\theta(x)$, we have
\begin{equation*}
\mathbb{A}_\theta(x)=\frac{1}{2}V(X)^\dagger \mathrm{diag}(L(i-\rho)L^\dagger(LL^\dagger)^{-1},i-\rho)V(X).
\end{equation*}
As in the procedure described above, we take $\Phi(X):=\mathbb{A}_\theta(X)$. Note that $\mathbb{A}_\theta$ is independent of $\theta$.

We have that the other coordinates of $\mathbb{A}$ are given by
\begin{equation*}
\mathbb{A}_l(x)=V(x)^\dagger \partial_l V(x)=V(X)^\dagger \partial_l V(X).
\end{equation*}
Similar to $\Phi$, we take $A_l(X):=\mathbb{A}_l(X)$ and note that in this gauge, $\mathbb{A}$ is independent of $\theta$. Thus, our monopole has the desired form.
\end{proof}

Suppose that $(\Phi,A)$ is a hyperbolic monopole with integral mass. Both $\Phi$ and $A$ are defined on $H^3$. The connection between hyperbolic monopoles with integral mass and instantons with circular $1$-symmetry grants such hyperbolic monopoles an additional symmetry. Such monopoles correspond to an instanton $\mathbb{A}$ with circular $1$-symmetry via $\mathbb{A}=\Phi d\theta +A$. The integral mass condition implies that the instanton $\mathbb{A}$ is defined on the whole of $S^4$. Thus, $\Phi$ and $A$ are also defined on the whole of $S^4$. Moreover, given an embedding of $H^3$ into $S^4$, the ball model for instance, the definitions on $S^4$ agree with the original definition on $H^3$. 
\begin{prop}
Consider the ball model $H^3=\{X=X_1i+X_2j+X_3k\in\mathfrak{sp}(1)\mid R^2:=X_1^2+X_2^2+X_3^2<1\}$ with curvature $-1$. Suppose we have a hyperbolic monopole with integral mass and Higgs field $\Phi$. For all $X\in H^3\setminus \{0\}$, the gauge-invariant quantity $|\Phi|$ satisfies\label{prop:palindromic}
\begin{equation}
|\Phi(X)|=\left|\Phi\left(\frac{X}{R^2}\right)\right|.
\end{equation}
\end{prop}

\begin{proof}
Such hyperbolic monopoles correspond to instantons with circular $1$-symmetry. In particular, as we are using the ball model, these monopoles correspond directly with instantons symmetric under Manton--Sutcliffe's conformal circle action~\cite[(4.4)]{manton_platonic_2014}. Given a point $0\neq X\in H^3\subseteq \mathbb{H}$, the orbit $S^1\cdot X$ under this circle action is given by 
\begin{equation*}
S^1\cdot X=\left\{\begin{bmatrix}
\cos\theta & \sin\theta \\ -\sin\theta & \cos\theta 
\end{bmatrix}.X\mid \theta\in\mathbb{R}\right\}.
\end{equation*} 
In particular, when $\theta=\frac{\pi}{2}$, we see that this orbit contains the point $X/R^2$.

We know that there is some gauge in which the instanton $\mathbb{A}:=\Phi d\theta +A$ is independent of $\theta$. Then, as $X$ and $X/R^2$ are in the same orbit, we have that $\mathbb{A}(X)=\mathbb{A}\left(\frac{X}{R^2}\right)$. Recall from above that as we have a hyperbolic monopole with integral mass, $\Phi$ is defined on the whole of $S^4$. Matching $d\theta$ terms, we have that $\Phi(X)=\Phi\left(\frac{X}{R^2}\right)$. This equality is gauge dependent. However, we have that the gauge-invariant quantities $|\Phi(X)|$ and $\left|\Phi\left(\frac{X}{R^2}\right)\right|$ are equal in any gauge. 
\end{proof}

The additional symmetry of hyperbolic monopoles with integral mass proven in Proposition~\ref{prop:palindromic} is only possible due to the connection between hyperbolic monopoles with integral mass and instantons with circular $1$-symmetry.

\begin{note}
In previous work, we examine the norm squared of the Higgs field of a $\mathrm{Sp}(2)$ and $\mathrm{Sp}(4)$ hyperbolic monopole, respectively~\cite[Notes~14~\&~15]{lang_hyperbolic_2023}. Both functions are rational, and are comprised of palindromic polynomials. Additionally, for each rational function, the numerator and denominator have the same degree. The palindromic nature of these polynomials follows from Proposition~\ref{prop:palindromic} and the matching degrees of the numerators and denominators. The respective energy densities of these monopoles are also rational functions of palindromic polynomials. Both of these hyperbolic monopoles have integral mass. The same behaviour is not guaranteed for hyperbolic monopoles with non-integral mass.

Indeed, consider the family of spherically symmetric $\mathrm{Sp}(1)$ hyperbolic monopoles given by Chakrabarti and Nash~\cite{chakrabarti_construction_1986,nash_geometry_1986}. Given $C>1$, the norm of the Higgs field is given by
\begin{equation*}
|\Phi|(r)=\frac{C}{2}\cdot\frac{\left(\frac{1+r}{1-r}\right)^{2C}+1}{\left(\frac{1+r}{1-r}\right)^{2C}-1}-\frac{r^2+1}{4r}.
\end{equation*}
Given a value $C$, the mass of the monopole is $\frac{C-1}{2}$. 

When $C$ is an integer, we see that $|\Phi|(r)$ satisfies Proposition~\ref{prop:palindromic}. However, when $C=\frac{3}{2}$, we see that $|\Phi|(r)=\frac{5r-r^3}{4r^2+12}$, which does not satisfy Proposition~\ref{prop:palindromic}. Moreover, when $2C$ is not an integer, the norm and its square are not rational functions of polynomials and not even defined for $r>1$. Finally, note that integral mass is the condition that twice the mass is in $\mathbb{N}_+$. In this case, the monopole has integral mass if and only if $C\in\mathbb{N}_+\setminus\{1\}$.
\end{note}

\subsection{Singular monopoles}\label{subsubsec:SingularMonopoles}

In this section, we discuss the connection between singular monopoles and circle-invariant instantons, though a different circle action from the one relating hyperbolic monopoles and instantons. In particular, we discuss how a class of singular monopoles correspond to circular $0$-symmetric ADHM data and how to transform this data into a singular monopole.

Let $E\rightarrow R^3\setminus\{0\}$ be a vector bundle with structure group $\mathrm{Sp}(n)$, equipped with a connection $A$, of curvature $F_A$, and a section $\Phi$ of $\mathrm{End}(E)$, called the Higgs field, just as with hyperbolic and Euclidean monopoles. Singular monopoles are solutions to the Bogomolny equations, with finite energy.

The relationship between instantons and singular monopoles was first described by Kronheimer and explored further by Pauly~\cite{kronheimer_monopoles_1985,pauly_gauge_1996,pauly_monopole_1998}. This relationship utilizes the Hopf fibration. 
\begin{definition}
The Hopf map $\pi\colon S^3\rightarrow S^2$ is given by
\begin{equation}
\pi(x,y,z,w):=(x^2+y^2-z^2-w^2,2(yz-xw),2(xz+yw)).
\end{equation}

We may rewrite the Hopf map as $\pi\colon\mathbb{H}\rightarrow\mathfrak{sp}(1)$, given by $\pi(x):=x^\dagger ix$. Note that this new definition matches the previous definition on $S^3\subseteq\mathbb{H}$.
\end{definition}

The Hopf fibration describes $S^3$ as a non-trivial principal bundle over $S^2$ with fibre $S^1$. That is, for every $p\in S^2$, $\pi^{-1}(p)\simeq S^1$. The same is true for all $p\in\mathbb{R}^3\setminus\{0\}$. Additionally, $\pi^{-1}(0)=\{0\}$. Indeed, consider the $R_0$ action on $\mathbb{H}$, taking $x\mapsto \theta.x:=e^{i\theta} x$. We see that $\pi(\theta.x)=\pi(x)$, so the $R_0$ action is transitive and free away from the origin. Note that $\pi$ is not invariant under any other $R_t$ action.

Singular monopoles are in a one-to-one correspondence with circle-invariant instantons on $\mathbb{R}^4\setminus\{0\}$~\cite{pauly_gauge_1996}.  
\begin{definition}
As we are studying instantons over all of $\mathbb{R}^4$, we are interested in singular monopoles arising from instantons that, under some gauge transformation, extend smoothly over the origin. Such objects are called \textbf{monopoles with Dirac type singularities}. More information on such monopoles can be found in Pauly's works, where they are called good singular monopoles~\cite{pauly_gauge_1996,pauly_monopole_1998}.
\end{definition}

Pauly's idea is that given a singular monopole, we can pull-back the vector bundle, Higgs field, and connection to $\mathbb{H}\setminus\{0\}$. Then, by using a one-form that is invariant under the $R_0$ action, we use the monopole to construct a circle-invariant instanton. Conversely, given a circle-invariant instanton, we can use $\pi$ to uniquely define a monopole on $\mathbb{R}^3\setminus\{0\}$. The exact process has been proven by others, but we demonstrate it here for completeness, modifying the process in order to use the usual orientations on three and four dimensional Euclidean space as well as relating monopoles directly to instantons~\cites[Proposition~2]{pauly_monopole_1998}[Lemma~4.1.7]{beckett_equivariant_2020}. Additionally, we provide a formula for computing the monopole directly from the ADHM data.
\begin{prop}
Suppose that $\hat{M}\in\mathcal{M}_{n,k}$ has circular $0$-symmetry, with respect to the usual $R_0$ group, generated by $\rho\in\mathfrak{so}(k)$. Let $X$ denote the coordinates of $\mathfrak{sp}(1)\setminus\{0\}$. Locally, we can identify any point in $\mathbb{H}\setminus\{0\}$ by $X\in\mathfrak{sp}(1)\setminus\{0\}$ and $\theta\in S^1$. Suppose that we have local sections $\mathfrak{sp}(1)\setminus\{0\}\rightarrow\mathbb{H}\setminus\{0\}$ taking $X\mapsto x_X$. Further, suppose we find $V(x_X)$ satisfying $V(x_X)^\dagger \Delta(x_X)=0$ and $V(x_X)^\dagger V(x_X)=I_n$. The corresponding singular monopole $(\Phi,A)$ is given, up to gauge, by 
\begin{align}
\Phi(X)&=\frac{1}{2|X|}V(x_X)^\dagger \mathrm{diag}(-L\rho L^\dagger (LL^\dagger)^{-1},i-\rho) V(x_X);\\
A_l(X)&=V(x_X)^\dagger \left(\pi^*\frac{\partial}{\partial X_l}\right)V(x_X).
\end{align}
The $\partial_l$ denotes partial differentiation with respect to the $l$th coordinate of $\mathfrak{sp}(1)\setminus\{0\}$.\label{prop:singularmonoconst}
\end{prop}

\begin{proof}
Let $\frac{\partial}{\partial \theta}$ be the vector field on $\mathbb{H}\setminus\{0\}$ given by
\begin{equation}
\frac{\partial}{\partial \theta}\Bigr\vert_x:=\frac{d}{d\theta}\Bigr\vert_{\theta=0} e^{i\theta}x=-x_1\frac{\partial}{\partial x_0}+x_0\frac{\partial}{\partial x_1}-x_3\frac{\partial}{\partial x_2}+x_2\frac{\partial}{\partial x_3}.
\end{equation}
Then let
\begin{equation}
\xi:=2(-x_1dx_0+x_0dx_1-x_3dx_2+x_2dx_3).\label{eq:defxi}
\end{equation}
We note that given any vector field $V$ on $\mathbb{H}\setminus\{0\}$, $\xi(V)=2\left\langle V,\frac{\partial}{\partial \theta}\right\rangle$. In particular, $\xi\left(\frac{\partial}{\partial \theta}\right)=2|x|^2$. 

Given any vector field $W$ on $\mathfrak{sp}(1)\setminus\{0\}$, denote by $\pi^*W$ the unique vector field on $\mathbb{H}\setminus\{0\}$ such that $d\pi(\pi^*W)=W$ and $\xi(\pi^*W)=0$. That is, $\pi^*W$ is the unique lift of $W$ orthogonal to the circle orbit.

The relationship between the instanton $\mathbb{A}$ on $\mathbb{H}\setminus\{0\}$ and the monopole $(\Phi,A)$ on $\mathfrak{sp}(1)\setminus\{0\}$ is given as follows. Given a monopole, we have 
\begin{equation}
\mathbb{A}=\pi^*A+\pi^*\Phi\otimes \xi.
\end{equation}
Given an instanton, let $s$ be a section and $W$ a vector field on $\mathfrak{sp}(1)\setminus\{0\}$; we have that the monopole is uniquely determined by
\begin{equation}
\begin{aligned}
\pi^*(\Phi s)&=\frac{1}{2|x|^2}\mathbb{A}_\theta\pi^*s;\\
\pi^*(A(W) s)&=\mathbb{A}(\pi^*W)\pi^*s.
\end{aligned}
\end{equation}
That $\Phi$ is linear with respect to functions follows because the derivative of $\pi^*s$ with respect to $\theta$ vanishes.

Given the relationship between $\mathbb{A}$ and $(\Phi,A)$, we see that the curvature of $\mathbb{A}$ is given by
\begin{equation*}
F_{\mathbb{A}}=\pi^* F_A+\pi^* D_A\Phi\wedge \xi+\pi^*\Phi\otimes d\xi.
\end{equation*}

From the definition of $\xi$, we see that with respect to the standard orientation on $\mathbb{H}\setminus\{0\}$, we have that $d\xi$ is self-dual. That is, $\star d\xi=d\xi$. Denoting the anti-self-dual part of a two-form $B$ by $B^-:=\frac{B-\star B}{2}$, we have that
\begin{equation*}
F_{\mathbb{A}}^-=(\pi^*F_A+\pi^* D_A\Phi\wedge\xi)^-.
\end{equation*}

For any one-form $\omega$ on $\mathfrak{sp}(1)\setminus\{0\}$, we can compute that $\star\pi^*(\star\omega)=-\xi\wedge\pi^*\omega$. As we are looking at the anti-self-dual part of the curvature, we have
\begin{equation*}
F_{\mathbb{A}}^-=(\pi^*(F_A-\star D_A\Phi))^-.
\end{equation*}
If $(\Phi,A)$ is a monopole, then $\mathbb{A}$ is an instanton. The finite action condition is satisfied as the monopole has finite energy and we are integrating with respect to the coordinates on $\mathfrak{sp}(1)\setminus\{0\}$ and $S^1$. 

Conversely, if $\mathbb{A}$ is an instanton, $(\Phi,A)$ is a monopole. Indeed, $F_{\mathbb{A}}^-=0$, so $\pi^* (F_A-\star D_A\Phi)$ is self-dual. Note that for any two-form $B$ on $\mathfrak{sp}(1)\setminus\{0\}$, we have $\star B$ is a one-form. Thus, $\star\pi^* B=-\xi\wedge\pi^*\star B$. As $\pi^*(F_A-\star D_A\Phi)$ is self-dual, this equality means that
\begin{equation*}
\pi^* (F_A-\star D_A\Phi)=\star \pi^*(F_A-\star D_A\Phi)=-\xi\wedge \pi^*\star (F_A-\star D_A\Phi).
\end{equation*}
However, $\xi$ is orthogonal to all pullbacks via $\pi$, so $\pi^*(F_A-\star D_A\Phi)=0$, meaning that $(\Phi,A)$ is a monopole. The finite energy condition is satisfied as $\mathbb{A}$ has finite action.

We can use this relationship between instantons and singular monopoles to compute the monopole in terms of the ADHM data. Note that any $x\in\mathbb{H}\setminus\{0\}$ can be uniquely written as a pair $(\theta,X)\in S^1\times (\mathfrak{sp}(1)\setminus\{0\})$ as $x=e^{i\theta}x_X$. Let $Q(\theta):=\mathrm{diag}(Le^{\rho\theta}L^\dagger (LL^\dagger)^{-1},e^{-i\theta}e^{\rho\theta})\in\mathrm{Sp}(n+k)$ and $K(\theta):=e^{\rho\theta}\in\mathrm{SO}(k)$. Per the proof of Theorem~\ref{thm:circularsym}, we have
\begin{equation*}
(Q(\theta),K(\theta)).\mathrm{diag}(e^{i\theta},1).(\hat{M},U)=(\hat{M},U).
\end{equation*}
Therefore, similar to the proof of Proposition~\ref{prop:makehypermono}, we have that $\Delta(x)=Q(\theta)^\dagger \Delta(X)K(\theta)$. 

Extending $V$ to all of $\mathbb{H}\setminus\{0\}$ by $V(x):=Q(\theta)^\dagger V(x_X)$, we note that we can use $V$ to create our instanton and thus our monopole.

The $\theta$ component of the instanton $\mathbb{A}$ corresponding to $(\hat{M},U)$ is given by
\begin{equation*}
\mathbb{A}_\theta(x)=V(x_X)^\dagger Q(\theta)\partial_\theta Q(\theta)^\dagger V(x_X).
\end{equation*}
Simplifying, using the commutativity proven in Proposition~\ref{prop:makehypermono}, we have
\begin{equation*}
\mathbb{A}_\theta(x)=V(x_X)^\dagger \mathrm{diag}(-L\rho L^\dagger (LL^\dagger)^{-1},i-\rho)V(x_X).
\end{equation*}
Recalling that as $X=\pi(x_X)$, so $|X|=|x_X|^2$, we define $\Phi$ via $\Phi(X):=\frac{1}{2|X|}\mathbb{A}_\theta(x_X)$, noting that $\Phi$ does not depend on $\theta$.

Similarly, we find that 
\begin{equation*}
\mathbb{A}\left(\pi^*\frac{\partial}{\partial X_l}\right)(x)=V(x_X)^\dagger \left(\pi^*\frac{\partial}{\partial X_l}\right)V(x_X).
\end{equation*}
Thus, we can define $A_l(X):=\mathbb{A}\left(\pi^*\frac{\partial}{\partial X_l}\right)(x_X)$, noting that $A_l(X)$ does not depend on $\theta$.
\end{proof}

\begin{note}
If we find $V(x)$ for all $x\in\mathbb{H}\setminus\{0\}$, we can construct $A_l(X)$ via
\begin{align}
A_1(X)&:=\frac{1}{2|X|}\left(x_0\mathbb{A}_0(x_X)+x_1\mathbb{A}_1(x_X)-x_2\mathbb{A}_2(x_X)-x_3\mathbb{A}_3(x_X)\right);\\
A_2(X)&:=\frac{1}{2|X|}\left(-x_3\mathbb{A}_0(x_X)+x_2\mathbb{A}_1(x_X)+x_1\mathbb{A}_2(x_X)-x_0\mathbb{A}_3(x_X)\right);\\
A_3(X)&:=\frac{1}{2|X|}\left(x_2\mathbb{A}_0(x_X)+x_3\mathbb{A}_1(x_X)+x_0\mathbb{A}_2(x_X)+x_1\mathbb{A}_3(x_X)\right).
\end{align}
This result follows from finding the lifts of $\frac{\partial}{\partial X_l}$.
\end{note}

The following lemma provides a relationship between symmetric singular monopoles and symmetric instantons.
\begin{lemma}
We have that for $\theta\in\mathbb{R}$, $x\in\mathbb{H}$, and $p\in\mathrm{Sp}(1)$,
\begin{equation}
\pi(e^{i\theta}xp^\dagger)=p\pi(x)p^\dagger.
\end{equation}
If $p\mapsto R_p$ denotes the double cover $\mathrm{Sp}(1)\rightarrow\mathrm{SO}(3)$, then this equation tells us that under $\pi$, the action of $(e^{i\theta},p)\in\mathrm{Sp}(2)$ becomes $R_p\in\mathrm{SO}(3)$. \label{lemma:symsingularmonopoles}
\end{lemma}

\begin{proof}
This result follows immediately from the definition of the Hopf map.
\end{proof}

In later sections, we investigate symmetric singular monopoles more deeply. However, we see that such monopoles must arise from instantons equivariant under $R_0$, or a group conjugate to it. 

\section{Toral symmetry}\label{subsec:ToralSymmetry}

In this section, we find an equation describing all instantons with toral symmetry, as given in Table~\ref{table:conformalsubgroups} as well as below. We then discuss the connection between such instantons and axially symmetric hyperbolic and singular monopoles.

First, we introduce the notion of toral symmetry. An instanton is said to have toral symmetry if it is equivariant under $\mathrm{diag}(e^{i\phi_1},e^{i\phi_2})$ for all $\phi_1,\phi_2\in\mathbb{R}$.
\begin{theorem}
Let $\hat{M}\in\mathcal{M}_{n,k}$. Then $\hat{M}$ has toral symmetry if and only if there exists $\rho_1,\rho_2\in\mathfrak{so}(k)$ such that $[\rho_1,\rho_2]=0$ and
\label{thm:toralsym}
\begin{align}
iM=[\rho_1,M],&\quad
-Mi=[\rho_2,M],\label{eq:toralsymmetry1}\\
[\rho_1,R]=0,&\quad
[\rho_2,R]=0.\label{eq:toralsymmetry2}
\end{align}
\end{theorem}

\begin{definition}
We call $\rho_1,\rho_2$ the \textbf{generators} of the toral symmetry of $\hat{M}$.
\end{definition}

\begin{note}
All connected Lie subgroups of $\mathrm{Sp}(2)$ with Lie algebra $\mathbb{R}\oplus\mathbb{R}$ are conjugate to $S^1\times S^1$. That is, there is some $A\in\mathrm{Sp}(2)$ such that the Lie group is of the form $A(S^1\times S^1)A^\dagger$. Instantons equivariant under this group are of the form $A^\dagger.(\hat{M},U)$, where $\hat{M}$ is equivariant under $S^1\times S^1$.
\end{note}

\begin{proof}
We use the same notation as introduced in the beginning of the proof of Theorem~\ref{thm:circularsym}. We follow the proof of Theorem~\ref{thm:mainthm} after setting the scene~\cite[Theorem~1.1]{lang_moduli_2024}.

Suppose that $\hat{M}$ has toral symmetry. Let $S\subseteq S^1\times S^1\times\mathrm{O}(k)$ be the stabilizer group of $(\hat{M},U)$ restricted to rotations in $S^1\times S^1$. That is
\begin{multline}
S:=\{(\mathrm{diag}(e^{i\theta_1},e^{i\theta_2}),K)\mid (\mathrm{diag}(Le^{-i\theta_2}KL^\dagger (LL^\dagger)^{-1},e^{-i\theta_1} K),K).\\
\mathrm{diag}(e^{i\theta_1},e^{i\theta_2}).(\hat{M},U)=(\hat{M},U)\}.\label{eq:toralsymmetryS}
\end{multline}
Indeed, if $(Q,K)\in\mathrm{Sp}(n+k)\times\mathrm{GL}(k,\mathbb{R})$ such that $\mathrm{diag}(e^{i\theta_1},e^{i\theta_2}).(Q,K).(\hat{M},U)=(\hat{M},U)$, then Corollary~\ref{cor:IsomS} tells us that $K\in\mathrm{O}(k)$, $[R,K]=0$, and $Q$ is given by the diagonal matrix in \eqref{eq:toralsymmetryS}. Therefore, the pairs in $S$ encapsulate all the toral symmetry of the instanton. That $S$ is a group follows from this fact.

Unlike circular $t$-symmetry, toral symmetry is not one-dimensional. As such, we want to use Theorem~\ref{thm:mainthm} and require a compact Lie group for the gauge group. We see that although our initial gauge action involved a non-compact Lie group, we can focus on a compact subgroup. As such, the stabilizer group is a subgroup of a compact group, just as in the proof of Theorem~\ref{thm:mainthm}~\cite[Theorem~1.1]{lang_moduli_2024}. We then proceed as in said proof.

In particular, we find that $\hat{M}$ has toral symmetry if and only if there is a Lie algebra homomorphism $\rho\colon\mathbb{R}\oplus\mathbb{R}\rightarrow\mathfrak{so}(k)$ such that for all $(\theta_1,\theta_2)\in\mathbb{R}\oplus\mathbb{R}$, we have
\begin{equation}
\left(\begin{bmatrix}
L\left(\rho(\theta_1,\theta_2)-\theta_2 i\right)L^\dagger(LL^\dagger)^{-1}L-L\left(\rho(\theta_1,\theta_2)-\theta_2 i\right) \\
M\theta_1 i-\theta_2 iM+[\rho(\theta_1,\theta_2),M]
\end{bmatrix},0\right)=(0,0).\label{eq:toralfulleq}
\end{equation}

We can simplify these constraints. In particular, suppose $\hat{M}$ has toral symmetry. If we let $\rho_1:=\rho(1,0)$ and $\rho_2:=\rho(0,1)$, then we have $[\rho_1,\rho_2]=0$. Moreover, evaluating \eqref{eq:toralfulleq} at $(\theta_1,\theta_2)=(1,0),(0,1)$, we have, respectively,
\begin{gather}
\begin{bmatrix}
L\rho_1L^\dagger (LL^\dagger)^{-1}L-L\rho_1 \\
-iM+[\rho_1,M]
\end{bmatrix}=0;\label{eq:toraldiff1}\\
\begin{bmatrix}
-LiL^\dagger (LL^\dagger)^{-1}L+L\rho_2L^\dagger (LL^\dagger)^{-1}L+Li-L\rho_2 \\
[\rho_2,M]+Mi
\end{bmatrix}=0.\label{eq:toraldiff2}
\end{gather} 
Focusing on the bottom rows, we see that \eqref{eq:toralsymmetry1} is satisfied. Furthermore, just as in the circular $t$-symmetry case, Corollary~\ref{cor:IsomS} tells us that $[R,e^{\rho_l \theta}]=0$, for all $\theta\in\mathbb{R}$. Differentiating and evaluating at $\theta=0$ for both $l=1,2$, we have $[R,\rho_1]=0$ and $[R,\rho_2]=0$.

We can use the same method as in the proof of Theorem~\ref{thm:circularsym} to prove the converse.
\end{proof}

Just as with circular $t$-symmetry, we need not check the final condition of $\mathcal{M}_{n,k}$ in Definition~\ref{def:Mstd} everywhere.
\begin{lemma}
Suppose that $\hat{M}$ satisfies \eqref{eq:toralsymmetry1} and \eqref{eq:toralsymmetry2} for some $\rho_1,\rho_2\in\mathfrak{so}(k)$ satisfying $[\rho_1,\rho_2]=0$, as well as the first three conditions of Definition~\ref{def:Mstd}. If the final condition is satisfied at all $x=x_0+x_2j\in\mathbb{H}$ with $x_0,x_2\geq 0$, then $\hat{M}\in\mathcal{M}_{n,k}$.
\end{lemma}

\begin{proof}
Let $z=z_0+z_1i+z_2j+z_3k\in \mathbb{H}$. Note that $\mathrm{diag}(e^{i\theta},e^{i\phi})$ acts on $z$ as $z\mapsto e^{i\theta}ze^{-i\phi}$. Specifically, it takes $z\mapsto (z_0+z_1i)e^{i(\theta-\phi)}+(z_2+z_3i)e^{i(\theta+\phi)}j$. Thus, there is some $\theta,\phi\in\mathbb{R}$ such that $x:=e^{i\theta}ze^{-i\phi}$ has no $k$ or $i$ component and the real and $j$ component are non-negative. The rest follows as in the proof of Lemma~\ref{lemma:finalconditioncirc}.
\end{proof}

\subsection{Structure of toral symmetry}

In light of Theorem~\ref{thm:toralsym}, we know that we can find instantons with toral symmetry given $\rho_1,\rho_2\in\mathfrak{so}(k)$. Note that such matrices generate a real representation of $\mathbb{R}\oplus\mathbb{R}$. Starting with such a real representation, we can narrow down the possible $M$, so we are only left with finding $L$ such that we have an instanton, which necessarily has toral symmetry. 

Much like the case of circular $t$-symmetry, when $t$ is rational, we can do better than representations of $\mathbb{R}\oplus\mathbb{R}$. Note that a representation $\rho\colon\mathbb{R}\oplus\mathbb{R}\rightarrow\mathfrak{so}(k)$ corresponds to a representation of $S^1\times S^1$ if and only if $\mathrm{exp}(\rho(2\pi,0))=I_k=\mathrm{exp}(\rho(0,2\pi))$. 
\begin{prop}
Let $\hat{M}\in\mathcal{M}_{n,k}$. Then $\hat{M}$ has toral symmetry if and only if there exists a real representation $\rho\colon\mathbb{R}\oplus\mathbb{R}\rightarrow\mathfrak{so}(k)$, corresponding to a representation of $S^1\times S^1$, such that for all $t_1,t_2\in\mathbb{R}$,
\label{prop:toralsymrep}
\begin{align}
t_1iM-t_2Mi=\frac{1}{2}[\rho(t_1,t_2),M],\label{eq:toralsymmetryrepversion}\\
[\rho(t_1,t_2),R]=0.\label{eq:toralsymmetryrepversion2}
\end{align}
\end{prop}

\begin{definition}
We call $\rho$ the \textbf{generating representation} of the toral symmetry of $\hat{M}$.
\end{definition}

\begin{proof}
If such a $\rho$ exists, by Theorem~\ref{thm:toralsym}, we have toral symmetry generated by the pair $\frac{1}{2}\rho(1,0)$ and $\frac{1}{2}\rho(0,1)$. Conversely, suppose that $\hat{M}$ has toral symmetry. By Theorem~\ref{thm:toralsym}, there exists $\rho_1,\rho_2\in\mathfrak{so}(k)$ satisfying $[\rho_1,\rho_2]=0$, \eqref{eq:toralsymmetry1}, and \eqref{eq:toralsymmetry2}. 

Note that $\rho_1$ and $\rho_2$ generate circular $0$-symmetry in $(\hat{M},U)$. However, $\rho_1$ generates the usual symmetry, whereas $\rho_2$ generates a conjugate circular $0$-symmetry.

As $\rho_1$ and $\rho_2$ are diagonalizeable and commute, they are simultaneously diagonalizeable. Furthermore, noting that the eigenvalues of a matrix in $\mathfrak{so}(k)$ come in pairs $\pm \lambda i$ for $\lambda\in\mathbb{R}$ or single zeros and $[\rho_1,\rho_2]=0$, there is some $U\in\mathrm{O}(k)$ such that 
\begin{equation*}
\rho_1=U\mathrm{diag}(\rho_{11},\ldots,\rho_{1m})U^T \quad\textrm{and}\quad
\rho_2=U\mathrm{diag}(\rho_{21},\ldots,\rho_{2m})U^T,
\end{equation*}
where $\rho_{1j}$ is either $\begin{bmatrix}
0 & a_{1j} \\ 
-a_{1j} & 0
\end{bmatrix}$ for some $a_{1j}\in\mathbb{R}$ or $\rho_{1j}=0$. Additionally, $\rho_{2j}$ is either zero or there exists some $a_{2j},b_{2j}\in\mathbb{R}$ such that $\rho_{2j}=\begin{bmatrix}
ib_{2j} & a_{2j} \\ -a_{2j} & ib_{2j}
\end{bmatrix}$. This fact follows from the simultaneous diagonalization of the matrices. 

Just as in the proof of Proposition~\ref{prop:circulartsymrep}, we can choose the $a_{1j}$ such that $2\rho_1$ generates a representation of $S^1$. Importantly, we note that changing the $a_{1j}$ does not affect the commutativity of $\rho_1$ and $\rho_2$. Similarly, we can choose a $U'\in\mathrm{O}(k)$ such that the forms of $\rho_1$ and $\rho_2$ are swapped. We can then repeat a similar analysis and modify $\rho_2$ such that $2\rho_2$ generates a representation of $S^1$.

As the matrices commute and each generates a representation of $S^1$, $\tilde{\rho}\colon\mathbb{R}\oplus\mathbb{R}\rightarrow\mathfrak{so}(k)$ taking $\tilde{\rho}(t_1,t_2):=2t_1\rho_1+2t_2\rho_2$ satisfies \eqref{eq:toralsymmetryrepversion} and generates a representation of $S^1\times S^1$. That $\tilde{\rho}$ satisfies \eqref{eq:toralsymmetryrepversion2} follows from Corollary~\ref{cor:IsomS}.
\end{proof}

Proposition~\ref{prop:toralsymrep} tells us that we can use the representation theory of $S^1\times S^1$ to better understand instantons with toral symmetry. 

The generating representation in Proposition~\ref{prop:toralsymrep} was obtained  by examining the bottom of \eqref{eq:toraldiff1} and \eqref{eq:toraldiff2}. By focusing instead on the top of those equations, we obtain an induced representation.
\begin{lemma}
Suppose that $\hat{M}\in\mathcal{M}_{n,k}$ has toral symmetry generated by a representation $(\mathbb{R}^k,\rho)$, which corresponds to a representation of $S^1\times S^1$. Let $\lambda\colon\mathbb{R}\oplus\mathbb{R}\rightarrow\mathfrak{sp}(n)$ be defined by $\lambda(t_1,t_2):=L\left(\rho(t_1,t_2)-2t_2i\right)L^\dagger (LL^\dagger)^{-1}$. We have $(\mathbb{H}^n,\lambda)$ is a quaternionic representation of $\mathbb{R}\oplus\mathbb{R}$, corresponding to a representation of $S^1\times S^1$.
\end{lemma}

\begin{proof}
By Proposition~\ref{prop:toralsymrep}, we know that for all $t_1,t_2\in\mathbb{R}$, \eqref{eq:toralsymmetryrepversion} and \eqref{eq:toralsymmetryrepversion2} hold. Using the former,
\begin{equation*}
\begin{aligned}
M^\dagger M\left(\rho(t_1,t_2)-2t_2i\right)&=M^\dagger[M,\rho(t_1,t_2)]+M^\dagger \rho(t_1,t_2) M-2t_2M^\dagger Mi\\
&=M^\dagger \rho(t_1,t_2) M-2t_1M^\dagger iM\\
&=[M^\dagger,\rho(t_1,t_2)]M+\rho(t_1,t_2) M^\dagger M-2t_1M^\dagger iM\\
&=\left(\rho(t_1,t_2)-2t_2i\right)M^\dagger M.
\end{aligned}
\end{equation*}
Hence, $\left[M^\dagger M,\rho(t_1,t_2)-2t_2i\right]=0$. Then, as $[\rho(t_1,t_2),R]=0$, by \eqref{eq:toralsymmetryrepversion2}, and $R$ is real, $\left[\rho(t_1,t_2)-2t_2i,R\right]=0$, so $\left[L^\dagger L,\rho(t_1,t_2)-2t_2i\right]=0$. Thus, we see that 
\begin{equation*}
\left[LL^\dagger ,L\left(\rho(t_1,t_2)-2t_2i\right)L^\dagger\right]=0.
\end{equation*} 
Therefore, $\lambda(t_1,t_2)\in\mathfrak{sp}(n)$. 

Just as in the circular $t$-symmetry case, when $t\in\mathbb{Q}$, we have that 
\begin{equation*}
e^{\lambda(t_1,t_2)}=Le^{\rho(t_1,t_2)-2t_2i}L^\dagger (LL^\dagger)^{-1}.
\end{equation*}
As $\rho$ and $i$ commute, we have that as $(\mathbb{R}^k,\rho)$ corresponds to a representation of $S^1\times S^1$, $e^{\lambda(2\pi,0)}=I_n=e^{\lambda(0,2\pi)}$. Therefore, $(\mathbb{H}^n,\lambda)$ corresponds to a representation of $S^1\times S^1$.
\end{proof}

We know that this induced representation satisfies the following equation.
\begin{cor}
If $\hat{M}\in\mathcal{M}_{n,k}$ has toral symmetry generated by $(\mathbb{R}^k,\rho)$, then $\lambda(t_1,t_2)L-L\rho(t_1,t_2)+2t_2Li=0$, for all $t_1,t_2\in\mathbb{R}$.
\end{cor}

\begin{proof}
This result follows from the top components of \eqref{eq:toraldiff1} and \eqref{eq:toraldiff2}.
\end{proof}

\subsection{Axially symmetric singular monopoles}\label{subsubsec:AxialSymSingularMono}

In this section, we discuss the connections between toral symmetric instantons and axially symmetric singular monopoles. Just like hyperbolic monopoles, axially symmetric singular monopole are monopoles symmetric under all rotations about one of their axes. 
\begin{prop}
A singular $\mathrm{Sp}(n)$-monopole with Dirac type singularities is axially symmetric if and only if its ADHM data has toral symmetry. \label{prop:axialsymsingularmono}
\end{prop}

\begin{proof}
Recall the correspondence between a singular monopole $(\Phi,A)$ and an instanton $\mathbb{A}$ given by \eqref{eq:defxi}. Suppose that we have an axially symmetric singular monopole with Dirac type singularities. Let $p\mapsto R_p$ be the double cover $\mathrm{Sp}(1)\rightarrow\mathrm{SO}(3)$. Then there is some $\upsilon\in\mathfrak{sp}(1)$ such that the monopole is equivariant under $R_{e^{\upsilon\phi}}$ for all $\phi\in\mathbb{R}$. Then there is some gauge transformation $g$ such that $g.\Phi=R_{e^{\upsilon\phi}}^*\Phi$ and $g.A=R_{e^{\upsilon\phi}}^*A$. Let $f_{\theta,\phi}\colon\mathbb{H}\setminus\{0\}\rightarrow\mathbb{H}\setminus\{0\}$ be the rotation $f_{\theta,\phi}(x):=e^{i\theta}xe^{-\upsilon\phi}$. It is straightforward to show that the one-form $\xi$ defined in \eqref{eq:defxi} is invariant under $f_{\theta,\phi}$. That is, $f_{\theta,\phi}^*\xi=\xi$. By Lemma~\ref{lemma:symsingularmonopoles}, we have that $\pi\circ f_{\theta,\phi}=R_{e^{\upsilon\phi}}\circ\pi$. Then we have that
\begin{equation*}
\begin{aligned}
f_{\theta,\phi}^*\mathbb{A}&=f_{\theta,\phi}^*\circ \pi^*A+(f_{\theta,\phi}^*\circ \pi^*\Phi)\otimes \xi\\
&=\pi^*\circ R_{e^{\upsilon\phi}}^*A+(\pi^*\circ R_{e^{\upsilon\phi}}^*\Phi)\otimes \xi\\
&=\pi^*(g).\pi^*A+\pi^*(g).\Phi\otimes \xi=\pi^*(g).\mathbb{A}.
\end{aligned}
\end{equation*}
Thus, $\mathbb{A}$ is equivariant under $f_{\theta,\phi}$.

Conversely, suppose that $\mathbb{A}$ is equivariant under $f_{\theta,\phi}$, as defined above. Then there is some gauge transformation $g$ such that $f_{\theta,\phi}^*\mathbb{A}=g.\mathbb{A}$. Hence,
\begin{equation*}
\pi^*( R_{e^{\upsilon\phi}}^*A)+\pi^*( R_{e^{\upsilon\phi}}^*\Phi)\otimes\xi=f_{\theta,\phi}^*\mathbb{A}=g.\mathbb{A}=g.\pi^*A+g.\pi^*\Phi\otimes\xi.
\end{equation*}
Evaluating on $\pi^*\partial_l$, recalling that $\xi(\pi^*\partial_l)=0$ and $d\pi (\pi^*\partial_l)=\partial_l$ by definition, we can decompose the previous expression, finding that the monopole is equivariant under $R_{e^{\upsilon\phi}}$. Indeed, the gauge transformation descends as everything but $g$ is already known to be circle-invariant.

Therefore, $(\Phi,A)$ is equivariant under all $R_{e^{\upsilon\phi}}$ if and only if $\mathbb{A}$ is equivariant under $\mathrm{diag}(e^{i\theta},e^{\upsilon\phi})$ for all $\theta,\phi\in\mathbb{R}$. We see that this subgroup is conjugate to $S^1\times S^1$, so we have an axially symmetric, singular monopole with Dirac type singularities if and only if its ADHM data has toral symmetry.
\end{proof}

Although there are two commuting circle actions conjugate to $R_0$ in $S^1\times S^1$, only one singular monopole can be created from an instanton with toral symmetry using the given Hopf map, as the other circle action does not preserve this map. The Hopf map must be changed in order for another circle action to create a singular monopole.

\subsection{Axially symmetric hyperbolic monopoles}

In this section, we discuss the connections between toral symmetric instantons and axially symmetric hyperbolic monopoles. An axially symmetric hyperbolic monopole is a monopole that is symmetric under rotations about one of its axes. As the subgroups rotating about axes are all conjugate, we focus on rotations about the $z$-axis. 
\begin{prop}
A hyperbolic $\mathrm{Sp}(n)$-monopole with integral mass is axially symmetric if and only if its ADHM data has toral symmetry. \label{prop:axialhypermono}
\end{prop}

\begin{proof}
Through isometries, we may assume that the hyperbolic model is the half-space model and the monopole is symmetric about the $z$-axis. Because of the conformal equivalence of $S^4\setminus S^2\equiv H^3\times S^1$, these isometries on hyperbolic space correspond to conformal maps on $S^4$.

The action of $R_1$ induces the half-space model of hyperbolic space, with coordinates $(x_0,x_1,r)$ with $r>0$. The two-sphere removed is $\mathbb{C}\cup\{\infty\}$. Letting $\theta$ be the coordinate of $S^1$, we have that a point $(x_0,x_1,r,\theta)\in H^3\times S^1$ corresponds to a point $x_0+x_1i+r\cos\theta j+r\sin\theta k\in\mathbb{H}$. 

We now look at how $S^1\times S^1$ acts on these coordinates. Given $x\in\mathbb{H}$, write $x=x_0+x_1i+r\cos\theta j+r\sin\theta k$. Also, let $\mathrm{diag}(e^{i(\phi_1+\phi_2)},e^{i(\phi_1-\phi_2)})\in S^1\times S^1$. This isometry takes $x\mapsto e^{i(\phi_1+\phi_2)}xe^{-i(\phi_1-\phi_2)}$. Simplifying, we see that $x_0,x_1,r$ are transformed via
\begin{equation*}
\begin{bmatrix}
x_0 \\ x_1 \\ r
\end{bmatrix}\mapsto \begin{bmatrix}
\cos 2\phi_2 & -\sin 2\phi_2 & 0 \\
\sin 2\phi_2 & \cos 2\phi_2 & 0 \\
0 & 0 & 1
\end{bmatrix}\begin{bmatrix}
x_0 \\ x_1 \\ r
\end{bmatrix}.
\end{equation*}
Additionally, $\theta\mapsto \theta-2\phi_1$. Thus, we see that this action rotates $H^3$ about the $z$-axis by $2\phi_2$ and rotates $S^1$ by $-2\phi_1$. 

Now that we understand the relationship between the toral action on $S^4$ and the corresponding action on $H^3\times S^1$, we can proceed. Note the correspondence between a hyperbolic monopole with integral mass $(\Phi,A)$ and an instanton $\mathbb{A}$ is given by $\mathbb{A}=\pi_1^*A+\pi_1^*\Phi d\theta$, where $\pi_1\colon S^4\setminus S^2\rightarrow H^3$. As $\pi^*\Phi d\theta$ and $\pi^*A$ are orthogonal, we can proceed as in the proof of Proposition~\ref{prop:axialsymsingularmono}. Ultimately, we have a hyperbolic monopole with integral mass is axially symmetric if and only if its ADHM data has toral symmetry.
\end{proof}

Unlike singular monopoles, as there are two commuting circle actions conjugate to $R_1$ in $S^1\times S^1$, different hyperbolic monopoles can be constructed from one instanton with toral symmetry by using different circle actions conjugate to $R_1$. In Proposition~\ref{prop:notgauge}, we see that these monopoles need not be gauge equivalent.

\section{Spherical symmetry}\label{subsec:SphericalSymmetry}

In this section, we find equations describing all instantons with the various kinds of spherical symmetry, as given in Table~\ref{table:conformalsubgroups}. We separate the cases of simple, isoclinic, and conformal spherical symmetry. We also discuss the connections between these symmetric instantons and hyperbolic analogues to Higgs bundles and Nahm data. Finally, we discuss instantons with full symmetry.

\begin{note}
All connected Lie subgroups of $\mathrm{Sp}(2)$ with Lie algebra $\mathfrak{sp}(1)$ are conjugate to the connected Lie subgroups $G_{3,1,1}$, $G_{4,1}$, or $G_{5}$ induced by $\mathfrak{h}_{3,1,1}$, $\mathfrak{h}_{4,1}$, and $\mathfrak{h}_{5}$, respectively, defined in \eqref{eq:sp1subalgebras}. That is, there is some $A\in\mathrm{Sp}(2)$ and $G$ one of the subgroups $G_{3,1,1}$, $G_{4,1}$, or $G_5$ such that the Lie group is of the form $AGA^\dagger$. Instantons equivariant under this group are of the form $A^\dagger.(\hat{M},U)$, where $\hat{M}$ is equivariant under $G$.
\end{note}

\begin{definition}
Let $\varepsilon$ be the Levi-Civita symbol. Let $(\upsilon_1,\upsilon_2,\upsilon_3):=(i/2,j/2,k/2)$ be the \textbf{standard basis} of $\mathfrak{sp}(1)$. Note that $[\upsilon_l,\upsilon_m]=\sum_{p=1}^3\varepsilon_{lmp}\upsilon_p$.\label{def:stdbasissp1}
\end{definition}

\subsection{Simple spherical symmetry}\label{subsubsec:SimpleSphericalSym}

In this section, we find an equation describing all instantons with simple spherical symmetry, as given in Table~\ref{table:conformalsubgroups} as well as below. We also discuss the connections between these instantons and hyperbolic monopoles with no continuous symmetries and a hyperbolic analogue to Higgs bundles.

First, we introduce the notion of simple spherical symmetry. An instanton is said to have simple spherical symmetry if it is equivariant under $\mathrm{diag}(p,p)$ for all $p\in\mathrm{Sp}(1)$.
\begin{theorem}
Let $\hat{M}\in\mathcal{M}_{n,k}$. Then $\hat{M}$ has simple spherical symmetry if and only if there exists a real representation $\rho\colon\mathfrak{sp}(1)\rightarrow\mathfrak{so}(k)$ such that for all $\upsilon\in\mathfrak{sp}(1)$,
\label{thm:simplesphericalsym}
\begin{align}
[\upsilon,M]+[\rho(\upsilon),M]=0,\label{eq:simplesphericalsym1}\\
[\rho(\upsilon),R]=0.\label{eq:simplesphericalsym2}
\end{align}
\end{theorem}

\begin{definition}
We call $\rho$ the \textbf{generating representation} of the simple spherical symmetry of $\hat{M}$.
\end{definition}

\begin{proof}
We use the same notation as introduced in the beginning of the proof of Theorem~\ref{thm:circularsym}. We follow the proof of Theorem~\ref{thm:mainthm} after setting the scene.

Suppose that $\hat{M}$ has simple spherical symmetry. Let $S\subseteq \mathrm{Sp}(1)\times\mathrm{O}(k)$ be the stabilizer group of $(\hat{M},U)$ restricted to rotations in the diagonal embedding of $\mathrm{Sp}(1)$ into $\mathrm{Sp}(2)$. That is
\begin{equation}
S:=\{(\mathrm{diag}(p,p),K)\mid (\mathrm{diag}(Lp KL^\dagger (LL^\dagger)^{-1},p K),K).\mathrm{diag}(p^\dagger,p^\dagger).(\hat{M},U)
=(\hat{M},U)\}.\label{eq:simplesphericalsymS}
\end{equation}
Indeed, if $(Q,K)\in\mathrm{Sp}(n+k)\times\mathrm{GL}(k,\mathbb{R})$ such that $\mathrm{diag}(p^\dagger,p^\dagger).(Q,K).(\hat{M},U)=(\hat{M},U)$, then Corollary~\ref{cor:IsomS} tells us that $K\in\mathrm{O}(k)$, $[R,K]=0$, and $Q=\mathrm{diag}(LpKL^\dagger (LL^\dagger)^{-1},p K)$. Therefore, the pairs in $S$ encapsulate all the simple spherical symmetry of the instanton. That $S$ is a group follows from this fact as well as the fact that the conformal action is a right Lie group action. Indeed, it is easy to see that $S$ contains the identity and is closed under inversion. To see that $S$ is closed under multiplication, suppose that $(p_1I_2,K_1),(p_2I_2,K_2)\in S$. Then there exists unique $Q_1,Q_2\in\mathrm{Sp}(n+k)$ such that for $l=1,2$,
\begin{equation*}
(Q_l,K_l).p_l^\dagger I_2.(\hat{M},U)=(\hat{M},U).
\end{equation*}
Then we see that 
\begin{equation*}
(Q_1Q_2,K_1K_2).(p_1p_2)^\dagger I_2.(\hat{M},U)=(Q_1,K_1).(Q_2,K_2).p_1^\dagger I_2.p_2^\dagger I_2.(\hat{M},U).
\end{equation*}
As the two actions commute, we see that $(Q_1Q_2,K_1K_2).(p_1p_2)^\dagger I_2.(\hat{M},U)=(\hat{M},U)$. By Corollary~\ref{cor:IsomS}, the product $Q_1Q_2$ is such that $(p_1p_2 I_2,K_1K_2)\in S$. That is, $S$ is a group.

Just like toral symmetry, the stabilizer group is a subgroup of a compact group, so we proceed as in the proof of Theorem~\ref{thm:mainthm}~\cite[Theorem~1.1]{lang_moduli_2024}. In particular, we find that $\hat{M}$ has simple spherical symmetry if and only if there is a Lie algebra homomorphism $\rho\colon\mathfrak{sp}(1)\rightarrow\mathfrak{so}(k)$ such that for all $\upsilon\in\mathfrak{sp}(1)$, we have
\begin{equation}
\left(\begin{bmatrix}
L\upsilon L^\dagger(LL^\dagger)^{-1}L+L\rho(\upsilon)L^\dagger (LL^\dagger)^{-1}L-L\rho(\upsilon)-L\upsilon \\
[\upsilon,M]+[\rho(\upsilon),M]
\end{bmatrix},0\right)=(0,0).\label{eq:simplesphericalsymdiff}
\end{equation}

We can simplify these constraints. In particular, suppose $\hat{M}$ has simple spherical symmetry. In Note~\ref{note:simplesphericalsymconstraint}, we discuss the top row of \eqref{eq:simplesphericalsymdiff}. For now, focusing on the bottom row, we see that $[\upsilon,M]+[\rho(\upsilon),M]=0$. Furthermore, as mentioned above, Corollary~\ref{cor:IsomS} tells us that $[R,e^{\theta\rho(\upsilon)}]=0$, for all $\theta\in\mathbb{R}$. Differentiating and evaluating at $\theta=0$, we have $[R,\rho(\upsilon)]=0$.

We can use the same method as in the proof of Theorem~\ref{thm:circularsym} to prove the converse.
\end{proof}

Just as with our previous symmetries, we need not check the final condition of $\mathcal{M}_{n,k}$ in Definition~\ref{def:Mstd} everywhere.
\begin{lemma}
Suppose that $\hat{M}$ satisfies \eqref{eq:simplesphericalsym1} and \eqref{eq:simplesphericalsym2} for some real representation $(\mathbb{R}^k,\rho)$ as well as the first three conditions of Definition~\ref{def:Mstd}. If the final condition is satisfied at all $x=x_0+x_1i\in\mathbb{H}$ with $x_1\geq 0$, then $\hat{M}\in\mathcal{M}_{n,k}$. 
\end{lemma}

\begin{proof}
Let $z=z_0+z_1i+z_2j+z_3k\in \mathbb{H}$. Note that $z-z_0\in\mathfrak{sp}(1)$. Recall that for $p\in\mathrm{Sp}(1)$ and $x\in\mathfrak{sp}(1)$, $x\mapsto pxp^\dagger$ corresponds to rotating $x$ by $R_p\in\mathrm{SO}(3)$. There exists some $p\in\mathrm{Sp}(1)$ and $x_1\geq 0$ such that $p(z-z_0)p^\dagger=x_1i$. We see that $pzp^\dagger=z_0+x_1i=:x$. The rest follows as in the proof of Lemma~\ref{lemma:finalconditioncirc}.
\end{proof}

\subsubsection{Structure of simple spherical symmetry}

Theorem~\ref{thm:simplesphericalsym} tells us exactly how to search for instantons with simple spherical symmetry: use a real $k$-representation of $\mathfrak{sp}(1)$ to narrow down the possible $M$. Then we are only left with finding $L$ such that $\hat{M}\in\mathcal{M}_{n,k}$. In this section, we investigate what representations generate instantons with simple spherical symmetry and what the corresponding ADHM data looks like.

\begin{note}
For every $k\in\mathbb{N}_+$, there is a unique, up to isomorphism, irreducible complex $k$-representation $(V_k,\rho_k)$ with highest weight $\frac{k-1}{2}$.

Additionally, for every $k\in\mathbb{N}_+$ odd or divisible by four, there is a unique, up to isomorphism, irreducible real $k$-representation $(\mathbb{R}^k,\varrho_k)$. When $k$ is odd, the complexification of this representation is isomorphic to $(V_k,\rho_k)$. When $k$ is divisible by four, the complexification of $(\mathbb{R}^k,\varrho_k)$ is isomorphic to $(V_{k/2},\rho_{k/2})^{\oplus 2}$.
\end{note}

\begin{definition}
Let $\mathrm{ad}\colon\mathfrak{sp}(1)\rightarrow\mathfrak{gl}(\mathbb{H})$ be the Lie algebra homomorphism $\mathrm{ad}(x)(z):=[x,z]$. Consider the real representation $(\mathbb{H},\mathrm{ad})$ and let $V:=\mathbb{R}^k$. Given a real representation $(V,\rho)$ of $\mathfrak{sp}(1)$, we define the induced real representation 
\begin{equation}
(\hat{V},\hat{\rho}):=(V,\rho)\otimes_\mathbb{R} (V^*,\rho^*)\otimes_\mathbb{R} (\mathbb{H},\mathrm{ad}).
\end{equation}
Unravelling how $\hat{\rho}$ acts on $\hat{V}=\mathrm{Mat}(k,k,\mathbb{H})$, let $\upsilon\in\mathfrak{sp}(1)$ and $A\in\mathrm{Mat}(k,k,\mathbb{H})$. Then
\begin{equation}
\hat{\rho}(\upsilon)(A)=[\rho(\upsilon),A]+[\upsilon,A].
\end{equation}  
\end{definition}

\begin{note}
As $\mathfrak{sp}(1)$ has Dynkin diagram $A_1$, all representations of $\mathfrak{sp}(1)$ are self-dual.
\end{note}

\begin{note}
Note that we can restrict the action of $\mathrm{ad}(x)$ to $\mathfrak{sp}(1)$. Denote the adjoint representation of $\mathfrak{sp}(1)$ by $(\mathfrak{sp}(1),\mathrm{ad})$ and the trivial representation by $(\mathbb{R},0)$. The real representation $(\mathbb{H},\mathrm{ad})$ decomposes as
\begin{equation}
(\mathbb{H},\mathrm{ad})=(\mathfrak{sp}(1),\mathrm{ad})\oplus (\mathbb{R},0).
\end{equation}
\end{note}

The definition of $(\hat{V},\hat{\rho})$ is well-motivated. Firstly, note that $M\in\hat{V}$. Secondly, given the connection between the action of $\hat{\rho}(\upsilon)$ and \eqref{eq:simplesphericalsym1}, we immediately obtain the following corollary.
\begin{cor}
Let $\hat{M}\in\mathcal{M}_{n,k}$. Then $\hat{M}$ has simple spherical symmetry if and only if there is some real $k$-representation $(V,\rho)$ such that $\hat{\rho}(\upsilon)(M)=0$ and $[\rho(\upsilon),R]=0$ for all $\upsilon\in\mathfrak{sp}(1)$.\label{cor:simplesphersym}
\end{cor}

\begin{note}
If $R$ is proportional to the identity, then it automatically commutes with everything, so we can ignore the commutators with $R$. Hence, such a condition does not appear anywhere in previous work dealing with hyperbolic monopoles~\cite{lang_hyperbolic_2023}.
\end{note}

\begin{note}
Suppose that $\hat{M}\in\mathcal{M}_{n,k}$ has simple spherical symmetry. Corollary~\ref{cor:simplesphersym} tells us that there is some real representation $(V,\rho)$ of $\mathfrak{sp}(1)$ such that $\mathrm{span}(M)\subseteq\hat{V}$ is an invariant subspace, which is acted on trivially. 

As $\mathfrak{sp}(1)$ is semi-simple, the representation $(\hat{V},\hat{\rho})$ decomposes into irreducible representations. We explore the case $M=0$ below, so suppose that $M\neq 0$. As $\mathrm{span}(M)$ is a one-dimensional invariant subspace, acted on trivially, $(\mathrm{span}(M),0)$ is a summand of the representation $(\hat{V},\hat{\rho})$. So the decomposition of $(\hat{V},\hat{\rho})$ must contain trivial summands. Moreover, $M$ is in the direct sum of these trivial summands. Trivially, if $M=0$, then it is in the direct sum of the trivial summands as well.
\end{note}

\begin{prop}[The $M=0$ case]
Suppose $M=0_k$ and $\hat{M}\in\mathcal{M}_{n,k}$. Then $n=k$ and there is some gauge in which $L$ is diagonal and positive-definite. Such a $\hat{M}$ has rotational symmetry.

Let the norm of an element $A\in\mathfrak{sp}(k)$ be given by $|A|^2:=-\frac{1}{k}\mathrm{Tr}(A^2)$. Then we have that the corresponding instanton is given, up to gauge, by
\begin{equation}
\mathbb{A}=\frac{xdx^\dagger-(dx)x^\dagger}{2}\cdot \mathrm{diag}\left(\frac{1}{|x|^2+\alpha_1^2},\ldots,\frac{1}{|x|^2+\alpha_k^2}\right).
\end{equation}\label{prop:M=0}
\end{prop}

\begin{note}
While such instantons have rotational symmetry, one class in the moduli space has full symmetry (the class with $L\in\mathrm{Sp}(k)$). However, we note that if all $\alpha_i$ are equal, then while the symmetry group of the instanton may not be $\mathrm{Sp}(2)$ precisely, it is conjugate to this group. This fact serves as a reminder that if an instanton possesses a certain kind of symmetry, it may possess more.
\end{note}

\begin{proof}
Suppose $M=0$. We know that $n\leq k$. As $\mathrm{rank}(L^\dagger L)\leq n$ and $R:=L^\dagger L+M^\dagger M=L^\dagger L$ is non-singular, we see that $n=k$. As $R$ is real and positive-definite, it has a unique symmetric square root $R^{1/2}$. Furthermore, as $R$ is positive-definite, $R^{1/2}$ is non-singular. Gauging $L$ by $R^{-1/2}L^\dagger$, we obtain $R^{1/2}$. As $R^{1/2}$ is real and symmetric, it is orthogonally diagonalizeable. Hence, we can gauge our $\hat{M}$ by such a matrix so as to ensure that $L$ is diagonal and positive-definite. That is, there are $\alpha_1,\ldots,\alpha_k>0$ such that $L=\mathrm{diag}(\alpha_1,\ldots,\alpha_k)$. 

We now show that the corresponding instanton has rotational symmetry. Consider $\mathrm{diag}(p,q)\in\mathrm{Sp}(2)$. Let $(\hat{M}',U'):=\mathrm{diag}(p,q).(\hat{M},U)$. We see that $\hat{M}'=\hat{M}q$ and $U'=Up$. Let $Q:=\mathrm{diag}(q^\dagger,p^\dagger)\otimes I_k\in\mathrm{Sp}(2k)$. Simplifying, we see that $Q\hat{M}'=\hat{M}$ and $QU'=U$. That is, 
\begin{equation*}
(Q,I_k).\mathrm{diag}(p,q).(\hat{M},U)=(\hat{M},U).
\end{equation*} 
Hence, $\hat{M}$ is $\mathrm{diag}(p,q)$-equivariant. As the transformation was arbitrary, $\hat{M}$ has rotational symmetry.

Finally, we compute the instanton. Let $\{e_1,\ldots,e_k\}$ be the standard basis for $\mathbb{H}^k$. For $l\in\{1,\ldots,k\}$, let $v_l(x):=\frac{1}{\sqrt{\alpha_l^2+|x|^2}}\begin{bmatrix}
x^\dagger e_l \\ \alpha_l e_l
\end{bmatrix}$. Then let $V(x):=\begin{bmatrix}
v_1(x) & \cdots & v_k(x)
\end{bmatrix}$. We see that $V(x)^\dagger \Delta(x)=0$ and $V(x)^\dagger V(x)=I_{k}$. Therefore, we can use $V$ to construct our instanton. Simplifying $\mathbb{A}=V(x)^\dagger dV(x)$, we obtain the desired form for $\mathbb{A}$.
\end{proof}

\begin{note}
Proposition~\ref{prop:M=0} tells us that $\mathrm{Sp}(k)$ instantons with $M=0$ are reducible, being formed from the class of $\mathrm{Sp}(1)$ instantons with rotational symmetry and $M=0$. 
\end{note}

We are interested in the trivial summands of $(\hat{V},\hat{\rho})$. 
\begin{lemma}
Given $m\geq n\geq 1$, consider $(V_m,\rho_m)\otimes_\mathbb{C} (V_n,\rho_n)\otimes_\mathbb{C} (V_3,\rho_3)$ and $(V_m,\rho_m)\otimes_\mathbb{C} (V_n,\rho_n)\otimes_\mathbb{C} (V_1,\rho_1)$. The former has a single trivial summand when $m=n\geq 2$ or $m=n+2$. The latter has a single trivial summand when $m=n\geq 1$. Otherwise, there are no trivial summands.\label{lemma:wheretrivsummands}
\end{lemma} 

\begin{proof}
The first statement follows from previous work on hyperbolic monopoles~\cite[Lemma~7]{lang_hyperbolic_2023}. The second statement follows from the Clebsch--Gordan decomposition, noting that the tensor product of some representation and $(V_1,\rho_1)$ is isomorphic to the original representation. 
\end{proof}

Given $A\in\mathrm{Mat}(k,k,\mathbb{H})$, let $A_0\in\mathrm{Mat}(k,k,\mathbb{R})$ and $\vec{A}\in\mathrm{Mat}(k,k,\mathfrak{sp}(1))$ such that $A=A_0+\vec{A}$. From the decomposition of $(\mathbb{H},\mathrm{ad})=(\mathfrak{sp}(1),\mathrm{ad})\oplus (\mathbb{R},0)$, we have that $\mathrm{span}(M_0)$ is an invariant subspace of $(V,\rho)\otimes_\mathbb{R} (V^*,\rho^*)\otimes_\mathbb{R} (\mathbb{R},0)$ and $\mathrm{span}(\vec{M})$ is an invariant subspace of $(V,\rho)\otimes_\mathbb{R}(V^*,\rho^*)\otimes_\mathbb{R} (\mathbb{R}^3,\varrho_3)$, each of which is acted on trivially. Therefore, the representation theory used to construct $\vec{M}$ is the same as that used for $M$ in previous work dealing with hyperbolic monopoles~\cite[\S4.1]{lang_hyperbolic_2023}. We now investigate the remaining trivial summands.
\begin{lemma}
Given $m\geq n\geq 1$, we have that $(\mathbb{R}^m,\varrho_m)\otimes_\mathbb{R} ((\mathbb{R}^n)^*,\varrho_n^*)\otimes_\mathbb{R} (\mathbb{R},0)$ has a single trivial summand when $m=n$, both odd, and four trivial summands when $m=n$, both divisible by four. Otherwise, there are no trivial summands.\label{lemma:howmanytrivsummands}
\end{lemma}

\begin{proof}
Note that if $m$ and $n$ are both odd, then the complexification of the real representation is isomorphic, as a complex representation, to $(V_m,\rho_m)\otimes_\mathbb{C} (V_n,\rho_n)\otimes_\mathbb{C} (V_1,\rho_1)$. By Lemma~\ref{lemma:wheretrivsummands}, we know that this tensor product has a single trivial summand when $m=n$ and none otherwise. Furthermore, as this tensor product only contains odd dimensional summands, the representation is real, meaning that the real decomposition is the same as the complex. Thus, we have proven the first part of the lemma.

Note that if $m$ and $n$ are both divisible by four, then the complexification of the real representation is isomorphic, as a complex representation, to $(V_{m/2},\rho_{m/2})^{\oplus 2}\otimes_\mathbb{C} (V_{n/2},\rho_{n/2})^{\oplus 2}\otimes_\mathbb{C} (V_1,\rho_1)$. By Lemma~\ref{lemma:wheretrivsummands}, we know that this product has four trivial summands when $m=n$ and none otherwise. Just as above, this decomposition contains only odd dimensional summands, so the representation is real, meaning that the real decomposition is the same as the complex. Thus, we have proven the seconds part of the lemma.

Finally, suppose that one of $m$ or $n$ is odd and the other divisible by four. Then, decomposing the complexification of the real representation as a complex representation and expanding the tensor product, we find that the tensor product contains only even dimensional summands, meaning there are no trivial summands, proving the lemma.
\end{proof}

We now identify the trivial summands discussed above.
\begin{lemma}
Given $n\in\mathbb{N}_+$ odd, $I_n$ spans the unique trivial summand of 
\begin{equation*}
(\mathbb{R}^n,\varrho_n)\otimes_\mathbb{R} ((\mathbb{R}^n)^*,\varrho_n^*)\otimes_\mathbb{R} (\mathbb{R},0).
\end{equation*}

Given $n\in\mathbb{N}_+$ divisible by four, there are four trivial summands of 
\begin{equation*}
(\mathbb{R}^n,\varrho_n)\otimes_\mathbb{R} ((\mathbb{R}^n)^*,\varrho_n^*)\otimes_\mathbb{R} (\mathbb{R},0).
\end{equation*} 
Let $Y_i:=\varrho_n(\upsilon_i)\in\mathfrak{so}(n)$ induce the irreducible real $n$-representation and $y_i:=\rho_{n/2}(\upsilon_i)\in\mathfrak{su}\left(\frac{n}{2}\right)$ induce the irreducible complex $\frac{n}{2}$-representation. Let $U\in\mathrm{SU}\left(\frac{n}{2}\right)$ such that 
\begin{equation*}
Y_i=U^\dagger \begin{bmatrix}
y_i & 0 \\ 0 & y_i
\end{bmatrix}U.
\end{equation*} 
The trivial summands are spanned by the following real matrices 
\begin{equation}
I_n, \quad 
U^\dagger \begin{bmatrix}
iI_{n/2} & 0 \\ 0 & -iI_{n/2} 
\end{bmatrix}U, \quad
U^\dagger \begin{bmatrix}
0 & I_{n/2} \\ -I_{n/2}, 0 
\end{bmatrix}U, \quad
U^\dagger \begin{bmatrix}
0 & iI_{n/2} \\ iI_{n/2} & 0 
\end{bmatrix}U.\label{eq:fourtrivsummands}
\end{equation}
\end{lemma}

\begin{note}
While the first matrix, $I_n$, in \eqref{eq:fourtrivsummands} is clearly symmetric, the remaining three matrices are antisymmetric. Indeed, as they are real, the transpose of these matrices agrees with their Hermitian conjugate, which we see results in multiplying the latter three matrices by $-1$.
\end{note}

\begin{proof}
A matrix $A$ is acted on trivially by $(\mathbb{R}^n,\varrho_n)\otimes_\mathbb{R} ((\mathbb{R}^n)^*,\varrho_n^*)\otimes_\mathbb{R} (\mathbb{R},0)$ if and only if $[\varrho_n(\upsilon),A]=0$ for all $\upsilon\in\mathfrak{sp}(1)$. The result follows from previous work on hyperbolic monopoles~\cite[Lemma~10]{lang_hyperbolic_2023}.
\end{proof}

\begin{note}
We have that $(R,0)$ is a trivial summand of $(V,\rho)\otimes_\mathbb{R} (V^*,\rho^*)$, as $R$ is acted on trivially by this representation. From the above work, we know exactly where $R$ must live. Moreover, this constraint helps us identify $L$, as $R$ requires the entirety of $\hat{M}$.
\end{note}

Now that we know exactly what spans the trivial summands of $(\hat{V},\hat{\rho})$, the following theorem tells us exactly what form $M$ takes if $\hat{M}$ has simple spherical symmetry.
\begin{theorem}[Simple Spherical Structure Theorem]\label{thm:simplestructure}
Let $\hat{M}\in\mathcal{M}_{n,k}$ have simple spherical symmetry. Theorem~\ref{thm:simplesphericalsym} tells us that $\hat{M}$ is generated by a real representation $(V,\rho)$ of $\mathfrak{sp}(1)$, which we can decompose as 
\begin{equation}
(V,\rho)\simeq \bigoplus_{a=1}^m (\mathbb{R}^{n_a},\varrho_{n_a}).\label{eq:rhodecompsimple}
\end{equation}
Without loss of generality, we may assume that $n_a\geq n_{a+1}$, for all $a\in\{1,\ldots,m-1\}$.

Let $Y_{i,a}:=\varrho_{n_a}(\upsilon_i)\in\mathfrak{so}(n_a)$ and
\begin{equation*}
Y_i:=\mathrm{diag}(Y_{i,1},\ldots,Y_{i,m})\in\mathfrak{so}(k).
\end{equation*}
Then $Y_i$ induces $(V,\rho)$. For $n_a$ is divisible by four, let $y_{i,a}:=\rho_{\frac{n_a}{2}}(\upsilon_i)\in\mathfrak{su}\left(\frac{n_a}{2}\right)$ induce the irreducible $\frac{n_a}{2}$ complex representation. Then there is some $U_{n_a}\in\mathrm{SU}(n_a)$ such that 
\begin{equation*}
Y_{i,a}=U_{n_a}^\dagger \begin{bmatrix}
y_{i,a} & 0 \\ 0 & y_{i,a}
\end{bmatrix}U_{n_a}.
\end{equation*}

Using the decomposition of $(V,\rho)$ given in \eqref{eq:rhodecompsimple}, we have $(M_\mu)_{ab}\in\mathbb{R}^{n_a}\otimes_\mathbb{R} (\mathbb{R}^{n_b})^*$ such that,
\begin{equation*}
M_\mu=\begin{bmatrix}
(M_\mu)_{11} & \cdots & (M_\mu)_{1m} \\
\vdots & \ddots & \vdots \\
(M_\mu)_{m1} & \cdots & (M_\mu)_{mm}
\end{bmatrix}.
\end{equation*}

Then, up to a $\rho$-invariant gauge and for all $a,b\in\{1,\ldots,m\}$ we have that $\vec{M}=M_1i+M_2j+M_3k$ is given by the Structure Theorem for hyperbolic monopoles~\cite[Theorem~3]{lang_hyperbolic_2023}. Additionally, up to the same $\rho$-invariant gauge and for all $a,b\in\{1,\ldots,m\}$, we have that 
\begin{itemize}
\item[(1)] if $a=b$, then $\exists \lambda_a\in\mathbb{R}$ such that
\begin{equation*}
(M_0)_{aa}=\lambda_aI_{n_a};
\end{equation*}
\item[(2)] if $a<b$ and $n_a=n_b$ are odd, then $\exists \lambda_{a,b}\in\mathbb{R}$ such that
\begin{equation*}
(M_0)_{ab}=\lambda_{a,b}I_{n_a}, \quad\textrm{and}\quad (M_0)_{ba}=\lambda_{a,b}I_{n_a};
\end{equation*}
\item[(3)] if $a<b$ and $n_a=n_b$ are divisible by four, then $\exists \kappa_{a,0},\kappa_{a,1},\kappa_{a,2},\kappa_{a,3}\in\mathbb{R}$ such that
\begin{align*}
(M_0)_{ab}=U_{n_a}^\dagger \begin{bmatrix}
(\kappa_{a,0}+\kappa_{a,1}i)I_{n_a/2} & (\kappa_{a,2}+\kappa_{a,3}i)I_{n_a/2} \\
(-\kappa_{a,2}+\kappa_{a,3}i)I_{n_a/2} & (\kappa_{a,0}-\kappa_{a,1}i)I_{n_a/2}
\end{bmatrix}U_{n_a}, \quad\textrm{and}\\
(M_0)_{ba}=U_{n_a}^\dagger \begin{bmatrix}
(\kappa_{a,0}-\kappa_{a,1}i)I_{n_a/2} & -\kappa_{a,2}-\kappa_{a,3}i)I_{n_a/2} \\
(\kappa_{a,2}-\kappa_{a,3}i)I_{n_a/2} & (\kappa_{a,0}+\kappa_{a,1}i)I_{n_a/2}
\end{bmatrix}U_{n_a};
\end{align*}
\item[(4)] otherwise, $(M_0)_{ab}=0$.
\end{itemize}

Conversely, if $M_0$ has the above form, $\vec{M}$ has the form given in the Structure Theorem for hyperbolic monopoles, for some real representation $(V,\rho)$, and $[\rho(\upsilon),R]=0$ for all $\upsilon\in\mathfrak{sp}(1)$, then $\hat{M}$ has simple spherical symmetry~\cite[Theorem~3]{lang_hyperbolic_2023}.
\end{theorem}

\begin{proof}
The proof follows the same structure as the proof of the Structure Theorem for hyperbolic monopoles~\cite[Theorem~3]{lang_hyperbolic_2023}.
\end{proof}

Examples of ADHM data with simple spherical symmetry are given in previous work, though there they are viewed through the lens of hyperbolic monopoles~\cite[Propositions~5,~6~\&~7]{lang_hyperbolic_2023}. Additionally, in Section~\ref{subsubsec:NotinMSset}, we construct a novel example of ADHM data with conformal superspherical symmetry, meaning this data has simple spherical symmetry.

\begin{note}
We can consider the constraints on the structure group of instantons with simple spherical symmetry just as we did for hyperbolic monopoles in previous work~\cite[\S4.3]{lang_hyperbolic_2023}. In this case, the top row of \eqref{eq:simplesphericalsymdiff} gives us the same induced representation as in previous work~\cite[Lemma~11]{lang_hyperbolic_2023}. Therefore, we do not retread the same ground here.\label{note:simplesphericalsymconstraint}
\end{note}

\subsubsection{Non-symmetric hyperbolic monopoles}

In this section, we discuss the connection between instantons with simple spherical symmetry and hyperbolic monopoles with integral mass and no continuous symmetries. 

Note that the simple spherical subgroup of $\mathrm{Sp}(2)$ contains $R_1$ as a subgroup, whereas the isoclinic and conformal spherical subgroups do not even contain a subgroup conjugate to $R_1$. Hence, we can associate to an instanton with simple spherical symmetry a hyperbolic monopole with integral mass. Unlike toral symmetry, simple spherical symmetry does not impart extra symmetry to hyperbolic monopoles. The lack of additional symmetry is because the rest of the simple spherical symmetry does not commute with the $R_1$ subgroup. 

Indeed, suppose that $\mathrm{diag}(p,p)$ descends to a symmetry of the hyperbolic space. As the circle action only acts on the $S^1$ component of $H^3\times S^1$, we have that at least on $S^4\setminus S^2$, $\mathrm{diag}(p,p)$ and the $R_1$ subgroup must commute. As $S^2\subseteq S^4$ is a codimension two subspace, and the commutator is a smooth map, $\mathrm{diag}(p,p)$ and $R_1$ must commute everywhere. Thus, we must have $[p,e^{i\theta}]=0$ for all $\theta\in\mathbb{R}$, which only happens when $p\in\mathbb{C}$. Thus, the only part of the simple spherical group that descends to a symmetry of $H^3\times S^1$ is the $R_1$ subgroup, which only acts on the $S^1$ component. 

Note that an instanton with simple spherical symmetry can possess more symmetry, in which case the hyperbolic monopole may have a continuous symmetry. In Section~\ref{subsubsec:SpherSymHyperMono}, we discuss the connection between instantons with conformal superspherical symmetry and hyperbolic monopoles with different symmetries. If the instanton has simple spherical symmetry and does not have conformal superspherical symmetry, then the corresponding hyperbolic monopole has no continuous symmetries. 

\subsubsection{Hyperbolic analogue to Higgs bundles}\label{subsubsec:HyperbolicHiggsBundles}

In this section, we discuss the connection between instantons with simple spherical symmetry and a hyperbolic analogue to Higgs bundles. 

Just as Euclidean monopoles, Higgs bundles are a dimensional reduction of instantons. However, while Euclidean monopoles are invariant under translation along one axis, Higgs bundles are invariant under translations along two axes. In Section~\ref{subsubsec:HyperbolicMonopoles}, we cover the connection between hyperbolic monopoles and circle-invariant instantons. This connection is due to the conformal equivalence $S^4\setminus S^2\equiv H^3\times S^1$. There is a whole family of these conformal equivalences. Indeed, for $d\in\{0,1,2\}$ we have $S^4\setminus S^d\equiv H^{d+1}\times S^{4-d-1}$. Note that we can replace $S^4\setminus S^d$ with $\mathbb{R}^4\setminus \mathbb{R}^d$. The hyperbolic monopole case is when $d=2$. Here, we investigate the case $d=1$. In Section~\ref{subsubsec:HyperbolicNahm}, we study the case $d=0$.

While searching for multi-instanton solutions, Witten searched for $\mathrm{Sp}(1)$-instantons that have simple spherical symmetry, in our language~\cite{Witten_exact_1977}. In doing so, Witten found that the symmetric $\mathrm{Sp}(1)$-instantons he was searching for correspond to vortices in the abelian Higgs model on the hyperbolic plane. Moreover, these vortices can be found by solving the Liouville equation. By embedding these vortices into three dimensions, Maldonado was able to generate hyperbolic monopoles~\cite{maldonado_hyperbolic_2017}. Here, we examine such symmetric instantons with higher rank structure group, where we do not necessarily obtain the abelian Higgs model.

Note that the fixed point set of the simple spherical action on $\mathbb{R}^4$ is $\mathbb{R}$. Moreover, for $x\notin \mathbb{R}$, we have that the orbit of $x$ under the simple spherical action is $S^2$. Indeed, given $x=x_0+\vec{x}$, where $x_0\in\mathbb{R}$ and $\vec{x}\in\mathfrak{sp}(1)$, we note that the simple spherical action of $p\in\mathrm{Sp}(1)$ on $x$ corresponds to fixing $x_0$ and rotating $\vec{x}$ by the element in $\mathrm{SO}(3)$ given by the double cover $\mathrm{Sp}(1)\rightarrow\mathrm{SO}(3)$. 

On $\mathbb{R}^4\setminus\mathbb{R}$, we can use the coordinates $(x_0,r,\theta,\phi)$, where the $\mathfrak{sp}(1)$ part is written in spherical coordinates. Doing so, the metric on $\mathbb{R}^4\setminus\mathbb{R}$ is given by
\begin{equation*}
ds^2=dx_0^2+dr^2+r^2(d\theta^2+\sin^2\theta d\phi^2).
\end{equation*}
As $r\neq 0$ on this set, we can divide by $r^2$, obtaining the metric of $H^2\times S^2$. Specifically, we note that we obtain the upper-half space model of the hyperbolic plane.

Therefore, just as in the hyperbolic monopole case, an instanton with simple spherical symmetry corresponds to a hyperbolic analogue to a Higgs bundle and vice versa. 

\subsection{Isoclinic spherical symmetry}\label{subsec:IsoclinicSpherical}

In this section, we find an equation describing all instantons with isoclinic spherical symmetry, as given in Table~\ref{table:conformalsubgroups} as well as below. We also discuss the connections between these instantons and singular monopoles with no continuous symmetries and a hyperbolic analogue to Nahm data.

First, we introduce the notion of isoclinic spherical symmetry. An instanton is said to have isoclinic spherical symmetry if it is equivariant under $\mathrm{diag}(p,1)$ for all $p\in\mathrm{Sp}(1)$.
\begin{theorem}
Let $\hat{M}\in\mathcal{M}_{n,k}$. Then $\hat{M}$ has isoclinic spherical symmetry if and only if there exists a real representation $\rho\colon\mathfrak{sp}(1)\rightarrow\mathfrak{so}(k)$ such that for all $\upsilon\in\mathfrak{sp}(1)$,
\label{thm:isoclinicsphericalsym}
\begin{align}
\upsilon M+[\rho(\upsilon),M]=0,\label{eq:isoclinicsphericalsym1}\\
[\rho(\upsilon),R]=0.\label{eq:isoclinicsphericalsym2}
\end{align}
\end{theorem}

\begin{definition}
We call $\rho$ the \textbf{generating representation} of the isoclinic spherical symmetry of $\hat{M}$.
\end{definition}

\begin{proof}
We use the same notation as introduced in the beginning of the proof of Theorem~\ref{thm:circularsym}. We follow the proof of Theorem~\ref{thm:mainthm} after setting the scene.

Suppose that $\hat{M}$ has isoclinic spherical symmetry. Let $S\subseteq \mathrm{Sp}(1)\times\mathrm{O}(k)$ be the stabilizer group of $(\hat{M},U)$ restricted to rotations in $\mathrm{Sp}(1)\times\{1\}\subseteq\mathrm{Sp}(2)$. That is
\begin{equation}
S:=\{(\mathrm{diag}(p,1),K)\mid (\mathrm{diag}(LKL^\dagger (LL^\dagger)^{-1},p K),K).\mathrm{diag}(p^\dagger,1).(\hat{M},U)
=(\hat{M},U)\}.\label{eq:isoclinicsphericalsymS}
\end{equation}
Indeed, if $(Q,K)\in\mathrm{Sp}(n+k)\times\mathrm{GL}(k,\mathbb{R})$ such that $\mathrm{diag}(p^\dagger,1).(Q,K).(\hat{M},U)=(\hat{M},U)$, then Corollary~\ref{cor:IsomS} tells us that $K\in\mathrm{O}(k)$, $[R,K]=0$, and $Q=\mathrm{diag}(LKL^\dagger (LL^\dagger)^{-1},p K)$. Therefore, the pairs in $S$ encapsulate all the isoclinic spherical symmetry of the instanton. That $S$ is a group follows from this fact as well as the fact that the conformal action is a right Lie group action, just as in the proof of Theorem~\ref{thm:simplesphericalsym}.

Just like our previous symmetries, other than circular $t$-symmetry, the stabilizer group is a subgroup of a compact group, so we proceed as in the proof of Theorem~\ref{thm:mainthm}~\cite[Theorem~1.1]{lang_moduli_2024}. In particular, we find that $\hat{M}$ has isoclinic spherical symmetry if and only if there is a Lie algebra homomorphism $\rho\colon\mathfrak{sp}(1)\rightarrow\mathfrak{so}(k)$ such that for all $\upsilon\in\mathfrak{sp}(1)$, we have
\begin{equation}
\left(\begin{bmatrix}
L\rho(\upsilon)L^\dagger(LL^\dagger)^{-1}L-L\rho(\upsilon) \\
\upsilon M+[\rho(\upsilon),M]
\end{bmatrix},0\right)=(0,0).\label{eq:isoclinicsphericalsymdiff}
\end{equation}

We can simplify these constraints. In particular, suppose $\hat{M}$ has isoclinic spherical symmetry. In Section~\ref{subsubsec:isocliniccontraint}, we study the top row. For now, focusing on the bottom row, we see that $\upsilon M+[\rho(\upsilon),M]=0$. Furthermore, as mentioned above, Corollary~\ref{cor:IsomS} tells us that $[R,e^{\theta\rho(\upsilon)}]=0$, for all $\theta\in\mathbb{R}$. Differentiating and evaluating at $\theta=0$, we have $[R,\rho(\upsilon)]=0$. 

We can use the same method as in the proof of Theorem~\ref{thm:circularsym} to prove the converse.
\end{proof}

Just as with our previous symmetries, we need not check the final condition of $\mathcal{M}_{n,k}$ in Definition~\ref{def:Mstd} everywhere.
\begin{lemma}
Suppose that $\hat{M}$ satisfies \eqref{eq:isoclinicsphericalsym1} and \eqref{eq:isoclinicsphericalsym2} for some real representation $(\mathbb{R}^k,\rho)$ as well as the first three conditions of Definition~\ref{def:Mstd}. If the final condition is satisfied at all $x=x_0\in\mathbb{H}$ with $x_0\geq 0$, then $\hat{M}\in\mathcal{M}_{n,k}$. 
\end{lemma}

\begin{proof}
Let $z=z_0+z_1i+z_2j+z_3k\in \mathbb{H}$. There is some $p\in\mathrm{Sp}(1)$ such that $z=|z|p^\dagger$. Thus, $x:=pz=|z|$. The rest follows as in the proof of Lemma~\ref{lemma:finalconditioncirc}.
\end{proof}

\subsubsection{Structure of isoclinic spherical symmetry}

Theorem~\ref{thm:isoclinicsphericalsym} tells us exactly how to search for instantons with isoclinic spherical symmetry. It follows the same story as the case of simple spherical symmetry, just with a different equation. In this section, we investigate what representations generate instantons with isoclinic spherical symmetry and what the corresponding ADHM data looks like.

\begin{definition}
Let $\iota\colon\mathfrak{sp}(1)\rightarrow\mathfrak{gl}(\mathbb{H})$ be the inclusion map $\iota(\upsilon)(x):=\upsilon x$. Consider the real representation $(\mathbb{H},\iota)$, where the scalars are restricted from $\mathbb{H}$ to $\mathbb{R}$, and let $V:=\mathbb{R}^k$. Given a real representation $(V,\rho)$ of $\mathfrak{sp}(1)$, we define the induced real representation 
\begin{equation}
(\hat{V},\hat{\rho}):=(V,\rho)\otimes_\mathbb{R} (V^*,\rho^*)\otimes_\mathbb{R} (\mathbb{H},\iota).
\end{equation}
Unravelling how $\hat{\rho}$ acts on $\hat{V}=\mathrm{Mat}(k,k,\mathbb{H})$, let $\upsilon\in\mathfrak{sp}(1)$ and $A\in\mathrm{Mat}(k,k,\mathbb{H})$. Then
\begin{equation}
\hat{\rho}(\upsilon)(A)=[\rho(\upsilon),A]+\upsilon A.
\end{equation}
\end{definition}

\begin{note}
By keeping $(\mathbb{H},\iota)$ as a quaternionic representation, we could have considered $(\hat{V},\hat{\rho})$ as a quaternionic representation. The choice to restrict scalars merely simplifies computation of ADHM data.
\end{note}

The definition of $(\hat{V},\hat{\rho})$ is well-motivated. Firstly, note that $M\in\hat{V}$. Secondly, given the connection between the action of $\hat{\rho}(\upsilon)$ and \eqref{eq:isoclinicsphericalsymdiff}, we immediately obtain the following corollary.
\begin{cor}
Let $\hat{M}\in\mathcal{M}_{n,k}$. Then $\hat{M}$ has isoclinic spherical symmetry if and only if there is some real $k$-representation $(V,\rho)$ such that $\hat{\rho}(\upsilon)(M)=0$ and $[\rho(\upsilon),R]=0$ for all $\upsilon\in\mathfrak{sp}(1)$.
\end{cor}

\begin{note}
Just as in the simple spherical symmetry case, we have $(\mathrm{span}(M),0)$ is a summand of $(\hat{V},\hat{\rho})$. As we have dealt with the case $M=0$, we must have that $(\hat{V},\hat{\rho})$ has trivial summands and $M$ lives in their direct sum.
\end{note}

\begin{note}
Restricting scalars to $\mathbb{R}$, $(\mathbb{H},\iota)$ is isomorphic to $(\mathbb{R}^4,\varrho_4)$. Indeed, one can examine how $\upsilon\in\mathfrak{sp}(1)$ acts on the standard real basis of $\mathbb{H}$ and compute the Casimir operator $C_\iota$ of this representation, obtaining $C_\iota=\frac{3}{4}I_4$. Additionally, the complexification of $(\mathbb{R}^4,\varrho_4)$ is isomorphic to $(V_2,\rho_2)^{\oplus 2}$.
\end{note}

We now look at where the trivial summands of $(\hat{V},\hat{\rho})$ are found. First, we examine the tensor product of the relevant complex representations.
\begin{lemma}
Given $m\geq n\geq 1$, consider $(V_m,\rho_m)\otimes_\mathbb{C} (V_n,\rho_n)\otimes_\mathbb{C} (V_2,\rho_2)$. This product has a single trivial summand when $m=n+1$ and none otherwise.\label{lemma:complextensor2}
\end{lemma}

\begin{proof}
Note that $(V_a,\rho_a)\otimes_\mathbb{C} (V_2,\rho_2)$ contains a single trivial summand when $a=2$ and none otherwise. Thus, the complete tensor product contains a single trivial summand for each $(V_2,\rho_2)$ summand of $(V_m,\rho_m)\otimes_\mathbb{C} (V_n,\rho_n)$. However, this product contains a single $(V_2,\rho_2)$ summand when $m=n+1$ and none otherwise.
\end{proof}

Now we are equipped to examine the tensor product of the relevant real representations.
\begin{lemma}
Given $m\geq n\geq 1$, we have that $(\mathbb{R}^m,\varrho_m)\otimes_\mathbb{R} ((\mathbb{R}^n)^*,\varrho_n^*)\otimes_\mathbb{R} (\mathbb{R}^4,\varrho_4)$ has four trivial summands when $m$ is divisible by four, $n$ is odd, and $m=2n\pm 2$. Otherwise, there are no trivial summands.
\end{lemma}

\begin{proof}
Suppose that $m$ is divisible by four and $n$ is odd. The complexification of the real representation is isomorphic, as a complex representation, to 
\begin{equation*}
(V_{m/2},\rho_{m/2})^{\oplus 2}\otimes_\mathbb{C} (V_n,\rho_n)\otimes_\mathbb{C} (V_2,\rho_2)^{\oplus 2}.
\end{equation*}
By Lemma~\ref{lemma:complextensor2}, we have that there are $4$ trivial summands when $m=2n\pm 2$ and none otherwise. Furthermore, as the tensor product only contains odd dimensional summands, the representation is real, meaning that the real decomposition is the same as the complex. Thus, we have proven the first part of the lemma.

Suppose that $m$ and $n$ are divisible by four. The complexification of the real representation is isomorphic, as a complex representation, to 
\begin{equation*}
(V_{m/2},\rho_{m/2})^{\oplus 2}\otimes_\mathbb{C} (V_{n/2},\rho_{n/2})^{\oplus 2}\otimes_\mathbb{C} (V_2,\rho_2)^{\oplus 2}.
\end{equation*}
As $m/2$ and $n/2$ are both even, by Lemma~\ref{lemma:complextensor2}, this product has no trivial summands. 

Suppose that $m$ and $n$ are both odd. The complexification of the real representation is isomorphic, as a complex representation, to
\begin{equation*}
(V_m,\rho_m)\otimes_\mathbb{C} (V_n,\rho_n)\otimes_\mathbb{C} (V_2,\rho_2)^{\oplus 2}.
\end{equation*}
As $m$ and $n$ are both odd, by Lemma~\ref{lemma:complextensor2}, this product has no trivial summands.

Finally, suppose $m$ is odd and $n$ is divisible by four. The complexification of the real representation is isomorphic, as a complex representation, to
\begin{equation*}
(V_m,\rho_m)\otimes_\mathbb{C} (V_{n/2},\rho_{n/2})^{\oplus 2}\otimes_\mathbb{C} (V_2,\rho_2)^{\oplus 2}.
\end{equation*}
By Lemma~\ref{lemma:complextensor2}, we have that there are $4$ trivial summands when $n=2m\pm 2$. As $m\geq n\geq 1$, we have $m\geq 2m\pm 2\geq 1$. This set of inequalities has no solution. Thus, there are no trivial summands, proving the lemma. 
\end{proof}

We now identify the trivial summands discussed above.
\begin{definition}
Let $C^{n}$ be the unit length generator of the unique one-dimensional representation in $(V_{n+1},\rho_{n+1})\otimes_\mathbb{C} (V_n^*,\rho_n^*)\otimes_\mathbb{C} (V_2,\rho_2)$. This generator is well-defined, up to a $\mathbb{C}^*$ factor, though a choice of scale leaves a $S^1$ factor. This generator can be thought of as a $\mathfrak{sp}(1)$-invariant pair $(C_0^n,C_1^n)$ of complex $(n+1)\times n$ matrices.
\end{definition}
These generators can be identified in much the same way as we compute the $B^n$ matrices in previous work~\cite[Appendix~E]{lang_thesis_2024}.

Above, we defined what spans the trivial summands of $(\hat{V},\hat{\rho})$ when decomposed as a complex representation. We use this definition to determine what spans the real representation.
\begin{lemma}
Given $n\in\mathbb{N}_+$ odd, let $Y_i^+:=\varrho_{2n+ 2}(\upsilon_i)\in\mathfrak{so}(2n+ 2)$ and $Y_i^-:=\varrho_n(\upsilon_i)\in\mathfrak{so}(n)$. Let $y_i^+:=\rho_{n+ 1}(\upsilon_i)\in\mathfrak{su}(n+ 1)$. As the complexification of $(\mathbb{R}^{2n+ 2},\varrho_{2n+ 2})$ is isomorphic to $(V_{n+ 1},\rho_{n+ 1})^{\oplus 2}$, there exists $U_+\in\mathrm{SU}(2n+ 2)$ such that
\begin{equation*}
Y_i^+=U_+^\dagger \begin{bmatrix}
y_i^+ & 0 \\ 0 & y_i^+
\end{bmatrix}U_+.
\end{equation*}
We can find four linearly independent quadruples of real matrices spanning the four trivial summands of $(\mathbb{R}^{2n+ 2},\varrho_{2n+ 2})\otimes_\mathbb{R} ((\mathbb{R}^n)^*,\varrho_n^*)\otimes_\mathbb{R} (\mathbb{R}^4,\varrho_4)$. For some choices of $\alpha,\beta,\gamma,\delta\in\mathbb{C}$, these linearly independent quadruples of matrices are given by
\begin{multline}
\left(U_+^\dagger \begin{bmatrix}
\alpha C_0^{n}-i\gamma C_1^n \\
\beta C_0^{n}-i\delta C_1^n
\end{bmatrix},
U_+^\dagger \begin{bmatrix}
-i\alpha C_0^{n}+\gamma C_1^n \\
-i\beta C_0^{n}+\delta C_1^n
\end{bmatrix},\right.\\
\left.U_+^\dagger \begin{bmatrix}
-\alpha C_1^{n}-i\gamma C_0^n \\
-\beta C_1^{n}-i\delta C_0^n
\end{bmatrix},
U_+^\dagger \begin{bmatrix}
-i\alpha C_1^{n}-\gamma C_0^n \\
-i\beta C_1^{n}-\delta C_0^n
\end{bmatrix}\right)\label{eq:Dmu}
\end{multline}

Given $n\in\mathbb{N}_+$ odd and at least three, let $Y_i^+:=\varrho_{2n- 2}(\upsilon_i)\in\mathfrak{so}(2n- 2)$ and $Y_i^-:=\varrho_n(\upsilon_i)\in\mathfrak{so}(n)$. Let $y_i^+:=\rho_{n- 1}(\upsilon_i)\in\mathfrak{su}(n-1)$. As the complexification of $(\mathbb{R}^{2n-2},\varrho_{2n- 2})$ is isomorphic to $(V_{n-1},\rho_{n-1})^{\oplus 2}$, there exists $U_+\in\mathrm{SU}(2n-2)$ such that
\begin{equation*}
Y_i^+=U_+^\dagger \begin{bmatrix}
y_i^+ & 0 \\ 0 & y_i^+
\end{bmatrix}U_+.
\end{equation*}
We can find four linearly independent quadruples of real matrices spanning the four trivial summands of $(\mathbb{R}^{2n-2},\varrho_{2n-2})\otimes_\mathbb{R} ((\mathbb{R}^n)^*,\varrho_n^*)\otimes_\mathbb{R} (\mathbb{R}^4,\varrho_4)$. For some choices of $\alpha,\beta,\gamma,\delta\in\mathbb{C}$, these linearly independent quadruples of matrices are given by
\begin{multline}
\left(U_+^\dagger \begin{bmatrix}
\alpha (C_1^{n-1})^\dagger+i\gamma (C_0^{n-1})^\dagger \\
\beta (C_1^{n-1})^\dagger+i\delta (C_1^{n-1})^\dagger
\end{bmatrix},
U_+^\dagger \begin{bmatrix}
-i\alpha (C_1^{n-1})^\dagger-\gamma (C_0^{n-1})^\dagger \\
-i\beta (C_1^{n-1})^\dagger-\delta (C_0^{n-1})^\dagger
\end{bmatrix},\right.\\\left.
U_+^\dagger \begin{bmatrix}
\alpha (C_0^{n-1})^\dagger-i\gamma (C_1^{n-1})^\dagger \\
\beta (C_0^{n-1})^\dagger-i\delta (C_1^{n-1})^\dagger
\end{bmatrix},
U_+^\dagger \begin{bmatrix}
i\alpha (C_0^{n-1})^\dagger-\gamma (C_1^{n-1})^\dagger \\
i\beta (C_0^{n-1})^\dagger-\delta (C_1^{n-1})^\dagger
\end{bmatrix}\right)\label{eq:Dmu2}
\end{multline}
\end{lemma}

\begin{proof}
Consider the standard Lie algebra isomorphism $\rho_2\colon\mathfrak{sp}(1)\rightarrow \mathfrak{su}(2)$ defined by 
\begin{equation*}
\rho_2\left(\frac{i}{2}\right):=\frac{1}{2}\begin{bmatrix}
i & 0 \\ 0 & -i
\end{bmatrix}, \quad
\rho_2\left(\frac{j}{2}\right):=\frac{1}{2}\begin{bmatrix}
0 & 1 \\ -1 & 0
\end{bmatrix}, \quad
\rho_2\left(\frac{k}{2}\right):=\frac{1}{2}\begin{bmatrix}
0 & i \\ i & 0
\end{bmatrix}.
\end{equation*}
Note that $\left(\mathbb{C}^2,\rho_2\right)$ is the irreducible complex 2-representation of $\mathfrak{sp}(1)$. Additionally, consider the Lie algebra homomorphism $\varrho_4\colon\mathfrak{sp}(1)\rightarrow\mathfrak{so}(4)$ defined by
\begin{equation*}
\begin{aligned}
\varrho_4\left(\frac{i}{2}\right)&:=\frac{1}{2}\begin{bmatrix}
0 & -1 & 0 & 0 \\
1 & 0 & 0 & 0 \\
0 & 0 & 0 & -1 \\
0 & 0 & 1 & 0 
\end{bmatrix},\\
\varrho_4\left(\frac{j}{2}\right)&:=\frac{1}{2}\begin{bmatrix}
0 & 0 & -1 & 0 \\
0 & 0 & 0 & 1 \\
1 & 0 & 0 & 0 \\
0 & -1 & 0 & 0 
\end{bmatrix},\\
\varrho_4\left(\frac{k}{2}\right)&:=\frac{1}{2}\begin{bmatrix}
0 & 0 & 0 & -1 \\
0 & 0 & -1 & 0 \\
0 & 1 & 0 & 0 \\
1 & 0 & 0 & 0 
\end{bmatrix}.
\end{aligned}
\end{equation*}
Note that $\left(\mathbb{R}^4,\varrho_4\right)$ is the irreducible real 4-representation of $\mathfrak{sp}(1)$. Recall that the complexification of $(\mathbb{R}^4,\varrho_4)$ is isomorphic, as a complex representation, to $(V_2,\rho_2)^{\oplus 2}$. That is, there is some $U\in\mathrm{SU}(4)$ such that for all $\upsilon\in\mathfrak{sp}(1)$, $\varrho_4(\upsilon)=U^\dagger \mathrm{diag}(\rho_2(\upsilon),\rho_2(\upsilon))U$. Given the above definitions for $\rho_2$ and $\varrho_4$, one such $U$ is given by
\begin{equation*}
U:=\frac{1}{\sqrt{2}}\begin{bmatrix}
1 & i & 0 & 0 \\
0 & 0 & -1 & i \\
0 & 0 & i & -1 \\
i & 1 & 0 & 0
\end{bmatrix}.
\end{equation*}

Consider the first part of the lemma. We are searching for quadruples $(D_\mu)_{\mu=0}^3$ such that for all $\upsilon\in\mathfrak{sp}(1)$, 
\begin{equation*}
\left(\varrho_{2n+2}(\upsilon)\otimes I_4\right)(D_\mu)_{\mu=0}^3-(D_\mu)_{\mu=0}^3\left(\varrho_n(\upsilon)\otimes I_4\right)+\left(I_{2n+2}\otimes \varrho_4\right)(D_\mu)_{\mu=0}^3=0.
\end{equation*}
We have that 
\begin{align*}
\varrho_{2n+2}(\upsilon)&=U_+^\dagger \mathrm{diag}(\rho_{n+1}(\upsilon),\rho_{n+1}(\upsilon)) U_+,\\
\varrho_n(\upsilon)&=\rho_n(\upsilon), \quad\textrm{and}\\
\varrho_4(\upsilon)&=U^\dagger \mathrm{diag}(\rho_2(\upsilon),\rho_2(\upsilon))U.
\end{align*} 
Letting $(D'_\mu)_{\mu=0}^3:=(U_+ U_{\mu\nu}D_\nu)_{\mu=0}^3$, we have that
\begin{multline*}
\left(\mathrm{diag}(\rho_{n+1}(\upsilon),\rho_{n+1}(\upsilon))\otimes I_4\right) (D'_\mu)_{\mu=0}^3 -(D'_\mu)_{\mu=0}^3 \left(\rho_n(\upsilon)\otimes I_4\right)\\
+\left(I_{2n+2}\otimes \mathrm{diag}(\rho_2(\upsilon),\rho_2(\upsilon))\right)(D'_\mu)_{\mu=0}^3=0.
\end{multline*}
Note that $(D'_\mu)_{\mu=0}^3$ is acted trivially upon by $(V_{n+1},\rho_{n+1})^{\oplus 2}\otimes_\mathbb{C} (V_n^*,\rho_n^*)\otimes_\mathbb{C} (V_2,\rho_2)^{\oplus 2}$. Writing $D'_\mu:=\begin{bmatrix}
A_\mu \\ B_\mu
\end{bmatrix}$, we see that $(A_0,A_1)$, $(A_2,A_3)$, $(B_0,B_1)$, and $(B_2,B_3)$ are acted trivially upon by 
\begin{equation*}
(V_{n+1},\rho_{n+1})\otimes_\mathbb{C} (V_n^*,\rho_n^*)\otimes_\mathbb{C} (V_2,\rho_2).
\end{equation*} 
That is, all of these pairs belong to the span of $(C_0^n,C_1^n)$. Hence, there is some $\alpha,\beta,\gamma,\delta\in\mathbb{C}$ such that 
\begin{align*}
(A_0,A_1)&=\sqrt{2}(\alpha C_0^n,\alpha C_1^n),\\ 
(A_2,A_3)&=\sqrt{2}(\gamma C_0^n,\gamma C_1^n),\\
(B_0,B_1)&=\sqrt{2}(\beta C_0^n, \beta C_1^n), \quad\textrm{and}\\
(B_2,B_3)&=\sqrt{2}(\delta C_0^n,\delta C_1^n).
\end{align*}

Given the definition of $D'_\mu$, we note that $D_\mu=U_+^\dagger (U^\dagger)_{\mu\nu}D'_\nu$. Substituting the values for $D'_\nu$, we find that $(D_\mu)_{\mu=0}^3$ is given by \eqref{eq:Dmu}. In particular, for some choices of $\alpha,\beta,\gamma,\delta$, we obtain four linearly independent quadruples of real matrices, as the space of quadruples of real matrices acted on trivially by $(\mathbb{R}^{2n+2},\varrho_{2n+2})\otimes_\mathbb{R} ((\mathbb{R}^n)^*,\varrho_n^*)\otimes_\mathbb{R} (\mathbb{R}^4,\varrho_4)$ is four dimensional.

Consider the second part of the lemma. We are searching for quadruples $(D_\mu)_{\mu=0}^3$ such that for all $\upsilon\in\mathfrak{sp}(1)$, 
\begin{equation*}
\left(\varrho_{2n-2}(\upsilon)\otimes I_4\right)(D_\mu)_{\mu=0}^3-(D_\mu)_{\mu=0}^3\left(\varrho_n(\upsilon)\otimes I_4\right)+\left(I_{2n-2}\otimes \varrho_4\right)(D_\mu)_{\mu=0}^3=0.
\end{equation*}
We have that 
\begin{align*}
\varrho_{2n-2}(\upsilon)&=U_+^\dagger \mathrm{diag}(\rho_{n-1}(\upsilon),\rho_{n-1}(\upsilon)) U_+,\\
\varrho_n(\upsilon)&=\rho_n(\upsilon), \quad\textrm{and}\\
\varrho_4(\upsilon)&=U^\dagger \mathrm{diag}(\rho_2(\upsilon),\rho_2(\upsilon))U.
\end{align*} 
Letting $(D'_\mu)_{\mu=0}^3:=(U_+ U_{\mu\nu}D_\nu)_{\mu=0}^3$, we have that
\begin{multline*}
\left(\mathrm{diag}(\rho_{n-1}(\upsilon),\rho_{n-1}(\upsilon))\otimes I_4\right) (D'_\mu)_{\mu=0}^3 -(D'_\mu)_{\mu=0}^3 \left(\rho_n(\upsilon)\otimes I_4\right)\\
+\left(I_{2n-2}\otimes \mathrm{diag}(\rho_2(\upsilon),\rho_2(\upsilon))\right)(D'_\mu)_{\mu=0}^3=0.
\end{multline*}
Note that $(D'_\mu)_{\mu=0}^3$ is acted trivially upon by $(V_{n-1},\rho_{n-1})^{\oplus 2}\otimes_\mathbb{C} (V_n^*,\rho_n^*)\otimes_\mathbb{C} (V_2,\rho_2)^{\oplus 2}$. Writing $D'_\mu:=\begin{bmatrix}
A_\mu \\ B_\mu
\end{bmatrix}$, we see that $(A_0,A_1)$, $(A_2,A_3)$, $(B_0,B_1)$, and $(B_2,B_3)$ are acted trivially upon by 
\begin{equation*}
(V_{n-1},\rho_{n-1})\otimes_\mathbb{C} (V_n^*,\rho_n^*)\otimes_\mathbb{C} (V_2,\rho_2).
\end{equation*} 
We know that this tensor product has a single trivial summand. 

Recalling that $(C_0^{n-1},C_1^{n-1})$ spans the unique trivial summand of 
\begin{equation*}
(V_n,\rho_n)\otimes_\mathbb{C} (V_{n-1}^*,\rho_{n-1}^*)\otimes_\mathbb{C} (V_2,\rho_2),
\end{equation*} 
we find that $((C_1^{n-1})^\dagger,-(C_0^{n-1})^\dagger)$ spans the trivial summand of $(V_{n-1},\rho_{n-1})\otimes_\mathbb{C} (V_n^*,\rho_n^*)\otimes_\mathbb{C} (V_2,\rho_2)$. Hence, there is some $\alpha,\beta,\gamma,\delta\in\mathbb{C}$ such that 
\begin{align*}
(A_0,A_1)&=\sqrt{2}(\alpha (C_1^{n-1})^\dagger,-\alpha (C_0^{n-1})^\dagger),\\
(A_2,A_3)&=\sqrt{2}(\gamma (C_1^{n-1})^\dagger,-\gamma (C_0^{n-1})^\dagger),\\
(B_0,B_1)&=\sqrt{2}(\beta (C_1^{n-1})^\dagger, -\beta (C_0^{n-1})^\dagger),\quad\textrm{and}\\
(B_2,B_3)&=\sqrt{2}(\delta (C_1^{n-1})^\dagger,-\delta (C_0^{n-1})^\dagger).
\end{align*} 

Given the definition of $D'_\mu$, we note that $D_\mu=U_+^\dagger (U^\dagger)_{\mu\nu}D'_\nu$. Substituting the values for $D'_\nu$, we find that $(D_\mu)_{\mu=0}^3$ is given by \eqref{eq:Dmu2}. In particular, for some choices of $\alpha,\beta,\gamma,\delta$, we obtain four linearly independent quadruples of real matrices, as the space of quadruples of real matrices acted on trivially by $(\mathbb{R}^{2n+2},\varrho_{2n+2})\otimes_\mathbb{R} ((\mathbb{R}^n)^*,\varrho_n^*)\otimes_\mathbb{R} (\mathbb{R}^4,\varrho_4)$ is four dimensional.
\end{proof}

Now that we know exactly what spans the trivial summands of $(\hat{V},\hat{\rho})$, the following theorem tells us exactly what form $M$ takes if $\hat{M}$ has isoclinic spherical symmetry.
\begin{theorem}[Isoclinic Spherical Structure Theorem]
Let $\hat{M}\in\mathcal{M}_{n,k}$ have isoclinic spherical symmetry. Theorem~\ref{thm:isoclinicsphericalsym} tells us that $\hat{M}$ is generated by a real representation $(V,\rho)$ of $\mathfrak{sp}(1)$, which we can decompose as 
\begin{equation}
(V,\rho)\simeq \bigoplus_{a=1}^m (\mathbb{R}^{n_a},\varrho_{n_a}).\label{eq:rhodecompiso}
\end{equation}
Without loss of generality, we may assume that $n_a\geq n_{a+1}$, for all $a\in\{1,\ldots,m-1\}$.

Let $Y_{i,a}:=\varrho_{n_a}(\upsilon_i)\in\mathfrak{so}(n_a)$ and
\begin{equation*}
Y_i:=\mathrm{diag}(Y_{i,1},\ldots,Y_{i,m})\in\mathfrak{so}(k).
\end{equation*}
Then $Y_i$ induces $(V,\rho)$. For $n_a$ is divisible by four, let $y_{i,a}:=\rho_{\frac{n_a}{2}}(\upsilon_i)\in\mathfrak{su}\left(\frac{n_a}{2}\right)$ induce the irreducible $\frac{n_a}{2}$ complex representation. Then there is some $U_{n_a}\in\mathrm{SU}(n_a)$ such that 
\begin{equation*}
Y_{i,a}=U_{n_a}^\dagger \begin{bmatrix}
y_{i,a} & 0 \\ 0 & y_{i,a}
\end{bmatrix}U_{n_a}.
\end{equation*}

Using the decomposition of $(V,\rho)$ given in \eqref{eq:rhodecompiso}, we have $(M_\mu)_{ab}\in\mathbb{R}^{n_a}\otimes_\mathbb{R} (\mathbb{R}^{n_b})^*$ such that,
\begin{equation*}
M_\mu=\begin{bmatrix}
(M_\mu)_{11} & \cdots & (M_\mu)_{1m} \\
\vdots & \ddots & \vdots \\
(M_\mu)_{m1} & \cdots & (M_\mu)_{mm}
\end{bmatrix}.
\end{equation*}

Then, up to a $\rho$-invariant gauge and for all $a,b\in\{1,\ldots,m\}$ we have that,
\begin{itemize}
\item[(1)] if $a<b$ and $n_a=2n_b+2$, then $\exists \alpha_{a,b},\beta_{a,b},\gamma_{a,b},\delta_{a,b}\in\mathbb{C}$ such that $(M_\mu)_{ab}$ take the form of \eqref{eq:Dmu} and are real and $(M_\mu)_{ba}=(M_\mu)_{ab}^T$;
\item[(2)] if $a<b$ and $n_a=2n_b-2$, then $\exists \alpha_{a,b},\beta_{a,b},\gamma_{a,b},\delta_{a,b}\in\mathbb{C}$ such that $(M_\mu)_{ab}$ take the form of \eqref{eq:Dmu2} and are real and $(M_\mu)_{ba}=(M_\mu)_{ab}^T$;
\item[(3)] otherwise, $(M_\mu)_{ab}=0$.
\end{itemize}

Conversely, if $M$ has the above form for some real representation $(V,\rho)$ and $[\rho(\upsilon),R]=0$ for all $\upsilon\in\mathfrak{sp}(1)$, then $\hat{M}$ has isoclinic spherical symmetry.
\end{theorem}

\begin{proof}
The proof follows the same structure as the proof of the Structure Theorem for hyperbolic monopoles~\cite[Theorem~3]{lang_hyperbolic_2023}.
\end{proof}

\subsubsection{Novel example of isoclinic spherical symmetry}\label{subsubsec:noveliso}

In this section, we identify novel examples of isoclinic spherical symmetry, using the Isoclinic Spherical Structure Theorem. The first representation that comes to mind to consider is $(\mathbb{R}^4,\varrho_4)\oplus (\mathbb{R},0)$. This representation generates an instanton with rotational symmetry and is discussed in Proposition~\ref{prop:RotInstEx}. Thus, we consider the next simplest case: $(\mathbb{R}^3,\varrho_3)\oplus (\mathbb{R}^4,\varrho_4)$.
\begin{prop}\label{prop:IsoEx}
Let $\lambda>0$. Then define $\hat{M}$ by 
\begin{equation}
\begin{aligned}
M&:=\lambda\begin{bmatrix}
0 & 0 & 0 & 0 & 1 & -i & j \\
0 & 0 & 0 & 0 & i & 1 & -k \\
0 & 0 & 0 & 0 & j & -k & -1 \\
0 & 0 & 0 & 0 & k & j & i \\
1 & i & j & k & 0 & 0 & 0 \\
-i & 1 & -k & j & 0 & 0 & 0 \\
j & -k & -1 & i & 0 & 0 & 0 
\end{bmatrix}\quad\textrm{and}\\
L&:=\lambda\begin{bmatrix}
3 & -i & -j & -k & 0 & 0 & 0 \\
0 & 2\sqrt{2} & k\sqrt{2} & -j\sqrt{2} & 0 & 0 & 0 \\
0 & 0 & \sqrt{6} & i\sqrt{6} & 0 & 0 & 0
\end{bmatrix}.
\end{aligned}
\end{equation}
We have $\hat{M}\in\mathcal{M}_{3,7}$ and it corresponds to a $\mathrm{Sp}(3)$ instanton with isoclinic spherical symmetry and instanton number $7$.
\end{prop}

\begin{proof}
Consider the representation $(V,\rho):=(\mathbb{R}^3,\varrho_3)\oplus (\mathbb{R}^4,\varrho_4)$. Seeking instantons with isoclinic spherical symmetry generated by $(V,\rho)$, the Isoclinic Spherical Structure Theorem tells us that $M$ has to have the form given in the statement, up to a $\rho$-invariant gauge. Thus, if $\hat{M}\in\mathcal{M}_{3,7}$, then it has isoclinic spherical symmetry. To that end, we have $LL^\dagger=12\lambda^2I_3$ and $R=12\lambda^2 I_4\oplus 4\lambda^2 I_3$. Due to the symmetry of the instanton, we need only check the final condition for $x=x_0\geq 0$. Doing so, we see that $\Delta(x)^\dagger\Delta(x)$ is given by
\begin{equation*}
\begin{bmatrix}
12\lambda^2+x_0^2 & 0 & 0 & 0 & -2x_0\lambda & 0 & 0 \\
0 & 12\lambda^2+x_0^2 & 0 & 0 & 0 & -2x_0\lambda & 0 \\
0 & 0 & 12\lambda^2+x_0^2 & 0 & 0 & 0 & 2x_0\lambda \\
0 & 0 & 0 & 12\lambda^2+x_0^2 & 0 & 0 & 0 \\
-2x_0\lambda & 0 & 0 & 0 & 4\lambda^2+x_0^2 & 0 & 0 \\
0 & -2x_0\lambda & 0 & 0 & 0 & 4\lambda^2+x_0^2 & 0 \\
0 & 0 & 2x_0\lambda & 0 & 0 & 0 & 4\lambda^2+x_0^2
\end{bmatrix}.
\end{equation*}
The eigenvalues of $\Delta(x)^\dagger\Delta(x)$ are $12\lambda^2+x_0^2$, $8\lambda^2+x_0^2\pm 2\lambda\sqrt{4\lambda^2+x_0^2}$ with multiplicity $1$, $3$, and $3$, respectively. Therefore, $\Delta(x)^\dagger\Delta(x)$ is positive-definite everywhere.
\end{proof}

\subsubsection{Constraint on isoclinic spherical symmetry}\label{subsubsec:isocliniccontraint}

In Theorem~\ref{thm:isoclinicsphericalsym}, we obtain a representation $(V,\rho)$ of $\mathfrak{sp}(1)$ from an instanton with isoclinic spherical symmetry by focusing on the bottom of \eqref{eq:isoclinicsphericalsymdiff}. This representation induces another $(\hat{V},\hat{\rho})$, which we use to determine what an instanton with isoclinic spherical symmetry looks like. However, by focusing on the top of \eqref{eq:isoclinicsphericalsymdiff}, we obtain a second representation.
\begin{lemma}
Let $\upsilon_1:=i/2$, $\upsilon_2:=j/2$, and $\upsilon_3:=k/2$. Suppose $\hat{M}\in\mathcal{M}_{n,k}$ has isoclinic spherical symmetry generated by $(V,\rho)$, a real representation. Let $y_i:=(LL^\dagger)^{-1}L\rho(\upsilon_i)L^\dagger\in\mathfrak{sp}(n)$. The $y_i$ induce a quaternionic representation of $\mathfrak{sp}(1)$, which we denote by $(W,\lambda)$, where $W:=\mathbb{H}^n$ and $\lambda$ is the linear map taking $\upsilon_i\mapsto y_i$.
\end{lemma}

\begin{proof}
Let $i,j\in\{1,2,3\}$. By Theorem~\ref{thm:isoclinicsphericalsym}, we know that $[\rho(\upsilon_i),M]=\upsilon_i M$, so
\begin{equation*}
\begin{aligned}
M^\dagger M\rho(\upsilon_i)&=M^\dagger [M,\rho(\upsilon_i)]+M^\dagger \rho(\upsilon_i)M\\
&=-M^\dagger \upsilon_i M+[M^\dagger ,\rho(\upsilon_i)]M+\rho(\upsilon_i)M^\dagger M\\
&=-M^\dagger \upsilon_i M+M^\dagger \upsilon_iM+\rho(\upsilon_i)M^\dagger M\\
&=\rho(\upsilon_i)M^\dagger M.
\end{aligned}
\end{equation*}
Thus, $[M^\dagger M,\rho(\upsilon_i)]=0$. Also, from Theorem~\ref{thm:isoclinicsphericalsym}, $[\rho(\upsilon_i),R]=0$. Hence, $[L^\dagger L,\rho(\upsilon_i)]=0$. Thus,
\begin{equation*}
LL^\dagger L\rho(\upsilon_i)L^\dagger=L\rho(\upsilon_i)L^\dagger LL^\dagger.
\end{equation*}
Thus, $[LL^\dagger,L\rho(\upsilon_i)L^\dagger]=0$, so $y_i\in\mathfrak{sp}(n)$. Therefore, we see
\begin{equation*}
\begin{aligned}
[y_i,y_j]&=(LL^\dagger)^{-1}L(\rho(\upsilon_i)L^\dagger L\rho(\upsilon_j)-\rho(\upsilon_j)L^\dagger L\rho(\upsilon_i))L^\dagger (LL^\dagger)^{-1}\\
&=(LL^\dagger)^{-1}L\rho([\upsilon_i,\upsilon_j])L^\dagger.
\end{aligned}
\end{equation*}
Thus, the $y_i$ satisfy the correct commutation relations.
\end{proof}

We use $(W,\lambda)$ to determine what $\mathrm{Sp}(n)$ structure groups are possible given $\rho$. Using $W=\mathbb{H}^n\simeq \mathbb{C}^{2n}$, we can restrict the scalars from $\mathbb{H}$ to $\mathbb{C}$. In doing so, the generators $y_i\in\mathfrak{sp}(n)$ correspond to elements of $\mathfrak{su}(2n)$. Thus, the induced complex representation is a $2n$-representation. Given an irreducible, quaternionic representation, the complex representation obtained by restricting scalars is either isomorphic to $(V_a,\rho_a)$ for some $a$ even or $(V_a,\rho_a)^{\oplus 2}$ for some $a$ odd. 
\begin{definition}
Given a real representation $(V,\rho)$ of $\mathfrak{sp}(1)\oplus\mathfrak{sp}(1)$, with induced representation $(W,\lambda)$ as defined above, we define the quaternionic representation 
\begin{equation}
(\hat{W},\hat{\lambda}):=(W,\lambda)\otimes_\mathbb{R} (V^*,\rho^*).
\end{equation}
\end{definition}

Unravelling how $\hat{\lambda}$ acts on $\hat{W}=\mathrm{Mat}(n,k,\mathbb{H})$, let $\upsilon\in\mathfrak{sp}(1)$ and $A\in\hat{W}$. Then
\begin{equation}
\hat{\lambda}(\upsilon)\left(A\right)=\lambda(\upsilon)A-A\rho(\upsilon).
\end{equation}

The definition of $(\hat{W},\hat{\lambda})$ is well-motivated. First, note that $L\in\hat{W}$. Moreover, the following lemma tells us exactly where $L$ lives in $\hat{W}$.
\begin{lemma}
Let $\hat{M}\in\mathcal{M}_{n,k}$ have isoclinic spherical symmetry. Let $(V,\rho)$, $(W,\lambda)$ and $(\hat{W},\hat{\lambda})$ be as above. Then $\hat{\lambda}(\upsilon)(L)=0$ for all $\upsilon\in\mathfrak{sp}(1)$.
\end{lemma}

\begin{proof}
As $\hat{M}$ has isoclinic symmetry, we have \eqref{eq:isoclinicsphericalsymdiff}. Then we see that for $\upsilon\in\mathfrak{sp}(1)$,
\begin{equation*}
\hat{\lambda}(\upsilon)(L)=\lambda(\upsilon)L-L\rho(\upsilon)=0.\qedhere
\end{equation*}
\end{proof}

\begin{cor}
Let $\hat{M}\in\mathcal{M}_{n,k}$ have isoclinic spherical symmetry. Then $(\hat{W},\hat{\lambda})$ must have trivial summands and $L$ must live in the direct sum of these summands.
\end{cor}

\begin{proof}
As $L\neq 0$, we know $\mathrm{span}(L)\subseteq \hat{W}$ is an invariant 1-dimensional subspace, so $(\mathrm{span}(L),0)$ is a trivial summand of the representation $(\hat{W},\hat{\lambda})$.
\end{proof}

Therefore, we know that $(\hat{W},\hat{\lambda})$ must have trivial summands, which narrows the possibilities for the structure group. Just like in the hyperbolic monopole case, we have straightforward restrictions for high dimensional representations in addition to more subtle restrictions for low dimensional representations. For instance, in Proposition~\ref{prop:IsoEx}, we examine isoclinic symmetric instantons generated by $(V,\rho)\simeq (\mathbb{R}^4,\varrho_{4})\oplus (\mathbb{R}^3,\varrho_3)$. This representation generated instantons with isoclinic spherical symmetry and structure group $\mathrm{Sp}(3)$. However, this representation cannot generate instantons with isoclinic spherical symmetry and lower rank structure groups.
\begin{prop}
Consider Proposition~\ref{prop:IsoEx}, with $(V,\rho)\simeq (\mathbb{R}^4,\varrho_{4})\oplus (\mathbb{R}^3,\varrho_3)$. This representation does not generate an instanton with isoclinic spherical symmetry and structure group $\mathrm{Sp}(1)$ or $\mathrm{Sp}(2)$. 
\end{prop}

\begin{proof}
Suppose $(V,\rho)\simeq (\mathbb{R}^4,\varrho_{4})\oplus (\mathbb{R}^3,\varrho_3)$ generates an instanton with isoclinic spherical symmetry. We know that $(\hat{W},\hat{\lambda})$ has a trivial summand. If we have a $\mathrm{Sp}(1)$ instanton, then we have that the representation obtained by restricting the scalars of $(W,\lambda)$ is isomorphic, as a complex representation, to $(V_1,\rho_1)^{\oplus 2}$ or $(V_2,\rho_2)$. The former does not give $(\hat{W},\hat{\lambda})$ any trivial summands. Thus we look to the latter, which gives two trivial complex summands. 

If the representation obtained by restricting the scalars of $(W,\lambda)$ is given by $(V_2,\rho_2)$, then there is some $q\in\mathrm{Sp}(1)$ such that $y_i=q\upsilon_i q^\dagger$. Taking $\tilde{L}:=q^\dagger L$ and dropping the tilde, we have 
\begin{equation*}
\upsilon_i L=L\rho(\upsilon_i).
\end{equation*}
Solving these equations, there is some $a\in\mathbb{H}$ such that $L=\begin{bmatrix}
1 & -i & -j & k & 0 & 0 & 0
\end{bmatrix}a$. But we know that as, up to a $\rho$-invariant gauge,
\begin{equation*}
M=\lambda\begin{bmatrix}
0 & 0 & 0 & 0 & 1 & -i & j \\
0 & 0 & 0 & 0 & i & 1 & -k \\
0 & 0 & 0 & 0 & j & -k & -1 \\
0 & 0 & 0 & 0 & k & j & i \\
1 & i & j & k & 0 & 0 & 0 \\
-i & 1 & -k & j & 0 & 0 & 0 \\
j & -k & -1 & i & 0 & 0 & 0 
\end{bmatrix},
\end{equation*}
we have $R:=L^\dagger L+M^\dagger M$ is given by
\begin{equation*}
R=a^\dagger\begin{bmatrix}
1 & -i & -j & k & 0 & 0 & 0 \\
i & 1 & -k & -j & 0 & 0 & 0 \\
j & k & 1 & i & 0 & 0 & 0 \\
-k & j & -i & 1 & 0 & 0 & 0 \\
0 & 0 & 0 & 0 & 0 & 0 & 0 \\
0 & 0 & 0 & 0 & 0 & 0 & 0
\end{bmatrix}a+\lambda^2\begin{bmatrix}
3 & 3i & 3j & 3k & 0 & 0 & 0 \\
-3i & 3 & -3k & 3j & 0 & 0 & 0 \\
-3j & 3k & 3 & -3i & 0 & 0 & 0 \\
-3k & -3j & 3i & 3 & 0 & 0 & 0 \\
0 & 0 & 0 & 0 & 4 & 0 & 0 \\
0 & 0 & 0 & 0 & 0 & 4 & 0 \\
0 & 0 & 0 & 0 & 0 & 0 & 4
\end{bmatrix}
\end{equation*}
As we must have $R$ real and commuting with $\rho$, we must have $a^\dagger ia=3\lambda^2i$, $a^\dagger ja=3\lambda^2j$, and $a^\dagger ka=-3\lambda^2k$. Solving these, we have $a=0$. But then $L=0$, contradiction! Therefore, we cannot have a $\mathrm{Sp}(1)$ instanton from this representation $(V,\rho)$.

Instead, if we have a $\mathrm{Sp}(2)$ instanton, then we have that the representation obtained by restricting the scalars of $(W,\lambda)$ is some combination of the previously discussed $2$-representations or $(V_4,\rho_4)$. However, the only choices for the representation obtained by restricting the scalars of $(W,\lambda)$ that create a trivial summand in $(\hat{W},\hat{\lambda})$ include $(V_2,\rho_2)$. But we already showed that these representations do not generate an instanton with isoclinic spherical symmetry. Thus, this representation does not generate a $\mathrm{Sp}(2)$ instanton with isoclinic spherical symmetry either.
\end{proof}

\subsubsection{Non-symmetric singular monopoles}

In this section, we discuss the connection between instantons with isoclinic spherical symmetry and singular monopoles with Dirac type singularities and no continuous symmetries. 

Note that the isoclinic spherical subgroup of $\mathrm{Sp}(2)$ contains $R_0$ as a subgroup, whereas the simple and conformal spherical subgroups do not even contain a subgroup conjugate to $R_0$. Hence, we can associate to an instanton with isoclinic spherical symmetry a singular monopole with Dirac type singularities. Unlike toral symmetry, isoclinic spherical symmetry does not impart extra symmetry to singular monopoles, as the rest of the isoclinic spherical symmetry does not commute with the $R_0$ subgroup. The lack of commuting actions follows from Lemma~\ref{lemma:symsingularmonopoles} and the work in Section~\ref{subsubsec:AxialSymSingularMono}, which tell us that any symmetry on the singular monopole corresponds to isoclinic symmetry of the form $\mathrm{diag}(1,p)$, with $p\in\mathrm{Sp}(1)$.

Note that an instanton with isoclinic spherical symmetry can possess more symmetry, in which case the singular monopole may have a continuous symmetry. In Section~\ref{subsubsec:SpherSymSingularMono}, we discuss the connection between instantons with isoclinic superspherical symmetry and singular monopoles with different symmetries. If the instanton has isoclinic spherical symmetry and does not have isoclinic superspherical symmetry, then the corresponding singular monopole has no continuous symmetries. 

\subsubsection{Hyperbolic analogue to Nahm data}\label{subsubsec:HyperbolicNahm}

In this section, we discuss the connection between instantons with isoclinic spherical symmetry and a hyperbolic analogue to Nahm data. 

Just as Euclidean monopoles and Higgs bundles, Nahm data are a dimensional reduction of instantons. However, while Euclidean monopoles are invariant under translation along one axis and Higgs bundles are invariant under translations along two axes, Nahm data are invariant under translations along three axes. Recall from Section~\ref{subsubsec:HyperbolicHiggsBundles}, for $d\in\{0,1,2\}$ we have the conformal equivalence $S^4\setminus S^d\equiv H^{d+1}\times S^{4-d-1}$. Note that we can replace $S^4\setminus S^d$ with $\mathbb{R}^4\setminus \mathbb{R}^d$.

The fixed point set of the isoclinic spherical action on $\mathbb{R}^4$ is $\{0\}$. Moreover, for $x\neq 0$, we have that the orbit of $x$ under the isoclinic spherical action is $S^3$. Indeed, $\mathrm{Sp}(1)\simeq S^3$ acts freely on $\mathbb{R}^4\setminus\{0\}$. 

On $\mathbb{R}^4\setminus\{0\}$, we can use four dimensional spherical coordinates $(r,\theta,\phi,\psi)$. Doing so, the metric on $\mathbb{R}^4\setminus\{0\}$ is given by
\begin{equation*}
ds^2=dr^2+r^2(d\theta^2+\sin^2\theta d\phi^2+\sin^2\theta\sin^2\phi d\psi^2).
\end{equation*}
As $r\neq 0$ on this set, we can divide by $r^2$, obtaining the metric of $H^1\times S^3$. Specifically, we note that we obtain the upper-half space model of $H^1$.

Therefore, just as in the case of hyperbolic monopoles and hyperbolic analogue to Higgs bundles, an instanton with isoclinic spherical symmetry corresponds to a hyperbolic analogue to Nahm data and vice versa. 

\subsection{Conformal spherical symmetry}

In this section, we find an equation describing all instantons with conformal spherical symmetry, as given in Table~\ref{table:conformalsubgroups} as well as below. We then prove that we already know all examples of such instantons. We also discuss instantons with full symmetry.

Recall the Lie algebra $\mathfrak{sp}(1)\simeq\mathfrak{h}_5\subseteq\mathfrak{sp}(2)$, defined in \eqref{eq:sp1subalgebras}. Let $\tau\colon\mathfrak{sp}(1)\rightarrow\mathfrak{h}_5$ be a Lie algebra isomorphism. An instanton has conformal spherical symmetry if and only if it is equivariant under the unique connected Lie subgroup of $\mathrm{Sp}(2)$ with Lie algebra $\mathfrak{h}_5=\tau(\mathfrak{sp}(1))\subseteq\mathfrak{sp}(2)$. Recall the notation $\hat{M}_\Upsilon$ and $U_\Upsilon$, for $\Upsilon\in\mathfrak{sp}(2)$, introduced in Lemma~\ref{lemma:conformalaction}.
\begin{theorem}
Let $\hat{M}\in\mathcal{M}_{n,k}$. Then $\hat{M}$ has conformal spherical symmetry if and only if there exists a representation $\rho\colon\mathfrak{sp}(1)\rightarrow\mathfrak{sp}(n+k)$ such that for all $\upsilon\in\mathfrak{sp}(1)$,
\label{thm:conformalsphericalsym}
\begin{align}
\rho(\upsilon)\hat{M}-\hat{M}_{\tau(\upsilon)} -\hat{M}U^T\rho(\upsilon) U+\hat{M}U^T U_{\tau(\upsilon)}&=0;\label{eq:conformalsphericalsym1}\\
\rho(\upsilon)U-U_{\tau(\upsilon)} -UU^T\rho(\upsilon)U+UU^T U_{\tau(\upsilon)}&=0.\label{eq:conformalsphericalsym2}
\end{align}
Additionally, we must have that $\rho$ induces a real representation $\lambda\colon\mathfrak{sp}(1)\rightarrow \mathfrak{so}(k)$ given by
\begin{equation}
\lambda(\upsilon):=U^T\rho(\upsilon)U-U^TU_{\tau(\upsilon)}.\label{eq:conformalsphericalsymrep2}
\end{equation}
\end{theorem}

\begin{proof}
We use the same notation as introduced in the beginning of the proof of Theorem~\ref{thm:circularsym}. We follow the proof of Theorem~\ref{thm:mainthm} after setting the scene.

Suppose that $\hat{M}$ has conformal spherical symmetry. Let $S\subseteq \mathrm{Sp}(1)\times\mathrm{Sp}(n+k)$ be the stabilizer group of $(\hat{M},U)$ restricted to conformal transformations in $\mathrm{Sp}(1)\simeq e^{\mathfrak{h}_5}\subseteq\mathrm{Sp}(2)$. Recalling that for $A\in\mathrm{Sp}(2)$, $(\hat{M}_A,U_A):=A.(\hat{M},U)$, we have that
\begin{equation}
S:=\{(A,Q)\mid K(A,Q):=U^TQU_{A^\dagger}\in\mathrm{GL}(k,\mathbb{R})\textrm{, }(Q,K(A,Q)).A^\dagger.(\hat{M},U)
=(\hat{M},U)\}.\label{eq:conformalsphericalsymmetryS}
\end{equation}
Indeed, if $(Q,K)\in\mathrm{Sp}(n+k)\times\mathrm{GL}(k,\mathbb{R})$ such that $A^\dagger.(Q,K).(\hat{M},U)=(\hat{M},U)$, then Lemma~\ref{lemma:KfromQ} tells us that $K=U^TQU_{A^\dagger}$. Note that as $(Q,K).(\hat{M},U)=A.(\hat{M},U)$, we know that $K^{-1}=U^TQ^\dagger U_A$. Therefore, the pairs in $S$ encapsulate all the conformal spherical symmetry of the instanton. The condition $K(A,Q)\in\mathrm{GL}(k,\mathbb{R})$ makes proving $S$ is a closed subgroup more challenging. However, we now prove that $S$ is, in fact, a closed subgroup, hence a compact Lie subgroup. In doing so, we prove that $K\colon S\rightarrow\mathrm{GL}(k,\mathbb{R})$ is a Lie group homomorphism.

Let $f_1\colon\mathrm{Sp}(1)\times\mathrm{Sp}(n+k)\rightarrow\mathcal{X}$ be the smooth map 
\begin{equation*}
f_1(A,Q):=(Q^\dagger\hat{M} K(A,Q),Q^\dagger UK(A,Q))-A^\dagger.(\hat{M},U).
\end{equation*}
Additionally, let $f_2\colon\mathrm{Sp}(1)\times\mathrm{Sp}(n+k)\rightarrow\mathrm{Mat}(k,k,\mathbb{H})$ be the smooth map
\begin{equation*}
f_2(A,Q):=K(A,Q)-\overline{K(A,Q)}.
\end{equation*}
It turns out that $S=f_1^{-1}(0)\cap f_2^{-1}(0)$. Indeed, the inclusion $S\subseteq f_1^{-1}(0)\cap f_2^{-1}(0)$ follows from the definition of $S$. Suppose that $(A,Q)\in f_1^{-1}(0)\cap f_2^{-1}(0)$. Then, in particular, $Q^\dagger UK(A,Q)=U_{A^\dagger}$, so $UK(A,Q)=QU_{A^\dagger}$. As $U_{A^\dagger}$ has rank $k$ and is a $(n+k)\times k$ quaternionic matrix, we know there exists a $k\times (n+k)$ quaternionic matrix $W$ such that $WU_{A^\dagger}=I_k$. Thus, $WQ^\dagger UK(A,Q)=I_k$, so $K$ is invertible. As $f_2(A,Q)=0$, $K(A,Q)$ is real, so belongs in $\mathrm{GL}(k,\mathbb{R})$. Finally, as $f_1(A,Q)=0$, $(A,Q)\in S$. Hence, $S$ is closed. We now show that $S$ is a group.

We see that $(I,I)\in S$. Additionally, suppose that $(A,Q)\in S$. Then 
\begin{equation*}
(Q,K(A,Q)).A^\dagger.(\hat{M},U)=(\hat{M},U).
\end{equation*} 
Hence, $(Q^\dagger,K(A,Q)^{-1}).A.(\hat{M},U)=(\hat{M},U)$. As $QUK(A,Q)^{-1}=U_A$, 
\begin{equation*}
K(A,Q)^{-1}=U^T Q^\dagger U_A=K(A^\dagger,Q^\dagger).
\end{equation*} 
Hence, $K(A^\dagger,Q^\dagger)\in\mathrm{GL}(k,\mathbb{R})$ and $(Q^\dagger,K(A^\dagger,Q^\dagger)).A.(\hat{M},U)=(\hat{M},U)$, so $(A^\dagger,Q^\dagger)\in S$. Finally, suppose that $(A_1,Q_1),(A_2,Q_2)\in S$. Then as the conformal action is a right Lie group action,
\begin{equation*}
\begin{aligned}
(Q_1Q_2,K(A_1,Q_1)K(A_2,Q_2)).(A_1A_2)^\dagger.(\hat{M},U)&=(Q_1,K(A_1,Q_1)).(Q_2,K(A_2,Q_2))\\
&\phantom{=}.A_1^\dagger.A_2^\dagger.(\hat{M},U)\\
&=(\hat{M},U),
\end{aligned}
\end{equation*}
as the conformal and gauge actions commute. Thus, 
\begin{equation*}
((Q_1Q_2)^{-1},(K(A_1,Q_1)K(A_2,Q_2))^{-1}).(\hat{M},U)=(A_1A_2)^\dagger.(\hat{M},U).
\end{equation*}
Hence,
\begin{equation*}
(Q_1Q_2)^{-1}UK(A_1,Q_1)K(A_2,Q_2)=U_{(A_1A_2)^\dagger},
\end{equation*}
so
\begin{equation*}
K(A_1,Q_1)K(A_2,Q_2)=U^TQ_1Q_2U_{(A_1A_2)^\dagger}=K(A_1A_2,Q_1Q_2).
\end{equation*}
Thus, $K(A_1A_2,Q_1Q_2)\in\mathrm{GL}(k,\mathbb{R})$ and $(Q_1Q_2,K(A_1A_2,Q_1Q_2)).(A_1A_2)^\dagger.(\hat{M},U)=(\hat{M},U)$. Hence, $(A_1A_2,Q_1Q_2)\in S$, so $S$ is a group and $K$ is a Lie group homomorphism.

Just like our previous symmetries, other than circular $t$-symmetry, the stabilizer group is a subgroup of a compact group, so we proceed as in the proof of Theorem~\ref{thm:mainthm}~\cite[Theorem~1.1]{lang_moduli_2024}. Though in this case, the compact subgroup is $\mathrm{Sp}(n+k)$ instead of $\mathrm{O}(k)$. In particular, we find that $\hat{M}$ has conformal spherical symmetry if and only if there is a Lie algebra homomorphism $\rho\colon\mathfrak{sp}(1)\rightarrow\mathfrak{sp}(n+k)$ such that for all $\upsilon\in\mathfrak{sp}(1)$, we have
\begin{align*}
\rho(\upsilon)\hat{M}U^TU-\hat{M}_{\tau(\upsilon)} U^TU-\hat{M}U^T\rho(\upsilon) U+\hat{M}U^T U_{\tau(\upsilon)}&=0;\\
\rho(\upsilon)U U^TU-U_{\tau(\upsilon)} U^T U-UU^T\rho(\upsilon)U+UU^T U_{\tau(\upsilon)}&=0.
\end{align*}
Noting that $U^TU=I_k$, we obtain \eqref{eq:conformalsphericalsym1} and \eqref{eq:conformalsphericalsym2}.

Additionally, if $\hat{M}$ has conformal spherical symmetry, then as $K\circ \Psi\colon \mathrm{Sp}(1)\rightarrow \mathrm{GL}(k,\mathbb{R})$ is a Lie group homomorphism, it gives rise to a Lie algebra homomorphism $\lambda\colon\mathfrak{sp}(1)\rightarrow\mathfrak{gl}(k)$. As $\mathrm{Sp}(1)$ is compact, all representations are unitary. In the Lie algebra realm, this means that all representations are skew-Hermitian, so $\lambda\colon\mathfrak{sp}(1)\rightarrow\mathfrak{so}(k)$. As $\mathfrak{so}(k)$ acts naturally on $\mathbb{R}^k$, $\left(\mathbb{R}^k,\lambda\right)$ is a real $k$-representation of $\mathfrak{sp}(1)$. In particular, as $K(\Psi(e^{\theta\upsilon}))=U^Te^{\theta\rho(\upsilon)}U_{e^{-\theta\tau(\upsilon)}}$, we have $\lambda(\upsilon)=U^T\rho(\upsilon)U-U^TU_{\tau(\upsilon)}$.

Unlike the previous symmetries, we do not need to consider the converse direction, as we are considering the full equations of symmetry, not just the bottom component.
\end{proof}

Using Theorem~\ref{thm:conformalsphericalsym}, we find that we already know all instantons with conformal spherical symmetry.
\begin{theorem}
Let $\hat{M}\in\mathcal{M}_{n,k}$. Then $\hat{M}$ has conformal spherical symmetry if and only if $n=k$ and $(\hat{M},U)$ is gauge equivalent to $\left(\begin{bmatrix}
I_k \\ 0
\end{bmatrix}, U\right)$. 
\label{thm:knowconfsphersym}
\end{theorem}

\begin{proof}
We know that all instantons with ADHM data of this form have full symmetry and therefore have conformal spherical symmetry. It remains to show that all such instantons have ADHM data of this form.

Suppose that $\hat{M}$ has conformal spherical symmetry. Consider the Lie algebra isomorphism $\tau\colon\mathfrak{sp}(1)\rightarrow\mathfrak{h}_5$ given by
\begin{equation*}
\tau(\upsilon_1):=\begin{bmatrix}
i/2 & 0 \\ 0 & 3i/2
\end{bmatrix}, \quad
\tau(\upsilon_1):=\begin{bmatrix}
j & \sqrt{3}/2 \\ -\sqrt{3}/2 & 0 
\end{bmatrix}, \quad
\tau(\upsilon_1):=\begin{bmatrix}
k & -\sqrt{3}i/2 \\ -\sqrt{3}i/2 & 0
\end{bmatrix}.
\end{equation*}

From Lemma~\ref{lemma:conformalaction}, we have that
\begin{align*}
\left(\hat{M}_{\tau(\upsilon_1)},U_{\tau(\upsilon_1)}\right)&=\left(\frac{3}{2}\hat{M}i,\frac{i}{2}U\right),\\
\left(\hat{M}_{\tau(\upsilon_2)},U_{\tau(\upsilon_2)}\right)&=\left(-\frac{\sqrt{3}}{2}U,\frac{\sqrt{3}}{2}\hat{M}+Uj\right),\\
\left(\hat{M}_{\tau(\upsilon_3)},U_{\tau(\upsilon_3)}\right)&=\left(\frac{\sqrt{3}}{2}Ui,\frac{\sqrt{3}}{2}\hat{M}i+Uk\right).
\end{align*}

By Theorem~\ref{thm:conformalsphericalsym}, we know there is some $\rho\colon\mathfrak{sp}(1)\rightarrow\mathfrak{sp}(n+k)$ such that $\lambda\colon\mathfrak{sp}(1)\rightarrow\mathfrak{so}(k)$, given by $\lambda(\upsilon):=U^T\rho(\upsilon)U-U^TU_{\tau(\upsilon)}$, is a real representation. Thus, we see that
\begin{equation*}
-U^T\rho(\upsilon)U-\left(U_{\tau(\upsilon)}\right)^\dagger U=\lambda(\upsilon)^\dagger=-\lambda(\upsilon)=-U^T\rho(\upsilon)U+U^TU_{\tau(\upsilon)}.
\end{equation*}
Evaluating at $\upsilon=\upsilon_2$, we see that
\begin{equation*}
\frac{\sqrt{3}}{2}U^T \hat{M}+U^TUj+\frac{\sqrt{3}}{2}\hat{M}^\dagger U-jU^TU=0.
\end{equation*}
As $M$ is symmetric, we have that $M_0=0$. Evaluating at $\upsilon=\upsilon_3$, we see that 
\begin{equation*}
\frac{\sqrt{3}}{2}U^T \hat{M}i+U^T Uk-\frac{\sqrt{3}}{2}i\hat{M}^\dagger U-kU^TU=0.
\end{equation*}
Simplifying, noting that $M^\dagger=-M$, as $M_0=0$, we have that $Mi+iM=0$. Thus, $M=M_1i$. 

As $\hat{M}$ has conformal spherical symmetry, it is $\mathrm{diag}(e^{i\theta/2},e^{3i\theta/2})$-equivariant for all $\theta\in\mathbb{R}$. Using the methods from Theorem~\ref{thm:circularsym}, we have that there is some $\tilde{\rho}\in\mathfrak{so}(k)$ such that 
\begin{equation*}
3Mi-iM+[\tilde{\rho},M]=0.
\end{equation*}
As $M$ is purely imaginary, we can separate the real and imaginary parts of this equation, finding $M=0$. By Proposition~\ref{prop:M=0}, we know that $n=k$ and $L$ is positive definite and diagonal in some gauge.

Finally, we have that $\rho$ satisfies \eqref{eq:conformalsphericalsym1} and \eqref{eq:conformalsphericalsym2}. There exists $A,C\in\mathfrak{sp}(k)$ and $B\in\mathrm{Mat}(k,k,\mathbb{H})$ such that $\rho(\upsilon_2)=\begin{bmatrix}
A & B \\ -B^\dagger & C
\end{bmatrix}$. Evaluating \eqref{eq:conformalsphericalsym2} at $\upsilon=\upsilon_2$, we see that
\begin{equation*}
\begin{bmatrix}
B \\ C
\end{bmatrix}-\frac{\sqrt{3}}{2}\begin{bmatrix}
L \\ 0
\end{bmatrix}-\begin{bmatrix} 0 \\ I_kj\end{bmatrix}-\begin{bmatrix}
0 \\ \lambda(\upsilon_2)
\end{bmatrix}=0.
\end{equation*}
Thus, $B=\frac{\sqrt{3}}{2}L$. Evaluating \eqref{eq:conformalsphericalsym1} at $\upsilon=\upsilon_2$, we see that
\begin{equation*}
\begin{bmatrix}
AL \\ -B^\dagger L
\end{bmatrix}+\frac{\sqrt{3}}{2}\begin{bmatrix}
0 \\ I_k
\end{bmatrix}-\begin{bmatrix}
L\lambda(\upsilon_2) \\ 0
\end{bmatrix}=0.
\end{equation*}
Thus, $B^\dagger L=\frac{\sqrt{3}}{2}I_k$. As $B=\frac{\sqrt{3}}{2}L$, we have that $L^\dagger L=I_k$. Therefore, $L\in\mathrm{Sp}(k)$, proving the theorem.
\end{proof}

The only connected Lie subgroup of $\mathrm{Sp}(2)$ containing the conformal spherical subgroup, other than itself, is the group $\mathrm{Sp}(2)$ itself. Theorem~\ref{thm:knowconfsphersym} tells us about instantons with full symmetry as well. Note that an instanton has full symmetry if and only if it is equivariant under every element of $\mathrm{Sp}(2)$.
\begin{cor}
Let $\hat{M}\in\mathcal{M}_{n,k}$. Then $\hat{M}$ has full symmetry if and only if $n=k$ and $(\hat{M},U)$ is gauge equivalent to $\left(\begin{bmatrix}
I_k \\ 0
\end{bmatrix}, U\right)$. 
\label{cor:knowfullsym}
\end{cor}

\begin{proof}
We know that such instantons have full symmetry. Conversely, if an instanton has full symmetry, then it has conformal spherical symmetry. Theorem~\ref{thm:knowconfsphersym} tells us the form that the ADHM data must take.
\end{proof}

\section{Superspherical symmetry}\label{subsec:SuperSphericalSymmetry}

In this section, we find equations describing all instantons with two kinds of superspherical symmetry, as given in Table~\ref{table:conformalsubgroups}. We separate the cases of isoclinic and conformal superspherical symmetry. We also discuss the connections between these symmetric instantons and symmetric hyperbolic monopoles, singular monopoles, and hyperbolic analogues to Higgs bundles.

\begin{note}
All connected Lie subgroups of $\mathrm{Sp}(2)$ with Lie algebra $\mathfrak{sp}(1)\oplus\mathbb{R}$ are conjugate to the connected Lie subgroups $G_{3,1,1}$ or $G_{4,1}$ induced by $\mathfrak{p}_{3,1,1}$ and $\mathfrak{p}_{4,1}$, respectively, defined in \eqref{eq:sp1s1subalgebras}. That is, there is some $A\in\mathrm{Sp}(2)$ and $G$ one of the subgroups $G_{3,1,1}$ or $G_{4,1}$ such that the Lie group is of the form $AGA^\dagger$. Instantons equivariant under this group are of the form $A^\dagger.(\hat{M},U)$, where $\hat{M}$ is equivariant under $G$.
\end{note}

\subsection{Isoclinic superspherical symmetry}

In this section, we find an equation describing all instantons with isoclinic superspherical symmetry, as given in Table~\ref{table:conformalsubgroups} as well as below. We also discuss the connection between these instantons and symmetric singular monopoles. Note that a representation $\rho\colon\mathfrak{sp}(1)\oplus\mathbb{R}\rightarrow\mathfrak{so}(k)$ corresponds to a representation of $\mathrm{Sp}(1)\times S^1$ if and only if $\mathrm{exp}(\rho(0,2\pi))=I_k$. 

First, we introduce the notion of isoclinic superspherical symmetry. An instanton is said to have isoclinic superspherical symmetry if it is equivariant under $\mathrm{diag}(p,e^{it})$ for all $p\in\mathrm{Sp}(1)$ and $t\in\mathbb{R}$.
\begin{theorem}
Let $\hat{M}\in\mathcal{M}_{n,k}$. Then $\hat{M}$ has isoclinic superspherical symmetry if and only if there exists a real representation $\rho\colon\mathfrak{sp}(1)\oplus\mathbb{R}\rightarrow\mathfrak{so}(k)$, corresponding to a representation of $\mathrm{Sp}(1)\times S^1$, such that for all $\upsilon\in\mathfrak{sp}(1)$ and $t\in\mathbb{R}$,
\label{thm:isoclinicsupersphericalsym}
\begin{align}
\upsilon M+[\rho(\upsilon,t),M]-2tMi=0,\label{eq:isoclinicsupersphericalsym1}\\
[\rho(\upsilon,t),R]=0.\label{eq:isoclinicsupersphericalsym2}
\end{align}
\end{theorem}

\begin{definition}
We call $\rho$ the \textbf{generating representation} of the isoclinic superspherical symmetry of $\hat{M}$.
\end{definition}

\begin{proof}
We use the same notation as introduced in the beginning of the proof of Theorem~\ref{thm:circularsym}. We follow the proof of Theorem~\ref{thm:mainthm} after setting the scene.

Suppose that $\hat{M}$ has isoclinic superspherical symmetry. Let $S\subseteq \mathrm{Sp}(1)\times S^1\times\mathrm{O}(k)$ be the stabilizer group of $(\hat{M},U)$ restricted to rotations in $\mathrm{Sp}(1)\times S^1\subseteq\mathrm{Sp}(2)$. That is
\begin{multline}
S:=\{(\mathrm{diag}(p,e^{i\theta}),K)\mid (\mathrm{diag}(Le^{2i\theta}KL^\dagger (LL^\dagger)^{-1},p K),K).\mathrm{diag}(p^\dagger,e^{-2i\theta}).(\hat{M},U)\\
=(\hat{M},U)\}.\label{eq:isoclinicsupersphericalsymS}
\end{multline}
Indeed, if $(Q,K)\in\mathrm{Sp}(n+k)\times\mathrm{GL}(k,\mathbb{R})$ such that $\mathrm{diag}(p^\dagger,e^{-2i\theta}).(Q,K).(\hat{M},U)=(\hat{M},U)$, then Corollary~\ref{cor:IsomS} tells us that $K\in\mathrm{O}(k)$, $[R,K]=0$, and $Q$ is given by the diagonal matrix in \eqref{eq:isoclinicsupersphericalsymS}. Therefore, the pairs in $S$ encapsulate all the isoclinic superspherical symmetry of the instanton. That $S$ is a group follows from this fact as well as the fact that the conformal action is a right Lie group action, just as in the proof of Theorem~\ref{thm:simplesphericalsym}.

Just like our previous symmetries, other than circular $t$-symmetry, the stabilizer group is a subgroup of a compact group, so we proceed as in the proof of Theorem~\ref{thm:mainthm}~\cite[Theorem~1.1]{lang_moduli_2024}. In particular, we find that $\hat{M}$ has isoclinic superspherical symmetry if and only if there is a Lie algebra homomorphism $\rho\colon\mathfrak{sp}(1)\oplus\mathbb{R}\rightarrow\mathfrak{so}(k)$ such that for all $\upsilon\in\mathfrak{sp}(1)$ and $t\in\mathbb{R}$, we have
\begin{equation}
\left(\begin{bmatrix}
2tLiL^\dagger (LL^\dagger)^{-1}L+L\rho(\upsilon,t)L^\dagger (LL^\dagger)^{-1}L-L\rho(\upsilon,t)-2tLi \\
\upsilon M+[\rho(\upsilon,t),M]-2tMi
\end{bmatrix},0\right)=(0,0).\label{eq:isoclinicsupersphericalsymdiff}
\end{equation}

We can simplify these constraints and be more specific with our homomorphism. In particular, suppose $\hat{M}$ has isoclinic superspherical symmetry. Focusing on the bottom row, we see that $\upsilon M+[\rho(\upsilon,t),M]-2tMi=0$. Furthermore, as mentioned above, Corollary~\ref{cor:IsomS} tells us that $[R,e^{\theta\rho(\upsilon,t)}]=0$, for all $\theta\in\mathbb{R}$. Differentiating and evaluating at $\theta=0$, we have $[R,\rho(\upsilon,t)]=0$. 

Moreover, looking at the map $t\mapsto \rho(0,t)$, we can proceed as in Proposition~\ref{prop:circulartsymrep} and Proposition~\ref{prop:toralsymrep} to find that without loss of generality, we may assume that $\rho$ corresponds to a representation of $\mathrm{Sp}(1)\times S^1$.

We can use the same method as in the proof of Theorem~\ref{thm:circularsym} to prove the converse.
\end{proof}

Above, we see that isoclinic superspherical symmetry is generated by a representation $(\mathbb{R}^k,\rho)$, corresponding to a representation of $\mathrm{Sp}(1)\times S^1$. This representation was obtained by examining the bottom of \eqref{eq:isoclinicsupersphericalsymdiff}. By focusing instead on the top of this equation, we obtain an induced representation.
\begin{lemma}
Suppose that $\hat{M}\in\mathcal{M}_{n,k}$ has isoclinic superspherical symmetry generated by a representation $(\mathbb{R}^k,\rho)$, which corresponds to a representation of $\mathrm{Sp}(1)\times S^1$. Let $\lambda\colon\mathfrak{sp}(1)\oplus\mathbb{R}\rightarrow\mathfrak{sp}(n)$ be defined by $\lambda(\upsilon,t):=L\left(\rho(\upsilon,t)+2ti\right)L^\dagger (LL^\dagger)^{-1}$. We have $(\mathbb{H}^n,\lambda)$ is a quaternionic representation of $\mathfrak{sp}(1)\oplus\mathbb{R}$, corresponding to a representation of $\mathrm{Sp}(1)\times S^1$.
\end{lemma}

\begin{proof}
By Theorem~\ref{thm:isoclinicsupersphericalsym}, we know that for all $t\in\mathbb{R}$ and $\upsilon\in\mathfrak{sp}(1)$, \eqref{eq:isoclinicsupersphericalsym1} and \eqref{eq:isoclinicsupersphericalsym2} hold. Using the former,
\begin{equation*}
\begin{aligned}
M^\dagger M\left(\rho(\upsilon,t)+2ti\right)&=M^\dagger[M,\rho(\upsilon,t)]+M^\dagger \rho(\upsilon,t) M+2tM^\dagger Mi\\
&=M^\dagger \rho(\upsilon,t) M+ M^\dagger \upsilon M\\
&=[M^\dagger,\rho(\upsilon,t)]M+\rho(\upsilon,t) M^\dagger M+M^\dagger \upsilon M\\
&=\left(\rho(\upsilon,t)+2ti\right)M^\dagger M.
\end{aligned}
\end{equation*}
Hence, $\left[M^\dagger M,\rho(\upsilon,t)+2ti\right]=0$. Then, as $[\rho(\upsilon,t),R]=0$, by \eqref{eq:isoclinicsupersphericalsym2}, and $R$ is real, we have $\left[\rho(\upsilon,t)+2ti,R\right]=0$, so $\left[L^\dagger L,\rho(\upsilon,t)+2ti\right]=0$. Thus, we see that 
\begin{equation*}
\left[LL^\dagger ,L\left(\rho(\upsilon,t)+2ti\right)L^\dagger\right]=0.
\end{equation*} 
Therefore, $\lambda(\upsilon,t)\in\mathfrak{sp}(n)$. 

Just as in the toral symmetry case, we have that 
\begin{equation*}
e^{\lambda(\upsilon,t)}=Le^{\rho(\upsilon,t)+2ti}L^\dagger (LL^\dagger)^{-1}.
\end{equation*}
As $\rho$ and $i$ commute, we have that as $(\mathbb{R}^k,\rho)$ corresponds to a representation of $\mathrm{Sp}(1)\times S^1$, $e^{\lambda(0,2\pi)}=I_n$. Therefore, $(\mathbb{H}^n,\lambda)$ corresponds to a representation of $\mathrm{Sp}(1)\times S^1$.
\end{proof}

We know that this induced representation satisfies the following equation.
\begin{cor}
If $\hat{M}\in\mathcal{M}_{n,k}$ has isoclinic superspherical symmetry generated by $(\mathbb{R}^k,\rho)$, then $\lambda(\upsilon,t)L-L\rho(\upsilon,t)-2tLi=0$, for all $t\in\mathbb{R}$ and $\upsilon\in\mathfrak{sp}(1)$.
\end{cor}

\begin{proof}
This result follows from the top component of \eqref{eq:isoclinicsupersphericalsymdiff}.
\end{proof}

\begin{note}
Just as with our previous symmetries, we need not check the final condition of $\mathcal{M}_{n,k}$ in Definition~\ref{def:Mstd} everywhere. Indeed, if $\hat{M}$ satisfies \eqref{eq:isoclinicsupersphericalsym1} and \eqref{eq:isoclinicsupersphericalsym2} for some real representation $(\mathbb{R}^k,\rho)$, then it satisfies \eqref{eq:isoclinicsphericalsym1} and \eqref{eq:isoclinicsphericalsym2} for some real representation of $\mathfrak{sp}(1)$. Thus, just as in the isoclinic spherical symmetry case, we need only check that the final condition of $\mathcal{M}_{n,k}$ in Definition~\ref{def:Mstd} is satisfied for all $x=x_0\in\mathbb{H}$ with $x_0\geq 0$. 
\end{note}

\begin{note}
The Lie algebra of isoclinic superspherical symmetry is comprised of the Lie algebras for isoclinic spherical symmetry and the Lie algebra of a commuting isoclinic circular $0$-symmetry. Therefore, an instanton has isoclinic superspherical symmetry if and only if it has isoclinic spherical symmetry and it is symmetric under this additional commuting circle action. This viewpoint is beneficial, as we can use the work done in Section~\ref{subsec:IsoclinicSpherical} to narrow down our search for isoclinic superspherical symmetry.
\end{note}

\begin{cor}
Let $\hat{M}\in\mathcal{M}_{n,k}$. Then $\hat{M}$ has isoclinic superspherical symmetry if and only if there is some $\rho\colon\mathfrak{sp}(1)\rightarrow \mathfrak{so}(k)$ and $\tau\in\mathfrak{so}(k)$ such that for all $\upsilon\in\mathfrak{sp}(1)$,
\begin{align}
\upsilon M+[\rho(\upsilon),M]&=0,\\
-Mi+[\tau,M]&=0, \\
[\rho(\upsilon),\tau]=[\rho(\upsilon),R]=[\tau,R]&=0.
\end{align}
\end{cor}

\begin{proof}
This result follows from Theorem~\ref{thm:isoclinicsupersphericalsym}, as we have merely decomposed the representation of $\mathfrak{sp}(1)\oplus\mathbb{R}$.
\end{proof}

Proposition~\ref{prop:RotInstEx} provides us with an example of an instanton with rotational symmetry. Such an instanton also possesses isoclinic superspherical symmetry.

\subsubsection{Spherically symmetric singular monopoles}\label{subsubsec:SpherSymSingularMono}

In this section, we discuss the connections between instantons with isoclinic superspherical symmetry and spherically symmetric singular monopoles with Dirac type singularities. Note that the conformal superspherical subgroup contains no subgroup conjugate to $R_0$, so it does not lead to any singular monopoles.

\begin{theorem}
A $\mathrm{Sp}(n)$ singular monopole with Dirac type singularities is spherically symmetric if and only if its ADHM data has isoclinic superspherical symmetry.
\end{theorem}

\begin{proof}
The proof follows from the proof of Proposition~\ref{prop:axialsymsingularmono}, changing from $e^{\upsilon\phi}$ to $p$.
\end{proof}

\begin{note}
We might expect a similar result regarding hyperbolic analogues to Nahm data, where the extra symmetry of isoclinic superspherical symmetry descends to a symmetry on these analogues. However, no such result exists. Indeed, recall that we use the conformal equivalence $\mathbb{R}^4\setminus\{0\}\equiv H^1\times S^3$ to relate instantons with isoclinic spherical symmetry to hyperbolic analogues to Nahm data. In particular, using four-dimensional spherical coordinates $(r,\theta,\phi,\psi)$ on $\mathbb{R}^4\setminus\{0\}$, the hyperbolic coordinate is $r$. However, $r$ is invariant under the entire isoclinic superspherical action. Thus, the extra symmetry of the instanton does not impart any extra symmetry to the hyperbolic analogues to Nahm data.
\end{note}

\subsection{Conformal superspherical symmetry}\label{subsubsec:conformalsupersphericalsym}

In this section, we find an equation describing all instantons with conformal superspherical symmetry, as given in Table~\ref{table:conformalsubgroups} as well as below. We also discuss the connection between these instantons and symmetric hyperbolic monopoles and hyperbolic analogues to Higgs bundles.

First, we introduce the notion of conformal superspherical symmetry. An instanton is said to have conformal superspherical symmetry if it is equivariant under the connected Lie group corresponding to $\mathfrak{p}_{3,1,1}\subseteq\mathfrak{sp}(2)$, defined in \eqref{eq:sp1s1subalgebras}. Recall the notation $\hat{M}_\Upsilon$ and $U_\Upsilon$, for $\Upsilon\in\mathfrak{sp}(2)$, introduced in Lemma~\ref{lemma:conformalaction}.
\begin{theorem}
Let $X:=\begin{bmatrix}
0 & 1 \\ -1 & 0
\end{bmatrix}\in\mathfrak{sp}(2)$. Let $\hat{M}\in\mathcal{M}_{n,k}$. Then $\hat{M}$ has conformal superspherical symmetry if and only if there exists a representation $\rho\colon\mathfrak{sp}(1)\oplus\mathbb{R}\rightarrow\mathfrak{sp}(n+k)$ such that for all $\upsilon\in\mathfrak{sp}(1)$ and $t\in\mathbb{R}$,
\begin{align}
\rho(\upsilon,t)\hat{M}-\hat{M}_{\upsilon I_2+tX} -\hat{M}U^T\rho(\upsilon,t) U+\hat{M}U^T U_{\upsilon I_2+tX}&=0;\label{eq:conformalsupersphericalsym1}\\
\rho(\upsilon,t)U-U_{\upsilon I_2+tX} -UU^T\rho(\upsilon,t)U+UU^T U_{\upsilon I_2+tX}&=0.\label{eq:conformalsupersphericalsym2}
\end{align}
Additionally, we must have that $\rho$ induces a real representation $\lambda\colon\mathfrak{sp}(1)\oplus\mathbb{R}\rightarrow \mathfrak{gl}(k)$ given by
\label{thm:conformalsupersphericalsym}
\begin{equation}
\lambda(\upsilon,t):=U^T\rho(\upsilon,t)U-U^TU_{\upsilon I_2+tX}.\label{eq:conformalsupersphericalrep2}
\end{equation}
\end{theorem}

\begin{definition}
We call $\rho$ the \textbf{generating representation} of the conformal superspherical symmetry of $\hat{M}$.
\end{definition}

\begin{note}
In contrast to the cases of conformal spherical and full symmetry, $\lambda$ is allowed to take values in $\mathfrak{gl}(k)$ outside of $\mathfrak{so}(k)$. This phenomenon is possible because the universal cover of this symmetry group is $\mathrm{Sp}(1)\times \mathbb{R}$, which is not compact.
\end{note}

\begin{proof}
We use the same notation as introduced in the beginning of the proof of Theorem~\ref{thm:circularsym}. We follow the proof of Theorem~\ref{thm:mainthm} after setting the scene.

Suppose that $\hat{M}$ has conformal superspherical symmetry. Let 
\begin{equation*}
S\subseteq ((\mathrm{Sp}(1)\times S^1)/\{\pm(1,1)\})\times\mathrm{Sp}(n+k)
\end{equation*} 
be the stabilizer group of $(\hat{M},U)$ restricted to conformal superspherical transformations. Recall that for $A\in\mathrm{Sp}(2)$, $(\hat{M}_A,U_A):=A.(\hat{M},U)$. Moreover, note that $e^{\theta X}=\begin{bmatrix}
\cos\theta & \sin\theta \\ -\sin\theta & \cos\theta
\end{bmatrix}$. We have
\begin{multline}
S:=\{([(p,e^{i\theta})],Q)\mid K([(p,e^{i\theta})],Q):=U^TQU_{p^\dagger e^{-\theta X}}\in\mathrm{GL}(k,\mathbb{R})\textrm{, }\\
(Q,K([(p,e^{i\theta})],Q)).p^\dagger e^{-\theta X}.(\hat{M},U)=(\hat{M},U)\}.\label{eq:conformalsupersphericalsymS}
\end{multline}
Just as in Theorem~\ref{thm:conformalsphericalsym}, noting that $(\mathrm{Sp}(1)\times S^1)/\{\pm (1,1)\}$ is compact, we have that $S$ is a compact Lie subgroup, encapsulating all the conformal superspherical symmetry of the instanton.

Just like our previous symmetries, other than circular $t$-symmetry, the stabilizer group is a subgroup of a compact group, so we proceed as in the proof of Theorem~\ref{thm:mainthm}~\cite[Theorem~1.1]{lang_moduli_2024}. Though in this case, the compact subgroup is $\mathrm{Sp}(n+k)$ instead of $\mathrm{O}(k)$. In particular, we find that $\hat{M}$ has conformal superspherical symmetry if and only if there is a Lie algebra homomorphism $\rho\colon\mathfrak{sp}(1)\oplus\mathbb{R}\rightarrow\mathfrak{sp}(n+k)$ such that for all $\upsilon\in\mathfrak{sp}(1)$ and $t\in\mathbb{R}$, we have
\begin{align*}
\rho(\upsilon,t)\hat{M}U^TU-\hat{M}_{\upsilon I_2+tX} U^TU-\hat{M}U^T\rho(\upsilon,t) U+\hat{M}U^T U_{\upsilon I_2+tX}&=0;\\
\rho(\upsilon,t)U U^TU-U_{\upsilon I_2+tX} U^T U-UU^T\rho(\upsilon,t)U+UU^T U_{\upsilon I_2+tX}&=0.
\end{align*}
Noting that $U^TU=I_k$, we obtain \eqref{eq:conformalsupersphericalsym1} and \eqref{eq:conformalsupersphericalsym2}.

As $K\circ \Psi\circ G\colon \mathrm{Sp}(1)\times\mathbb{R}\rightarrow \mathrm{GL}(k,\mathbb{R})$ is a Lie group homomorphism, it gives rise to a Lie algebra homomorphism $\lambda\colon\mathfrak{sp}(1)\oplus\mathbb{R}\rightarrow\mathfrak{gl}(k)$. As $\mathfrak{gl}(k)$ acts naturally on $\mathbb{R}^k$, $\left(\mathbb{R}^k,\lambda\right)$ is a real $k$-representation of $\mathfrak{sp}(1)\oplus\mathbb{R}$. In particular, as $K(\Psi(G(e^{\theta\upsilon}e^{\theta tX})))=U^Te^{\theta\rho(\upsilon,t)}U_{e^{-\theta\upsilon}e^{-\theta tX}}$, we have $\lambda(\upsilon,t)=U^T\rho(\upsilon,t)U-U^TU_{\upsilon I_2+tX}$.

Just like in the proof of Theorem~\ref{thm:conformalsphericalsym}, we do not need to consider the converse direction, as we are considering the full equations of symmetry, not just the bottom component.
\end{proof}

\begin{note}
Just as with our previous symmetries, we need not check the final condition of $\mathcal{M}_{n,k}$ in Definition~\ref{def:Mstd} everywhere. Indeed, if $\hat{M}$ has conformal superspherical symmetry, then it has simple spherical symmetry. Thus, just as in the simple spherical symmetry case, we need only check that the final condition of $\mathcal{M}_{n,k}$ in Definition~\ref{def:Mstd} is satisfied for all $x=x_0+x_1i\in\mathbb{H}$ with $x_1\geq 0$. 
\end{note}

\begin{note}
The Lie algebra of conformal superspherical symmetry is comprised of the Lie algebras for simple spherical symmetry and symmetry under Manton--Sutcliffe's conformal circle action. Therefore, an instanton has conformal superspherical symmetry if and only if it has simple spherical symmetry and is symmetric under Manton--Sutcliffe's conformal circle action. This viewpoint is beneficial as we can use the work done in Section~\ref{subsubsec:SimpleSphericalSym} to narrow down our search for conformal superspherical symmetry.\label{note:twoactions}
\end{note}

\begin{cor}
Let $\hat{M}\in\mathcal{M}_{n,k}$. Then $\hat{M}$ has conformal superspherical symmetry if and only if it has simple spherical symmetry and there exists $\rho\in\mathfrak{sp}(n+k)$ such that $U^T\rho U-U^T\hat{M}$ is real and\label{cor:conformalsupersphericalsym}
\begin{align}
\rho\hat{M}+U-\hat{M}U^T\rho U+\hat{M}U^T\hat{M}&=0;\label{eq:MS1}\\
\rho U-\hat{M}-UU^T\rho U+UU^T \hat{M}&=0.\label{eq:MS2}
\end{align}
\end{cor}

\begin{proof}
Suppose $\hat{M}\in\mathcal{M}_{n,k}$. Let $\Upsilon:=\begin{bmatrix}
0 & 1 \\ -1 & 0
\end{bmatrix}\in\mathfrak{sp}(2)$ be the generator of Manton--Sutcliffe's conformal circle action. Using Lemma~\ref{lemma:conformalaction}, we see that $\hat{M}_\Upsilon=-U$ and $U_\Upsilon=\hat{M}$. Using methods from the proof of Theorem~\ref{thm:conformalsphericalsym}, we find that $\hat{M}$ is symmetric under Manton--Sutcliffe's conformal circle action if and only if there exists $\rho\in\mathfrak{sp}(n+k)$ such that \eqref{eq:MS1} and \eqref{eq:MS2} hold and $U^T\rho U-U^T\hat{M}$ is real.

Note~\ref{note:twoactions} tells us that an instanton has conformal superspherical symmetry if and only if it has simple spherical symmetry and is symmetric under Manton--Sutcliffe's conformal circle action, proving the result.
\end{proof}

\begin{note}
Because the simple spherical symmetry and the additional circular $1$-symmetry commute, $\rho$ must commute with the matrix $Q$ used in simple spherical symmetry, given in the proof of Theorem~\ref{thm:simplesphericalsym}.
\end{note}

\subsubsection{Novel example of conformal superspherical symmetry}\label{subsubsec:NotinMSset}

In this section, we identify a novel example of an instanton with conformal superspherical symmetry. Note that several examples were identified in previous work, though there they are viewed through the lens of hyperbolic monopoles~\cite[Propositions~5,~6~\&~7]{lang_hyperbolic_2023}. In fact, the ADHM data for any spherically symmetric hyperbolic monopole in previous work has conformal superspherical symmetry.

More generally, the ADHM data for any hyperbolic monopole in previous work is symmetric under Manton--Sutcliffe's conformal circle action, by design. We identify the condition that forces the ADHM data to take the form in previous work on hyperbolic monopoles~\cite[Definition~1]{lang_hyperbolic_2023}.
\begin{prop}
Consider $\hat{M}\in\mathcal{M}_{n,k}$. Then $\hat{M}$ is symmetric under Manton--Sutcliffe's conformal circle action and $M_0=0$ if and only if $\hat{M}$ is in the generalization of Manton--Sutcliffe's set~\cite[Definition~1]{lang_hyperbolic_2023}. That is, $M$ is pure and $R:=L^\dagger L+M^\dagger M=I_k$. \label{prop:WheninMSset}
\end{prop}

\begin{proof}
Manton--Sutcliffe prove that data in their set is symmetric under their conformal circle action~\cite[\S4]{manton_platonic_2014}. We prove the converse. Suppose that $\hat{M}$ is symmetric under Manton--Sutcliffe's conformal circle action and $M_0=0$. 

From the proof of Corollary~\ref{cor:conformalsupersphericalsym}, we have that there is some $\rho\in\mathfrak{sp}(n+k)$ such that $U^T\rho U-U^T\hat{M}$ is real and \eqref{eq:MS1} and \eqref{eq:MS2} hold.

Write $\rho:=\begin{bmatrix}
X & Y \\ -Y^\dagger & Z
\end{bmatrix}$, so $X\in\mathfrak{sp}(n)$, $Y\in\mathrm{Mat}(n,k,\mathbb{H})$, and $Z\in\mathfrak{sp}(k)$. We then see that $Z-M$ must be real. Thus, there exists $P\in\mathfrak{so}(k)$ such that $Z=M+P$. Substituting and simplifying \eqref{eq:MS2}, we see that $Y=L$. Substituting and simplifying \eqref{eq:MS1}, we see that 
\begin{equation*}
\begin{bmatrix}
XL+LM-LZ+LM\\
-L^\dagger L+ZM+I_k-MZ+M^2
\end{bmatrix}=0.
\end{equation*}
Note that as $M$ is pure, $M^\dagger M=-M^2$, so the bottom row is $(I_k-R)+[Z,M]=0$. As $Z=M+P$ and $M$ is pure, $[Z,M]$ is pure. As $I_k-R$ is real, we have that $I_k-R=0=[Z,M]$. So, in particular, $R=I_k$. Therefore, the ADHM data is in the generalization of Manton--Sutcliffe's set~\cite[Definition~1]{lang_hyperbolic_2023}.
\end{proof}

The preceding proposition tells us that if we are to find ADHM symmetric under Manton--Sutcliffe's conformal circle action but not in the previously considered set, it must have $M_0\neq 0$. We find exactly such an instanton. The following is not the first example of such an instanton, but it is the first with a higher rank structure group~\cite[(7.1)]{manton_platonic_2014}.
\begin{prop}
For $B\in\left(0,\frac{2}{3}\sqrt{6}\right)$, let 
\begin{equation}
\begin{aligned}
A&:=\frac{\sqrt{12 - 15B^2 + 12\sqrt{B^4 - B^2 + 1}}}{3},\quad\textrm{and}\\ 
a&:=\frac{2B^2 - \sqrt{B^4 - B^2 + 1} - 1}{\sqrt{-3B^2 + 6 + 6\sqrt{B^4 - B^2 + 1}}}.
\end{aligned}
\end{equation} 
Note that for all $B$ in our range, $A>0$ and $a\in\mathbb{R}$. Then consider 
\begin{equation}
\hat{M}:=\begin{bmatrix}
A & 0 \\ 
0 & B \\
0 & a \\
a & 0
\end{bmatrix}.
\end{equation}
We have that $\hat{M}\in\mathcal{M}_{2,2}$ and has conformal superspherical symmetry and instanton number $2$. Moreover, when $B\neq 1$, the ADHM data is not in the generalization of Manton--Sutcliffe's set~\cite[Definition~1]{lang_hyperbolic_2023}.\label{prop:notinMSset}
\end{prop}

\begin{proof}
First we show that $\hat{M}$ satisfies the equations for simple spherical symmetry. Indeed, consider $(\mathbb{R}^2,\rho):=(\mathbb{R},0)^{\oplus 2}$. As $a\in\mathbb{R}$ for all $B\in\left(0,\frac{2}{3}\sqrt{6}\right)$, we have that $M$ is a real matrix. As $\rho(\upsilon)=0$ for all $\upsilon\in\mathfrak{sp}(1)$, we have that \eqref{eq:simplesphericalsym1} and \eqref{eq:simplesphericalsym2} are satisfied for all $\upsilon\in\mathfrak{sp}(1)$. Hence, if $\hat{M}\in\mathcal{M}_{2,2}$, then it has simple spherical symmetry.

Next we show that $\hat{M}$ satisfies the equations for symmetry under Manton--Sutcliffe's conformal circle action. Let 
\begin{equation}
\begin{aligned}
b&:=-\frac{\sqrt{12 - 15B^2 + 12\sqrt{B^4 - B^2 + 1}}}{\sqrt{-3B^2 + 6 + 6\sqrt{B^4 - B^2 + 1}}}B,\quad\textrm{and}\\ 
c&:=-\frac{\sqrt{-3B^2 + 6 + 6\sqrt{B^4 - B^2 + 1}}}{3}.
\end{aligned}
\end{equation} 
Then consider $\tilde{\rho}\in\mathfrak{sp}(4)$ given by 
\begin{equation}
\tilde{\rho}:=\begin{bmatrix}
0 & b & A & 0 \\
-b & 0 & 0 & B \\
-A & 0 & 0 & c \\
0 & -B & -c & 0
\end{bmatrix}.
\end{equation}
Then \eqref{eq:MS1} and \eqref{eq:MS2} are satisfied, along with the condition that $U^T \tilde{\rho}U-U^T\hat{M}$ is real. Hence, if $\hat{M}\in\mathcal{M}_{2,2}$, then it has conformal superspherical symmetry. In particular, it is symmetric under Manton--Sutcliffe's conformal circle action.

Note that if $B\neq 1$, then $a\neq 0$. Thus, the ADHM data is not in the generalization of Manton--Sutcliffe's set. It remains to show that $\hat{M}\in\mathcal{M}_{2,2}$. In particular, we note that $M$ is symmetric and $LL^\dagger$ is always positive-definite.  Additionally, $R=L^\dagger L+M^\dagger M$ is real and non-singular. It remains to check the final condition in Definition~\ref{def:Mstd}. Due to the aforementioned symmetry, we need only check that this condition is satisfied for all $x=x_0+x_1i$ with $x_1\geq 0$. The eigenvalues of $\Delta(x)^\dagger \Delta(x)$ for such $x$ are
\begin{multline*}
-\frac{2}{3}B^2+|x|^2+\frac{4}{3}\sqrt{B^4-B^2+1}+\frac{1}{3}\\
\pm\frac{2}{3}\sqrt{(-4B^2 + 6x_0^2 + 2)\sqrt{B^4 - B^2 + 1} - (3B^2 + 3)x_0^2 + 5B^4 - 5B^2 + 2}.
\end{multline*}
Letting $I:=-2B^2+3x_0^2+4\sqrt{B^4-B^2+1}+1>0$, we have that the eigenvalues can be written as 
\begin{equation*}
\frac{I+3x_1^2\pm \sqrt{I^2-9(x_0^2+1)^2}}{3}>0.
\end{equation*}
Therefore, we see that $\Delta(x)^\dagger\Delta(x)$ is non-singular everywhere, so $\hat{M}\in\mathcal{M}_{2,2}$. 
\end{proof}

\begin{note}
In the previous example, we have that $t\mapsto \lambda(0,t)$ does not correspond to a representation of $S^1$ when $B\neq 1$, as $\lambda(0,t)$ is not skew-symmetric. This fact seems to contradict the result in Proposition~\ref{prop:circulartsymrep}. However, this apparent contradiction is only because we are using a non-standard circle action, conjugate to the usual $R_1$ action. In particular, we are using a conformal circle action.

The usual $R_1$ circle action and Manton--Sutcliffe's are related, for all $\theta\in\mathbb{R}$, via
\begin{equation*}
\frac{1}{2}\begin{bmatrix}
-k & -j \\ 
i & 1
\end{bmatrix}\mathrm{exp}\left(\theta\begin{bmatrix}
0 & 1 \\ -1 & 0
\end{bmatrix}\right)\begin{bmatrix}
k & -i \\ j & 1
\end{bmatrix}=\mathrm{exp}\left(\theta \begin{bmatrix}
i & 0 \\ 0 & i
\end{bmatrix}\right).
\end{equation*}

So, taking the ADHM data $(\hat{M},U)\in\mathcal{N}_{2,2}$ equivariant under Manton--Sutcliffe's conformal circle action, consider the ADHM data $(\hat{M}',U')\in\mathcal{N}_{2,2}$ defined by
\begin{equation*}
(\hat{M}',U'):=\frac{1}{\sqrt{2}}\begin{bmatrix}
k & -i \\
j & 1
\end{bmatrix}.(\hat{M},U)=\left(\frac{1}{\sqrt{2}}\begin{bmatrix}
A & 0 \\
0 & B \\
i & a \\
a & i
\end{bmatrix},\begin{bmatrix}
-Aj & 0 \\
0 & -Bj \\
k & -aj \\
-aj & k
\end{bmatrix}\right).
\end{equation*}
Note that the symmetry under Manton--Sutcliffe's conformal circle action in $\hat{M}$ corresponds to symmetry under the usual $R_1$ circle action for $(\hat{M}',U')$. 

This data is not in the standard form. Let
\begin{align*}
Q&:=\begin{bmatrix}
\sqrt{\frac{a^2+1}{A^2+a^2+1}} & 0 & -\frac{Ai}{\sqrt{(A^2+a^2+1)(a^2+1)}} & -\frac{Aa}{\sqrt{(A^2+a^2+1)(a^2+1)}} \\
0 & \sqrt{\frac{a^2+1}{A^2+a^2+1}} & -\frac{Ba}{\sqrt{(B^2+a^2+1)(a^2+1)}} & -\frac{Bi}{\sqrt{(B^2+a^2+1)(a^2+1)}} \\
\frac{Aj}{\sqrt{A^2+a^2+1}} & 0 & -\frac{k}{\sqrt{A^2+a^2+1}} & \frac{aj}{\sqrt{A^2+a^2+1}} \\
0 & \frac{Bj}{\sqrt{B^2+a^2+1}} & \frac{aj}{\sqrt{B^2+a^2+1}} & -\frac{k}{\sqrt{B^2+a^2+1}}
\end{bmatrix}\in\mathrm{Sp}(4) \quad\textrm{and}\\
K&:=\mathrm{diag}\left(\sqrt{\frac{A^2+a^2+1}{2}},\sqrt{\frac{B^2+a^2+1}{2}}\right)\in\mathrm{GL}(2,\mathbb{R}).
\end{align*}
Gauging by $(Q,K)$, we get something in standard form: $(\hat{M}'',U):=(Q,K).(\hat{M}',U')\in\mathcal{N}_{2,2}$. 

Then, given $\rho=\begin{bmatrix}
0 & b \\ -b & 0
\end{bmatrix}\in\mathfrak{so}(2)$, we have that $\rho$ generates the circular $1$-symmetry of $\hat{M}''$ if and only if $b=-1$. Thus, the symmetry is generated by a representation corresponding to one of $S^1$.
\end{note}

\subsubsection{Spherically symmetric hyperbolic monopoles}\label{subsubsec:SpherSymHyperMono}

In this section, we discuss the connection between superspherical symmetric instantons and symmetric hyperbolic monopoles. 

Note that both the conformal and isoclinic superspherical subgroups of $\mathrm{Sp}(2)$ contain subgroups conjugate to $R_1$. Hence, we can associate to an instanton with conformal or isoclinic superspherical symmetry a hyperbolic monopole with integral mass. 

Just like instantons with toral symmetry, the conformal and isoclinic superspherical subgroups impart additional symmetry to the hyperbolic monopole. In the case of the isoclinic superspherical subgroup, there is a $S^1\times S^1$ subgroup of $R_1$ circle actions. Thus, such monopoles are axially symmetric. As the rest of the isoclinic subgroup does not commute with either of the $R_1$ circle actions, it does not impart any additional symmetry to the hyperbolic monopole.

The case of conformal superspherical symmetry imparts additional symmetry, as one of the $R_1$ circle actions commutes with more than just an $S^1$ subgroup. However, note that if we choose a different $R_1$ circle action, then we obtain an hyperbolic monopole with axial symmetry, not spherical symmetry.
\begin{prop}
A hyperbolic $\mathrm{Sp}(n)$-monopole with integral mass is spherically symmetric if and only if its ADHM data has conformal superspherical symmetry. 
\end{prop}

\begin{proof}
Through isometries, we may assume that the hyperbolic model is the ball model. Because of the conformal equivalence of $S^4\setminus S^2\equiv H^3\times S^1$, these isometries on hyperbolic space correspond to conformal maps on $S^4$. Because of the conformal equivalence of $S^4\setminus S^2\equiv H^3\times S^1$, these isometries on hyperbolic space correspond to conformal maps on $S^4$.

The action of $\begin{bmatrix}
\cos\theta & -\sin\theta \\ \sin\theta & \cos\theta
\end{bmatrix}$ induces the ball model of hyperbolic space, with coordinates $(x,y,z)$ satisfying $r:=\sqrt{x^2+y^2+z^2}<1$~\cite[\S4]{manton_platonic_2014}, where $(x,y,z)\in H^3$ corresponds to a point $xi+yj+zk\in\mathfrak{sp}(1)\subseteq\mathbb{H}$. The two-sphere removed is the $S^2\subseteq \mathfrak{sp}(1)$ comprised of all points $xi+yj+zk$ with $r=1$. 

An element of the conformal superspherical subgroup is of the form $p\begin{bmatrix}
\cos\phi & -\sin\phi \\ \sin\phi & \cos\phi
\end{bmatrix}$. Let $R_p\in\mathrm{SO}(3)$ be the rotation generated by $\mathrm{diag}(p,p)$ acting on $\mathfrak{sp}(1)$. The transformation sends $(x,y,z,\theta)$ to $(x',y',z',\theta-2\phi)$, where $(x',y',z'):=R_p(x,y,z)$~\cite[\S4]{manton_platonic_2014}. Thus, we see that $p$ part of the transformation just rotates the $H^3$ part and the $\phi$ part rotates the $S^1$ part. Note that $p\mapsto R_p$ is the double cover $\mathrm{Sp}(1)\rightarrow\mathrm{SO}(3)$. 

Therefore, just as in Proposition~\ref{prop:axialhypermono}, a hyperbolic monopole with integral mass has spherical symmetry if and only if its corresponding ADHM data has conformal superspherical symmetry.
\end{proof}

\begin{note}
A hyperbolic monopole with integral mass corresponds with a circle-invariant instanton. Let $\mathfrak{g}\subseteq\mathfrak{sp}(2)$ be the Lie subalgebra generating the circle action. From the above proposition, we know that a hyperbolic monopole with integral mass is spherically symmetric if and only if $\mathfrak{g}$ commutes with a Lie subalgebra $\mathfrak{h}\subseteq\mathfrak{sp}(2)$ isomorphic to $\mathfrak{sp}(1)$. 
\end{note}

\begin{note}
The conformal superspherical symmetry subgroup that corresponds to spherically symmetric hyperbolic monopoles using the Braam--Austin circle action has Lie algebra
\begin{equation}
\left\langle\begin{bmatrix}
i & 0 \\ 0 & -i
\end{bmatrix},\frac{1}{2}\begin{bmatrix}
i & 0 \\ 0 & i
\end{bmatrix},\frac{1}{2}\begin{bmatrix}
0 & j \\ j & 0
\end{bmatrix},\frac{1}{2}\begin{bmatrix}
0 & k \\ k & 0
\end{bmatrix}\right\rangle\simeq\mathbb{R}\oplus\mathfrak{sp}(1).
\end{equation}
\end{note}

Toral symmetry is generated by two circle actions. We can take both circle actions to be conjugate to $R_1$. Indeed, we have the commuting circle actions of $\mathrm{diag}(e^{i\theta},e^{-i\theta})$ and $\mathrm{diag}(e^{i\theta},e^{i\theta})$. Using either one of these circle actions, we can transform an instanton with toral symmetry into an axially symmetric hyperbolic monopole. However, these monopoles need not be gauge equivalent.
\begin{prop}
Consider the basic instanton, with ADHM data $\hat{M}=\begin{bmatrix}
1 & 0 
\end{bmatrix}^T\in\mathcal{M}_{1,1}$. In Proposition~\ref{prop:basicsymmetries}, we see that the group of symmetries of this instanton is $\mathrm{Sp}(2)$. The hyperbolic monopoles obtained using two commuting circle actions conjugate to $R_1$ are not gauge equivalent, they have different symmetries.\label{prop:notgauge}
\end{prop}

\begin{proof}
Using Manton--Sutcliffe's conformal circle action $\begin{bmatrix}
\cos\theta & -\sin\theta \\ \sin\theta & \cos\theta
\end{bmatrix}$, which is conjugate to $R_1$, we obtain a spherically symmetric hyperbolic monopole~\cite[Proposition~3]{lang_hyperbolic_2023}.

The circle action $\mathrm{diag}(e^{i\theta},e^{i\theta})$ commutes with Manton--Sutcliffe's and is exactly the usual $R_1$ circle action. If the hyperbolic monopole obtained using this action were gauge equivalent to the one from above, there would have to be a Lie subgroup of $\mathrm{Sp}(2)$ (the group of symmetries of the instanton), commuting with $R_1$, with Lie algebra isomorphic to $\mathfrak{sp}(1)$. However, no such subgroup exists. Indeed, the Lie algebra isomorphic to $\mathfrak{sp}(1)$ would have to be a Lie subalgebra of $\mathfrak{sp}(2)$ commuting with $\mathrm{diag}(i,i)$, which generates the $R_1$ circle action, and no such subalgebra exists. Therefore, the two hyperbolic monopoles have different symmetries and are therefore not gauge equivalent.
\end{proof}

\begin{note}
This gauge inequivalence is not a problem, as the choice of circle action makes no difference. Either choice obtains all axially symmetric hyperbolic monopoles with integral mass.
\end{note}

\subsubsection{Hyperbolic analogue to Higgs bundles with axial symmetry}

In this section, we discuss the connection between instantons with conformal superspherical symmetry and axially symmetric hyperbolic analogues to Higgs bundles.
\begin{prop}
A hyperbolic $\mathrm{Sp}(n)$ analogue to a Higgs bundle is axially symmetric if and only if its ADHM data has conformal superspherical symmetry.
\end{prop}

\begin{proof}
Through isometries, we may assume that the hyperbolic model is the half-space model. Because of the conformal equivalence of $S^4\setminus S^1\equiv H^2\times S^2$, these isometries correspond to conformal maps on $S^4$. 

The simple spherical action induces the half-space model of hyperbolic space, with coordinates $(x_0,r)$ with $r>0$. The circle removed is $\mathbb{R}\cup\{\infty\}$. Letting $\theta,\phi$ be the coordinates of $S^2$, we have that a point $(x_0,r,\theta,\phi)\in H^2\times S^2$ corresponds to a point $x_0+r\sin\theta\cos\phi i+r\sin\theta\sin\phi j+r\cos\theta k\in\mathbb{H}$. 

We now look at how the conformal superspherical action acts on these coordinates. Given $x\in\mathbb{H}$, write it in the above coordinates. Let $[(p,e^{i\theta})]\in (\mathrm{Sp}(1)\times S^1)/\{\pm(1,1)\}$. If $x\notin\mathbb{R}$, then this conformal transformation takes $x$ to another element in $\mathbb{H}$. Specifically, $x\mapsto p(x\cos\theta+\sin\theta)(-x\sin\theta+\cos\theta)^{-1}p^\dagger$. Simplifying, we find that this transformation acts on $H^2$ taking
\begin{equation*}
\begin{aligned}
x_0&\mapsto \frac{x_0\cos(2\theta)+\frac{1-x_0^2-r^2}{2}\sin(2\theta)}{(\cos\theta-x_0\sin\theta)^2+r^2\sin^2\theta} \quad\textrm{and} \\
r&\mapsto \frac{r}{(\cos\theta-x_0\sin\theta)^2+r^2\sin^2\theta}.
\end{aligned}
\end{equation*}
This action corresponds to rotation of $(x_0,r)$ by $2\theta$ about the point $(0,1)$. That is, it is an isometry. 

Therefore, just as in Proposition~\ref{prop:axialhypermono}, a hyperbolic analogue to a Higgs bundle is axially symmetric if and only if its ADHM data has conformal superspherical symmetry.
\end{proof}

\section{Rotational symmetry}\label{subsec:RotationalSymmetry}

In this section, we find an equation describing all instantons with rotational symmetry, as given in Table~\ref{table:conformalsubgroups} as well as below.

First, we introduce the notion of rotational symmetry. An instanton is said to have rotational symmetry if it is equivariant under $\mathrm{diag}(p,q)$ for all $p,q\in\mathrm{Sp}(1)$.  
\begin{theorem}
Let $\hat{M}\in\mathcal{M}_{n,k}$. Then $\hat{M}$ has rotational symmetry if and only if there exists a real representation $\rho\colon\mathfrak{sp}(1)\oplus\mathfrak{sp}(1)\rightarrow\mathfrak{so}(k)$ such that for all $\upsilon,\omega\in\mathfrak{sp}(1)$,
\label{thm:rotationalsym}
\begin{align}
\upsilon M+[\rho(\upsilon,\omega),M]-M\omega=0,\label{eq:rotationalsym1}\\
[\rho(\upsilon,\omega),R]=0.\label{eq:rotationalsym2}
\end{align}
\end{theorem}

\begin{definition}
We call $\rho$ the \textbf{generating representation} of the rotational symmetry of $\hat{M}$.
\end{definition}

\begin{note}
All connected Lie subgroups of $\mathrm{Sp}(2)$ with Lie algebra $\mathfrak{sp}(1)\oplus\mathfrak{sp}(1)$ are conjugate to $\mathrm{Sp}(1)\times\mathrm{Sp}(1)$. That is, there is some $A\in\mathrm{Sp}(2)$ such that the Lie group is of the form $A(\mathrm{Sp}(1)\times\mathrm{Sp}(1))A^\dagger$. Instantons equivariant under this group are of the form $A^\dagger.(\hat{M},U)$, where $\hat{M}$ is equivariant under $\mathrm{Sp}(1)\times\mathrm{Sp}(1)$.
\end{note}

\begin{proof}
We use the same notation as introduced in the beginning of the proof of Theorem~\ref{thm:circularsym}. We follow the proof of Theorem~\ref{thm:mainthm} after setting the scene.

Suppose that $\hat{M}$ has rotational symmetry. Let $S\subseteq \mathrm{Sp}(1)\times \mathrm{Sp}(1)\times\mathrm{O}(k)$ be the stabilizer group of $(\hat{M},U)$ restricted to rotations in $\mathrm{Sp}(1)\times \mathrm{Sp}(1)\subseteq\mathrm{Sp}(2)$. That is
\begin{equation}
S:=\{(\mathrm{diag}(p,q),K)\mid (\mathrm{diag}(LqKL^\dagger (LL^\dagger)^{-1},p K),K).\mathrm{diag}(p^\dagger,q^\dagger).(\hat{M},U)
=(\hat{M},U)\}.\label{eq:rotationalsymS}
\end{equation}
Indeed, if $(Q,K)\in\mathrm{Sp}(n+k)\times\mathrm{GL}(k,\mathbb{R})$ such that $\mathrm{diag}(p^\dagger,q^\dagger).(Q,K).(\hat{M},U)=(\hat{M},U)$, then Corollary~\ref{cor:IsomS} tells us that $K\in\mathrm{O}(k)$, $[R,K]=0$, and $Q=\mathrm{diag}(LqKL^\dagger (LL^\dagger)^{-1},p K)$. Therefore, the pairs in $S$ encapsulate all the rotational symmetry of the instanton. That $S$ is a group follows from this fact as well as the fact that the conformal action is a right Lie group action, just as in the proof of Theorem~\ref{thm:simplesphericalsym}.

Just like our previous symmetries, other than circular $t$-symmetry, the stabilizer group is a subgroup of a compact group, so we proceed as in the proof of Theorem~\ref{thm:mainthm}~\cite[Theorem~1.1]{lang_moduli_2024}. In particular, we find that $\hat{M}$ has rotational symmetry if and only if there is a Lie algebra homomorphism $\rho\colon\mathfrak{sp}(1)\oplus\mathfrak{sp}(1)\rightarrow\mathfrak{so}(k)$ such that for all $\upsilon,\omega\in\mathfrak{sp}(1)$, we have
\begin{equation}
\left(\begin{bmatrix}
L\omega L^\dagger (LL^\dagger)^{-1}L+L\rho(\upsilon,\omega)L^\dagger (LL^\dagger)^{-1}L-L\rho(\upsilon,\omega)-L\omega \\
\upsilon M+[\rho(\upsilon,\omega),M]-M\omega
\end{bmatrix},0\right)=(0,0).\label{eq:rotationalsymdiff}
\end{equation}

We can simplify these constraints. In particular, suppose $\hat{M}$ has rotational symmetry. In Section~\ref{subsubsec:RotSymStructureGroups}, we study the top row. For now, focusing on the bottom row, we see that $\upsilon M+[\rho(\upsilon,\omega),M]-M\omega=0$. Furthermore, as mentioned above, Corollary~\ref{cor:IsomS} tells us that $[R,e^{\theta\rho(\upsilon,\omega)}]=0$, for all $\theta\in\mathbb{R}$. Differentiating and evaluating at $\theta=0$, we have $[R,\rho(\upsilon,\omega)]=0$. 

We can use the same method as in the proof of Theorem~\ref{thm:circularsym} to prove the converse.
\end{proof}

\begin{note}
All instantons with rotational symmetry correspond to axially symmetric hyperbolic monopoles because there is a $S^1\times S^1$ subgroup of $\mathrm{Sp}(1)\times\mathrm{Sp}(1)$ comprised of two commuting circle actions conjugate to $R_1$. The rest of $\mathrm{Sp}(1)\times\mathrm{Sp}(1)$ does not commute with either of these circle actions. Thus, the rest of the group does not impart any additional symmetry. While not all axially symmetric hyperbolic monopoles are described by instantons with rotational symmetry, the extra restrictions make it easier to find such monopoles.

Additionally, all instantons with rotational symmetry correspond to spherically symmetric singular monopoles with Dirac type singularities because there is an isoclinic superspherical subgroup of $\mathrm{Sp}(1)\times\mathrm{Sp}(1)$. The rest of $\mathrm{Sp}(1)\times\mathrm{Sp}(1)$ does not impart any additional symmetry. While not all spherically symmetric singular monopoles are described by instantons with rotational symmetry, the extra restrictions make it easier to find such monopoles. 

Unlike with the instantons with isoclinic superspherical symmetry, any choice of subgroup conjugate to $R_0$ leads to a spherically symmetric singular monopole. Finally, any choice of isoclinic spherical symmetry subgroup leads to an axially symmetric hyperbolic Higgs bundle.
\end{note}

\begin{note}
Just as with our previous symmetries, we need not check the final condition of $\mathcal{M}_{n,k}$ in Definition~\ref{def:Mstd} everywhere. Indeed, if $\hat{M}$ satisfies \eqref{eq:rotationalsym1} and \eqref{eq:rotationalsym2} for some real representation $(\mathbb{R}^k,\rho)$, then it satisfies \eqref{eq:isoclinicsphericalsym1} and \eqref{eq:isoclinicsphericalsym2} for some real representation of $\mathfrak{sp}(1)$. Thus, just as in the isoclinic spherical symmetry case, we need only check that the final condition of $\mathcal{M}_{n,k}$ in Definition~\ref{def:Mstd} is satisfied for all $x=x_0\in\mathbb{H}$ with $x_0\geq 0$. 
\end{note}

\subsection{Structure of rotational symmetry}

In light of Theorem~\ref{thm:rotationalsym}, we know that we can find instantons with rotational symmetry. That is, starting with a real representation of $\mathfrak{spin}(4)$, we can narrow down the possible $M$, so we are only left with finding $L$ such that we have $\hat{M}\in\mathcal{M}_{n,k}$, which necessarily has rotational symmetry. In this section, we investigate what representations generate such instantons and what the corresponding ADHM data looks like. 

\begin{note}
The representation theory of $\mathfrak{sp}(1)\oplus\mathfrak{sp}(1)$ is closely related to that of $\mathfrak{sp}(1)$. In particular, for every $m,n\in\mathbb{N}_+$, there is a unique, up to isomorphism, irreducible complex $mn$-representation $(V_{m,n},\rho_{m,n})$ with highest weight $\begin{bmatrix}
\frac{m-1}{2} & \frac{n-1}{2}
\end{bmatrix}$. In particular, recalling that we denote the irreducible complex $n$-representation by $(V_n,\rho_n)$, we can take $V_{m,n}:=V_m\otimes V_n$ and $\rho_{m,n}:=\rho_m\otimes\mathrm{id}_{V_n}+\mathrm{id}_{V_m}\otimes \rho_n$~\cite[Proposition~D.3.8]{lang_thesis_2024}.

Additionally, for every $m,n\in\mathbb{N}_+$, if $n\equiv m\textrm{ (mod 2)}$, then there is a unique, up to isomorphism, irreducible real $mn$-representation $(\mathbb{R}^{mn},\varrho_{m,n})$. However, if $n\not\equiv m\textrm{ (mod 2)}$, then there is a unique, up to isomorphism, irreducible real $2mn$-representation $(\mathbb{R}^{2mn},\varrho_{m,n})$. In the former case, the complexification of this representation is isomorphic to $(V_{m,n},\rho_{m,n})$ and in the latter, the complexification is isomorphic to $(V_{m,n},\rho_{m,n})^{\oplus 2}$.
\end{note}

\begin{definition}
Recall the standard basis of $\mathfrak{sp}(1)$ given in Definition~\ref{def:stdbasissp1}. Let $\upsilon_1,\ldots,\upsilon_6$ be the \textbf{standard basis} for $\mathfrak{sp}(1)\oplus\mathfrak{sp}(1)$. That is, $\upsilon_1,\upsilon_2,\upsilon_3$ give the standard basis for $\mathfrak{sp}(1)\oplus 0$ and $\upsilon_4,\upsilon_5,\upsilon_6$ give the standard basis for $0\oplus\mathfrak{sp}(1)$.
\end{definition}

\begin{definition}
Let $V:=\mathbb{R}^k$. Consider a $k$-representation $(V,\rho)$ of $\mathfrak{sp}(1)\oplus\mathfrak{sp}(1)$. Consider the real representation $(\mathbb{H},\nu)$ of $\mathfrak{sp}(1)\oplus\mathfrak{sp}(1)$ given by $\nu(\upsilon,\omega)(x)=\upsilon x-x\omega$. Define the induced real representation 
\begin{equation}
(\hat{V},\hat{\rho}):=(V,\rho)\otimes_\mathbb{R} (V^*,\rho^*)\otimes_\mathbb{R} (\mathbb{H},\nu).
\end{equation}
Note that $(\mathbb{H},\nu)\simeq (\mathbb{R}^4,\varrho_{2,2})$. Explicitly, the action of $\hat{\rho}(\upsilon,\omega)$ on $A\in\hat{V}=\mathrm{Mat}(k,k,\mathbb{H})$ is 
\begin{equation}
\hat{\rho}(\upsilon,\omega)A= \upsilon A-A\omega+ [\rho(\upsilon,\omega),A].
\end{equation}
\end{definition}

\begin{note}
As $\mathfrak{sp}(1)\oplus\mathfrak{sp}(1)\simeq\mathfrak{spin}(4)$ has Dynkin diagram $D_2$, all representations of $\mathfrak{sp}(1)\oplus\mathfrak{sp}(1)$ are self-dual.
\end{note}

The definition of $(\hat{V},\hat{\rho})$ is well-motivated. Firstly, note that $M\in\hat{V}$. Secondly, given the connection between the action of $\hat{\rho}(\upsilon,\omega)$ and \eqref{eq:rotationalsym1}, we immediately obtain the following corollary.
\begin{cor}
Let $\hat{M}\in\mathcal{M}$. Then $\hat{M}$ has rotational symmetry if and only if there is some real $k$-representation $(V,\rho)$ of $\mathfrak{sp}(1)\oplus\mathfrak{sp}(1)$ such that $[R,\rho(\upsilon,\omega)]=0$ and $\hat{\rho}(\upsilon,\omega)(M)=0$ for all $\upsilon,\omega\in\mathfrak{sp}(1)$.\label{cor:spin4triv}
\end{cor}

\begin{note}
Suppose $\hat{M}\in\mathcal{M}$ has rotational symmetry. Corollary~\ref{cor:spin4triv}, tells us that there is some real representation $(V,\rho)$ such that $\mathrm{span}(M)\subseteq \hat{V}$ is an invariant subspace, which is acted on trivially.

As $\mathfrak{sp}(1)\oplus \mathfrak{sp}(1)$ is semi-simple, the representation $(\hat{V},\hat{\rho})$ decomposes into irreducible representations. We have previously explored the $M=0$ case, so suppose $M\neq 0$. As $\mathrm{span}(M)$ is a one-dimensional invariant subspace, acted on trivially, $(\mathrm{span}(M),0)$ is a summand of the representation. So the decomposition of $(\hat{V},\hat{\rho})$ must contain trivial summands. Moreover, $M$ is in the direct sum of these trivial summands. In fact, if $M=0$, then it is in the direct sum of trivial summands as well.
\end{note}

In particular, we note the following representations. The trivial representation $(\mathbb{C},0)$ is the representation $(V_{1,1},\rho_{1,1})$. The real representation $(\mathbb{H},\nu)$, where $\nu(\upsilon,\omega)(x):=\upsilon x-x\omega$, is the representation $(\mathbb{R}^4,\varrho_{2,2})$, whose complexification is $(V_{2,2},\rho_{2,2})$. Consider the representations $(\mathbb{H},\iota_1)$ and $(\mathbb{H},\iota_2)$, where $\iota_1(\upsilon,\omega)(x):=\upsilon x$ and $\iota_2(\upsilon,\omega)(x):=\omega x$. Restricting the scalars to $\mathbb{C}$, these are isomorphic to $(V_{2,1},\rho_{2,1})$ and $(V_{1,2},\rho_{1,2})$, respectively.

When searching for instantons with rotational symmetry, $M$ is found in the trivial summands of $(\hat{V},\hat{\rho})$. The following lemmas tell us where these trivial summands are, given a decomposition of $(\hat{V},\hat{\rho})$. First, we see where the trivial summands are when decomposed as a complex representation, providing us with the information needed for the real decomposition. First, we note that we know how to decompose tensor products of irreducible complex representations of $\mathfrak{sp}(1)\oplus\mathfrak{sp}(1)$~\cite[Proposition~D.3.12]{lang_thesis_2024}: 
\begin{equation}
(V_{a,b},\rho_{a,b})\otimes (V_{c,d},\rho_{c,d})\simeq\bigoplus_{i=1}^{\mathrm{min}(a,c)}\bigoplus_{j=1}^{\mathrm{min}(b,d)}(V_{a+c+1-2i,b+d+1-2j},\rho_{a+c+1-2i,b+d+1-2j}).
\end{equation}
\begin{lemma}
Given $a\geq b\geq 1$ and $m,n\geq 1$, consider $(V_{a,m},\rho_{a,m})\otimes_\mathbb{C} (V_{b,n},\rho_{b,n})\otimes_\mathbb{C} (V_{2,2},\rho_{2,2})$. The product has a single trivial summand when $a=b+1$ and $|n-m|=1$. Otherwise, there are no trivial summands.\label{lemma:spin4complextrivsummands}
\end{lemma}

\begin{proof}
First we look at the product $(V_{c,d},\rho_{c,d})\otimes_\mathbb{C} (V_{2,2},\rho_{2,2})$. We see that this product contains a single trivial component when $c=d=2$ and otherwise contains none. Thus, the total product contains a single trivial component for every $(V_{2,2},\rho_{2,2})$ summand in the decomposition of $(V_{a,m},\rho_{a,m})\otimes_\mathbb{C} (V_{b,n},\rho_{b,n})$. 

Examining the decomposition of $(V_{a,m},\rho_{a,m})\otimes_\mathbb{C} (V_{b,n},\rho_{b,n})$, we see that this product contains a single $(V_{2,2},\rho_{2,2})$ summand when $a=b+1$ and $|n-m|=1$. Otherwise, it contains none, proving the lemma.
\end{proof}

We first investigate the trivial summand present in $(V_{m+1,n\pm 1},\rho_{m+1,n\pm 1})\otimes_\mathbb{C} (V_{m,n}^*,\rho_{m,n}^*)\otimes_\mathbb{C} (V_{2,2},\rho_{2,2})$.
\begin{definition}
Let $B^{m,n,\pm}$ be the unit length generator of the unique one-dimensional representation in $(V_{m+1,n\pm 1},\rho_{m+1,n\pm 1})\otimes_\mathbb{C} (V_{m,n}^*,\rho_{m,n}^*)\otimes_\mathbb{C} (V_{2,2},\rho_{2,2})$. This generator is well-defined, up to a $\mathbb{C}^*$ factor, though a choice of scale leaves a $S^1$ factor. This generator can be thought of as a $\mathfrak{sp}(1)\oplus\mathfrak{sp}(1)$-invariant quadruple $(B_\mu^{m,n,\pm})_{\mu=0}^3$ of complex $(m+1)(n\pm 1)\times mn$ matrices. 
\end{definition}
These generators can be identified in much the same way as we compute the $B^n$ matrices in previous work~\cite[Appendix~E]{lang_thesis_2024}.

Above, we defined what spans the trivial summands of $(\hat{V},\hat{\rho})$ when decomposed as a complex representation. We use this definition to determine what spans the real decomposition.

Let $N(c,d):=\frac{3-(-1)^{c+d}}{2}cd$. That is, $N(c,d)=cd$ if $c$ and $d$ have the same parity and $N(c,d)=2cd$ otherwise. 
\begin{lemma}
Given $a\geq b\geq 1$ and $m,n\geq 1$, we have that 
\begin{equation*}
(\mathbb{R}^{N(a,m)},\varrho_{a,m})\otimes_\mathbb{R} ((\mathbb{R}^{N(b,n)})^*,\varrho_{b,n}^*)\otimes_\mathbb{R} (\mathbb{R}^4,\varrho_{2,2})
\end{equation*} 
possesses trivial summands when $a=b+1$ and $|n-m|=1$. Otherwise, there are no trivial summands. More specifically, if $a$ and $m$ have the same parity, then there is a single trivial summand. If $a$ and $m$ have different parities, then there are four trivial summands.\label{lemma:howmanytrivsummandsSpin4}
\end{lemma}

\begin{proof}
Suppose that $a\neq b+1$ or $|n-m|\neq 1$. Regardless of the relative parities of $a$, $b$, $m$, and $n$, the complexification of the real representations does not contain any trivial summands. Henceforth, we assume that $a=b+1$ and $|n-m|=1$.

Suppose that $a$ and $m$, and hence $b$ and $n$, have the same parity. Then the complexification of the real representation is isomorphic, as a complex representation, to $(V_{a,m},\rho_{a,m})\otimes_\mathbb{C} (V_{b,n},\rho_{b,n})\otimes_\mathbb{C} (V_{2,2},\rho_{2,2})$. By Lemma~\ref{lemma:spin4complextrivsummands}, we know that this tensor product has a single trivial summand. Furthermore, as this tensor product only contains summands $(V_{c,d},\rho_{c,d})$ with $c$ and $d$ having the same parity, the representation is real, meaning that the real decomposition is the same as the complex. Thus, we have proven the first part of the lemma.

Suppose that $a$ and $m$, and hence $b$ and $n$, have different parity. Then the complexification of the real representation is isomorphic, as a complex representation, to $(V_{a,m},\rho_{a,m})^{\oplus 2}\otimes_\mathbb{C} (V_{b,n},\rho_{b,n})^{\oplus 2}\otimes_\mathbb{C} (V_{2,2},\rho_{2,2})$. By Lemma~\ref{lemma:spin4complextrivsummands}, we know that this tensor product has four trivial summands. Furthermore, as this tensor product only contains summands $(V_{c,d},\rho_{c,d})$ with $c$ and $d$ having the same parity, the representation is real, meaning that the real decomposition is the same as the complex. Thus, we have proven the second part of the lemma.
\end{proof}

Using Lemma~\ref{lemma:howmanytrivsummandsSpin4}, we investigate the trivial summands of 
\begin{equation*}
(\mathbb{R}^{N(m+1,n\pm 1)},\varrho_{m+1,n\pm 1})\otimes_\mathbb{R} ((\mathbb{R}^{N(m,n)})^*,\varrho_{m,n}^*)\otimes_\mathbb{R} (\mathbb{R}^4,\varrho_{2,2}).
\end{equation*} 
\begin{lemma}
Let $m,n\in\mathbb{N}_+$. If $m$ and $n$ have the same parity, then there is a unique choice of $S^1$ factor, up to sign, such that the $B_i^{m,n,\pm}$ matrices are real. Then the quadruple spans the unique trivial summand of $(\mathbb{R}^{N(m+1,n\pm 1)},\varrho_{m+1,n\pm 1})\otimes_\mathbb{R} ((\mathbb{R}^{N(m,n)})^*,\varrho_{m,n}^*)\otimes_\mathbb{R} (\mathbb{R}^4,\varrho_{2,2})$. 

If $m$ and $n$ have different parity, then $N(m,n)=2mn$ and $N(m+1,n\pm 1)=2(m+1)(n\pm 1)$. Let $Y_i^+:=\varrho_{m+1,n\pm 1}(\upsilon_i)\in\mathfrak{so}(2(m+1)(n\pm 1))$ and $Y_i^-:=\varrho_{m,n}(\upsilon_i)\in\mathfrak{so}(2mn)$. Let $y_i^+:=\rho_{m+1,n\pm 1}(\upsilon_i)\in\mathfrak{su}((m+1)(n\pm 1))$ and $y_i^-:=\rho_{m,n}(\upsilon_i)\in\mathfrak{su}(mn)$. As the complexification of $(\mathbb{R}^{2(m+1)(n\pm 1)},\varrho_{m+1,n\pm 1})$ is isomorphic to $(V_{m+1,n\pm 1},\rho_{m+1,n\pm 1})^{\oplus 2}$ and similarly for $(\mathbb{R}^{2mn},\varrho_{m,n})$, there exists $U_+\in\mathrm{SU}((m+1)(n\pm 1))$ and $U_-\in\mathrm{SU}(mn)$ such that 
\begin{equation*}
Y_i^\pm=U_\pm^\dagger \begin{bmatrix}
y_i^\pm & 0 \\ 0 & y_i^\pm 
\end{bmatrix}U_\pm. 
\end{equation*} 
We can find four linearly independent quadruples of real matrices spanning the four trivial summands of $(\mathbb{R}^{N(m+1,n\pm 1)},\varrho_{m+1,n\pm 1})\otimes_\mathbb{R} ((\mathbb{R}^{N(m,n)})^*,\varrho_{m,n}^*)\otimes_\mathbb{R} (\mathbb{R}^4,\varrho_{2,2})$. For some choices of $\alpha,\beta,\gamma,\delta\in\mathbb{C}$, these linearly independent quadruples of matrices are given by
\begin{equation}
\left(U_+^\dagger \begin{bmatrix}
\alpha B_\mu^{m,n,\pm} & \beta B_\mu^{m,n,\pm} \\ 
\gamma B_\mu^{m,n,\pm} & \delta B_\mu^{m,n,\pm}
\end{bmatrix}U_-\right)_{\mu=0}^3.
\end{equation}
\end{lemma}

\begin{proof}
The proof follows from previous work on hyperbolic monopoles~\cite[Lemma~9]{lang_hyperbolic_2023}.
\end{proof}

Now that we know exactly what spans the trivial summands of $(\hat{V},\hat{\rho})$, the following theorem tells us exactly what form $M$ takes if $\hat{M}$ has rotational symmetry. 
\begin{theorem}[Rotational Structure Theorem]
Let $\hat{M}\in\mathcal{M}_{n,k}$ have rotational symmetry. Theorem~\ref{thm:rotationalsym} tells us that $\hat{M}$ is generated by a real representation $(V,\rho)$ of $\mathfrak{sp}(1)\oplus\mathfrak{sp}(1)$, which we can decompose as
\begin{equation}
(V,\rho)\simeq \bigoplus_{a=1}^k (\mathbb{R}^{N(m_a,n_a)},\varrho_{m_a,n_a}).\label{eq:repdecompspin4}
\end{equation}
Without loss of generality, we may assume that $m_a\geq m_{a+1}$, $\forall a\in\{1,\ldots,I-1\}$.

Let $Y_{i,a}:=\rho_{m_a,n_a}(\upsilon_i)$ and
\begin{equation*}
Y_i:=\mathrm{diag}(Y_{i,1},\ldots,Y_{i,k}).
\end{equation*}
Then $Y_i$ induces $(V,\rho)$. When $n_a$ and $m_a$ have different parity, let $y_{i,a}:=\rho_{m_a,n_a}(\upsilon_i)\in\mathfrak{su}(n_am_a)$ induce the corresponding irreducible complex representation. Then there is some $U_{m_a,n_a}\in\mathrm{SU}(n_am_a)$ such that 
\begin{equation*}
Y_{i,a}=U^\dagger_{m_a,n_a}\mathrm{diag}(y_{i,a},y_{i,a})U_{m_a,n_a}.
\end{equation*}

Using the decomposition of $(V,\rho)$ given in \eqref{eq:repdecompspin4}, we have $(M_\mu)_{ab}\in\mathrm{R}^{N(m_a,n_a)}\otimes_\mathbb{R} (\mathbb{R}^{N(m_a,n_a)})^*$ such that 
\begin{equation*}
M_\mu=\begin{bmatrix}
(M_\mu)_{11} & \cdots & (M_\mu)_{1k} \\
\vdots & \ddots & \vdots \\
(M_\mu)_{k1} & \cdots & (M_\mu)_{kk}
\end{bmatrix}.
\end{equation*}

Then, up to a $\rho$-invariant gauge and for all $a,b\in\{1,\ldots,I\}$ we have that
\begin{enumerate}
\item[(1)] if $a<b$, $m_a=m_b+1$, $n_a=n_b\pm 1$, and $m_a$ and $n_a$ have the same parity, then $\exists \lambda_{ab}\in\mathbb{R}$ such that 
\begin{align*}
(M_\mu)_{ab}&=\lambda_{ab}B_\mu^{m_b,n_b,\pm}, \quad\textrm{and}\\
(M_\mu)_{ba}&=\lambda_{ab}(B_\mu^{m_b,n_b,\pm})^T,
\end{align*}
where the $S^1$ factor has been chosen so that the $B_\mu^{m_b,n_b,\pm}$ matrices are real;
\item[(2)] if $a<b$, $m_a=m_b+1$, $n_a=n_b\pm 1$, and $m_a$ and $n_a$ have different parity, then $\exists \kappa_{a,b,0}$,$\kappa_{a,b,1}$, $\kappa_{a,b,2}$,$\kappa_{a,b,3}\in\mathbb{C}$ such that 
\begin{align*}
(M_\mu)_{ab}&=U_{m_a,n_a}^\dagger \begin{bmatrix}
\kappa_{a,b,0} B_\mu^{m_b,n_b,\pm} & \kappa_{a,b,1} B_\mu^{m_b,n_b,\pm} \\ 
\kappa_{a,b,2} B_\mu^{m_b,n_b,\pm} & \kappa_{a,b,3} B_\mu^{m_b,n_b,\pm}
\end{bmatrix}U_{m_b,n_b}, \quad\textrm{and}\\
(M_\mu)_{ba}&=U_{m_b,n_b}^\dagger \begin{bmatrix}
\overline{\kappa_{a,b,0}} (B_\mu^{m_b,n_b,\pm})^\dagger & \overline{\kappa_{a,b,2}} (B_\mu^{m_b,n_b,\pm})^\dagger \\ 
\overline{\kappa_{a,b,1}} (B_\mu^{m_b,n_b,\pm})^\dagger & \overline{\kappa_{a,b,3}} (B_\mu^{m_b,n_b,\pm})^\dagger
\end{bmatrix}U_{m_b,n_b},
\end{align*}
and these blocks are real;
\item[(3)] otherwise, $(M_\mu)_{ab}=0$.
\end{enumerate}

Conversely, if $M$ has the above form for some real representation $(V,\rho)$, and for all $\upsilon_1,\upsilon_2\in\mathfrak{sp}(1)$, $[\rho(\upsilon_1,\upsilon_2),R]=0$, then $\hat{M}$ has rotational symmetry.
\end{theorem}

\begin{proof}
The proof follows the same structure as the proof of the Structure Theorem for hyperbolic monopoles~\cite[Theorem~3]{lang_hyperbolic_2023}.
\end{proof}

\subsection{Novel example of rotational symmetry}\label{subsec:novelrot}

In this section, we identify novel examples of rotational symmetry, using the Rotational Structure Theorem.

\begin{prop}\label{prop:RotInstEx}
Let $\lambda>0$. Then define $\hat{M}$ by
\begin{equation}
\begin{aligned}
M&:=\lambda\begin{bmatrix}
0 & 0 & 0 & 0 & 1 \\
0 & 0 & 0 & 0 & i \\
0 & 0 & 0 & 0 & j \\
0 & 0 & 0 & 0 & k \\
1 & i & j & k & 0
\end{bmatrix} \quad\textrm{and}\\
L&:=\lambda\begin{bmatrix}
\sqrt{3} & -\frac{i}{\sqrt{3}} & -\frac{j}{\sqrt{3}} & -\frac{k}{\sqrt{3}} & 0 \\
0 & 2\sqrt{\frac{2}{3}} & k\sqrt{\frac{2}{3}} & -j\sqrt{\frac{2}{3}} & 0 \\
0 & 0 & \sqrt{2} & i\sqrt{2} & 0 
\end{bmatrix}.
\end{aligned}
\end{equation}
We have $\hat{M}\in\mathcal{M}_{3,5}$ and it is a $\mathrm{Sp}(3)$ instanton with rotational symmetry and instanton number $5$. 
\end{prop}

\begin{proof}
Consider the representation $(V,\rho):=(\mathbb{H},\nu)\oplus(\mathbb{R},0)$. Seeking instantons with rotational symmetry generated by $(V,\rho)$, the Rotational Structure Theorem tells us that $M$ has to have the form given in the statement. Thus, if $\hat{M}\in\mathcal{M}_{3,5}$, then it has rotational symmetry. To that end, we have that $LL^\dagger=4\lambda^2 I_3$. Additionally, $R=4\lambda^2 I_5$. Due to the symmetry of the instanton, we need only check the final condition for $x=x_0\geq 0$. Doing so, we see that 
\begin{equation*}
\Delta(x)^\dagger \Delta(x)=\begin{bmatrix}
4\lambda^2+x_0^2 & 0 & 0 & 0 & -2x_0\lambda \\
0 & 4\lambda^2+x_0^2 & 0 & 0 & 0 \\
0 & 0 & 4\lambda^2+x_0^2 & 0 & 0 \\
0 & 0 & 0 & 4\lambda^2+x_0^2 & 0 \\
-2x_0\lambda & 0 & 0 & 0 & 4\lambda^2+x_0^2
\end{bmatrix}.
\end{equation*}
This matrix has eigenvalues $4\lambda^2+x_0^2$, $4\lambda^2\pm 2x_0\lambda+x_0^2$ with multiplicity 3, 1, and 1, respectively. Therefore, $\Delta(x)^\dagger \Delta(x)$ is positive-definite everywhere, proving the proposition. 
\end{proof}

\subsection{Constraint on rotational symmetry}\label{subsubsec:RotSymStructureGroups}

In Theorem~\ref{thm:rotationalsym}, we obtain a representation $(V,\rho)$ of $\mathfrak{sp}(1)\oplus\mathfrak{sp}(1)$ from an instanton with rotational symmetry by focusing on the bottom of \eqref{eq:rotationalsymdiff}. This representation induces another $(\hat{V},\hat{\rho})$, which we use to determine what an instanton with rotational symmetry looks like. However, by focusing on the top of \eqref{eq:rotationalsymdiff}, we obtain a second representation.
\begin{lemma}
Let $\pi_l\colon\mathfrak{sp}(1)\oplus\mathfrak{sp}(1)\rightarrow\mathfrak{sp}(1)$ be projection onto the $l$th component, for $l\in\{1,2\}$. Suppose $(\hat{M},U)\in\mathcal{M}_{n,k}$ has rotational symmetry generated by $(V,\rho)$, a real representation. Let $y_i:=(LL^\dagger)^{-1}L(\pi_2(\upsilon_i)I_k+\rho(\upsilon_i))L^\dagger\in\mathfrak{sp}(n)$. Then $y_i$ induce a quaternionic representation of $\mathfrak{sp}(1)\oplus\mathfrak{sp}(1)$, which we denote by $(W,\lambda)$, where $W:=\mathbb{H}^n$ and $\lambda$ is the linear map taking $\upsilon_i\mapsto y_i$.
\end{lemma}

\begin{proof}
Let $i,j\in\{1,\ldots,6\}$. By Theorem~\ref{thm:rotationalsym}, we know that $[\rho(\upsilon_i),M]=M\pi_2(\upsilon_i)-\pi_1(\upsilon_i)M$, so,
\begin{align*}
M^\dagger M(\pi_2(\upsilon_i)+\rho(\upsilon_i))&=M^\dagger [M,\rho(\upsilon_i)]+M^\dagger \rho(\upsilon_i)M+M^\dagger M \pi_2(\upsilon_i)\\
&=M^\dagger \pi_1(\upsilon_i)M+[M^\dagger,\rho(\upsilon_i)]M+\rho(\upsilon_i)M^\dagger M\\
&=(\pi_2(\upsilon_i)I_k+\rho(\upsilon_i))M^\dagger M.
\end{align*}
Thus, $[M^\dagger M,\pi_2(\upsilon_i)I_k+\rho(\upsilon_i)]=0$. Also from Theorem~\ref{thm:rotationalsym}, $[\rho(\upsilon_i),R]=0$. As $R$ is real, $[R,\pi_2(\upsilon_i)]=0$. Hence, $[L^\dagger L,\pi_2(\upsilon_i)I_k+\rho(\upsilon_i)]=0$. Thus
\begin{equation*}
LL^\dagger L(\rho(\upsilon_i)+\pi_2(\upsilon_i))L^\dagger=L(\rho(\upsilon_i)+\pi_2(\upsilon_i))L^\dagger LL^\dagger.
\end{equation*}
Hence, $[LL^\dagger,L(\rho(\upsilon_i)+\pi_2(\upsilon_i))L^\dagger]=0$, so $y_i\in\mathfrak{sp}(n)$. Therefore, we see
\begin{align*}
[y_i,y_j]&=(LL^\dagger)^{-1}[L(\rho(\upsilon_i)+\pi_2(\upsilon_i))L^\dagger,L(\rho(\upsilon_j)+\pi_2(\upsilon_j))L^\dagger ](LL^\dagger)^{-1}\\
&=(LL^\dagger)^{-1}L((\rho(\upsilon_i)+\pi_2(\upsilon_i))L^\dagger L(\rho(\upsilon_j)+\pi_2(\upsilon_j))\\
&\phantom{=(LL^\dagger)^{-1}L(}-(\rho(\upsilon_j)+\pi_2(\upsilon_j))L^\dagger L(\rho(\upsilon_i)+\pi_2(\upsilon_i)))L^\dagger (LL^\dagger)^{-1}\\
&=(LL^\dagger)^{-1}L[(\rho(\upsilon_i)+\pi_2(\upsilon_i)),(\rho(\upsilon_j)+\pi_2(\upsilon_j))]L^\dagger\\
&=(LL^\dagger)^{-1}L(\rho([\upsilon_i,\upsilon_j])+\pi_2([\upsilon_i,\upsilon_j]))L^\dagger.
\end{align*}
Thus, the $y_i$ satisfy the correct commutation relations. 
\end{proof}

We use $(W,\lambda)$ to determine what $\mathrm{Sp}(n)$ structure groups are possible given $\rho$. Using $W=\mathbb{H}^n\simeq \mathbb{C}^{2n}$, we can restrict the scalars from $\mathbb{H}$ to $\mathbb{C}$. In doing so, the generators $y_i\in\mathfrak{sp}(n)$ correspond to elements of $\mathfrak{su}(2n)$. Thus, the induced complex representation is a $2n$-representation. The complex representation obtained by restricting the scalars of an irreducible, quaternionic representation is either isomorphic to $(V_{a,b},\rho_{a,b})$ for some $a$ and $b$ with opposite parity or $(V_{a,b},\rho_{a,b})^{\oplus 2}$ for some $a$ and $b$ with the same parity~\cite[Appendix~D.3.1]{lang_thesis_2024}.
\begin{definition}
Recall the representation $(\mathbb{H},\iota_2)$, where $\iota_2(\upsilon,\omega)(x)=\omega x$. Restricting the scalars to $\mathbb{C}$, this representation is isomorphic to $(V_{1,2},\rho_{1,2})$.

Given a real representation $(V,\rho)$ of $\mathfrak{sp}(1)\oplus\mathfrak{sp}(1)$, with induced representation $(W,\lambda)$ as defined above, we define the real representation 
\begin{equation}
(\hat{W},\hat{\lambda}):=\left((W,\lambda)\otimes_\mathbb{R} (V^*,\rho^*)\right)\otimes_\mathbb{H} (\mathbb{H}^*,\iota_2^*)=(W,\lambda)\otimes_\mathbb{H}\left((V^*,\rho^*)\otimes_\mathbb{R}(\mathbb{H}^*,\iota_2^*)\right).
\end{equation}
\end{definition}

Unravelling how $\hat{\lambda}$ acts on $\hat{W}=\mathrm{Mat}(n,k,\mathbb{H})$, let $\upsilon,\omega\in\mathfrak{sp}(1)$ and $A\in\hat{W}$. Then
\begin{equation}
\hat{\lambda}(\upsilon,\omega)\left(A\right)=\lambda(\upsilon,\omega)A-A\rho(\upsilon,\omega) -A\omega.
\end{equation}

The definition of $(\hat{W},\hat{\lambda})$ is well-motivated. First, note that $L\in\hat{W}$. Moreover, the following lemma tells us exactly where $L$ lives in $\hat{W}$.
\begin{lemma}
Let $\hat{M}\in\mathcal{M}_{n,k}$ have rotational symmetry. Let $(V,\rho)$, $(W,\lambda)$ and $(\hat{W},\hat{\lambda})$ be as above. Then $\hat{\lambda}(\upsilon,\omega)(L)=0$ for all $(\upsilon,\omega)\in\mathfrak{sp}(1)\oplus\mathfrak{sp}(1)$.
\end{lemma}

\begin{proof}
As $\hat{M}$ has rotational symmetry, we have \eqref{eq:rotationalsymdiff}. Then we see that for $\upsilon,\omega\in\mathfrak{sp}(1)$,
\begin{equation*}
\hat{\lambda}(\upsilon,\omega)(L)=\lambda(\upsilon,\omega)L-L\rho(\upsilon,\omega)-L\omega=0.\qedhere
\end{equation*}
\end{proof}

\begin{cor}
Let $\hat{M}\in\mathcal{M}_{n,k}$ have rotational symmetry. Then $(\hat{W},\hat{\lambda})$ must have trivial summands and $L$ must live in the direct sum of these summands.
\end{cor}

\begin{proof}
As $L\neq 0$, we know $\mathrm{span}(L)\subseteq \hat{W}$ is an invariant 1-dimensional subspace, so $(\mathrm{span}(L),0)$ is a trivial summand of the representation $(\hat{W},\hat{\lambda})$.
\end{proof}

Therefore, we know that $(\hat{W},\hat{\lambda})$ must have trivial summands, which narrows the possibilities for the structure group. Just like in the hyperbolic monopole case, we have straightforward restrictions for high dimensional representations in addition to more subtle restrictions for low dimensional representations. For instance, in Proposition~\ref{prop:RotInstEx}, we examine instantons with rotational symmetry generated by $(V,\rho)\simeq (\mathbb{R}^4,\varrho_{2,2})\oplus (\mathbb{R},0)$. This representation generated instantons with rotational symmetry and structure group $\mathrm{Sp}(3)$. However, this representation cannot generate instantons with rotational symmetry and lower rank structure groups.
\begin{prop}
Consider Proposition~\ref{prop:RotInstEx}, with $(V,\rho)\simeq (\mathbb{R}^4,\varrho_{2,2})\oplus (\mathbb{R},0)$. This representation does not generate an instanton with rotational symmetry and structure group $\mathrm{Sp}(1)$ or $\mathrm{Sp}(2)$. 
\end{prop}

\begin{proof}
Suppose $(V,\rho)\simeq (\mathbb{R}^4,\varrho_{2,2})\oplus (\mathbb{R},0)$ generates an instanton with rotational symmetry. We know that $(\hat{W},\hat{\lambda})$ has a trivial summand. Consider $(V^*,\rho^*)\otimes_\mathbb{R} (\mathbb{H}^*,\iota^*)$. Restricting the scalars, this representation is isomorphic, as a complex representation, to $(V_{1,2},\rho_{1,2})\oplus (V_{2,1},\rho_{2,1})\oplus (V_{2,3},\rho_{2,3})$.

If we have a $\mathrm{Sp}(1)$ instanton, then we have that, after restricting the scalars of $(W,\lambda)$, the representation is isomorphic, as a complex representation, to $(V_{1,1},\rho_{1,1})^{\oplus 2}$, $(V_{1,2},\rho_{1,2})$, or $(V_{2,1},\rho_{2,1})$. The former does not give $(\hat{W},\hat{\lambda})$ any trivial summands. Thus we look to the latter two, which do. 

If the representation obtained by restricting the scalars of $(W,\lambda)$ is given by $(V_{1,2},\rho_{1,2})$, then there is some $q\in\mathrm{Sp}(1)$ such that $y_i=q\pi_2(\upsilon_i)q^\dagger$. Taking $\tilde{L}:=q^\dagger L$ and dropping the tilde, we have 
\begin{equation*}
[\pi_2(\upsilon_i),L]=L\rho(\upsilon_i).
\end{equation*}
The equations when $i\in\{1,2,3\}$ imply that $L=\begin{bmatrix}
0 & 0 & 0 & 0 & d
\end{bmatrix}$ for some $d\in\mathbb{H}$. The equations for $i\in\{4,5,6\}$ imply that $d\in\mathbb{R}$. But we know that as
\begin{equation*}
M=a\begin{bmatrix}
0 & 0 & 0 & 0 & 1 \\
0 & 0 & 0 & 0 & i \\
0 & 0 & 0 & 0 & j \\
0 & 0 & 0 & 0 & k \\
1 & i & j & k & 0
\end{bmatrix},
\end{equation*}
we have
\begin{equation*}
R=L^\dagger L+M^\dagger M=d^2\begin{bmatrix}
0 & 0 & 0 & 0 & 0 \\
0 & 0 & 0 & 0 & 0 \\
0 & 0 & 0 & 0 & 0 \\
0 & 0 & 0 & 0 & 0 \\
0 & 0 & 0 & 0 & 1 \\
\end{bmatrix}+a^2\begin{bmatrix}
1 & i & j & k & 0 \\
-i & 1 & -k & j & 0 \\
-j & k & 1 & -i & 0 \\
-k & -j & i & 1 & 0 \\
0 & 0 & 0 & 0 & 4
\end{bmatrix}
\end{equation*}
As $R$ is real, $a=0$. But then $R$ is singular. Contradiction! 

If the representation obtained by restricting the scalars of $(W,\lambda)$ is given by $(V_{2,1},\rho_{2,1})$ then there is some $q\in\mathrm{Sp}(1)$ such that $y_i=q\pi_1(\upsilon_i)q^\dagger$. Gauging $L$ as above, then we have
\begin{equation*}
\pi_1(\upsilon_i)L-LY_i-L\pi_2(\upsilon_i)=0.
\end{equation*}
The equations when $i\in\{1,2,3\}$ imply that there is some $A\in\mathbb{H}$ such that \begin{equation*}
L=\begin{bmatrix}
1 & i & j & k & 0
\end{bmatrix}A.
\end{equation*} 
The remaining equations imply that $A\in\mathbb{R}$. Then we have
\begin{equation*}
R=A^2\begin{bmatrix}
1 & i & j & k & 0 \\
-i & 1 & -k & j & 0 \\
-j & k & 1 & -i & 0 \\
-k & -j & i & 1 & 0 \\
0 & 0 & 0 & 0 & 0
\end{bmatrix}+a^2\begin{bmatrix}
1 & i & j & k & 0 \\
-i & 1 & -k & j & 0 \\
-j & k & 1 & -i & 0 \\
-k & -j & i & 1 & 0 \\
0 & 0 & 0 & 0 & 4
\end{bmatrix}
\end{equation*}
As $R$ is real, $A=a=0$, so $L=0$. Contradiction! Therefore, we cannot have a $\mathrm{Sp}(1)$ instanton from this representation $(V,\rho)$.

Instead, if we have a $\mathrm{Sp}(2)$ instanton, then we have that the representation obtained by restricting the scalars of $(W,\lambda)$ is some combination of the previously discussed $2$-representations or one of $(V_{4,1},\rho_{4,1})$ or $(V_{1,4},\rho_{1,4})$. However, the only choices for the representation obtained by restricting the scalars of $(W,\lambda)$ that create a trivial summand in $(\hat{W},\hat{\lambda})$ include $(V_{1,2},\rho_{1,2})$ or $(V_{2,1},\rho_{2,1})$. But we already showed that these representations do not generate an instanton with rotational symmetry. Thus, this representation does not generate a $\mathrm{Sp}(2)$ instanton with rotational symmetry either.
\end{proof}

\appendix
\section{Continuous subgroups of conformal transformations}\label{appendix:contsubgroups}
Our ultimate goal is to classify instantons with continuous conformal symmetries. In this appendix, we classify the relevant connected Lie subgroups up to conjugacy, providing a list of continuous conformal symmetries that an instanton can have. If we were just studying instantons symmetric under isometries, the finite action condition would imply that we would just need to study instantons symmetric under rotations~\cite[Corollary~B.0.8]{lang_thesis_2024}. However, the conformal invariance of the self-dual equations allows us to broaden our study to conformal symmetries.

Recall Conjecture~\ref{conj}, which states that up to conjugation by an element of $\mathrm{SL}(2,\mathbb{H})$, the group of symmetries of a non-flat instanton is a subgroup of $\mathrm{Sp}(2)$. Before we move on to classifying the Lie subalgebras of $\mf{sp}(2)$, let us discuss rotations in $\mathbb{R}^4$. In four dimensions there are three main types of rotations.
\begin{definition}
Let $R\in\mathrm{SO}(4)$. Associated to $R$ are two angles $\alpha,\beta\in[0,2\pi)$ and orthogonal two-planes $A$ and $B$ such that $A\oplus B\simeq \mathbb{R}^4$. The planes $A$ and $B$ are invariant under the action of $R$. Specifically, on $A$ and $B$, $R$ rotates about the origin by $\alpha$ and $\beta$, respectively. Note that $R$ is completely determined by its action on $A$ and $B$.

If either of $\alpha$ or $\beta$ is zero, then $R$ is said to be a \textbf{simple} rotation. Without loss of generality, we may assume $\alpha=0$. In this case, $R$ fixes a two-plane, $A$. Every completely orthogonal two-plane $B$, that is where every line in $B$ is orthogonal to every line in $A$, intersects $A$ at a point $P$. Then $R$ rotates elements in $B$ by $\beta$ about the point $P$. 

If $\alpha=\beta>0$, then $R$ is said to be \textbf{isoclinic}. In this case, there are infinitely many invariant two-planes and every vector is rotated by $R$ by $\alpha=\beta$. Otherwise, $R$ is said to be a \textbf{double} rotation.
\end{definition}

\begin{lemma}
For every rotation $R\in\mathrm{SO}(4)$ there exists $p,q\in\mathrm{Sp}(1)$ such that $Rx=pxq^\dagger$ for all $x\in\mathbb{H}\simeq\mathbb{R}^4$. The rotation is isoclinic if and only if $p$ or $q$ is real. The rotation $R$ is simple if and only if $\mathrm{Re}(p)=\mathrm{Re}(q)$. For instance, $x\mapsto -x$ can be constructed with $p=-1$ and $q=1$ and this transformation is not a simple rotation. Indeed, this rotation has no fixed points other than zero. 
\end{lemma}

\begin{proof}
The double cover of $\mathrm{SO}(4)$ is $\mathrm{Spin}(4)\simeq \mathrm{Sp}(1)\times\mathrm{Sp}(1)$, via the map $\pi\colon \mathrm{Sp}(1)\times\mathrm{Sp}(1)\rightarrow \mathrm{SO}(4)$ which acts on $x\in\mathbb{R}^4\simeq \mathbb{H}$ as $pxq^\dagger=\pi(p,q)x$. 

First, we note that for $p\in\mathrm{Sp}(1)$ and $x\in\mf{sp}(1)$, we have $pxp^\dagger\in\mf{sp}(1)$. Thus, if we decompose $x\in\mathbb{H}$ as $x=x_0+\vec{x}$, where $x_0\in\mathbb{R}$ and $\vec{x}\in\mf{sp}(1)$, we decompose $pxp^\dagger=x_0+p\vec{x}p^\dagger$. Note that the double cover $\pi\colon\mathrm{Sp}(1)\rightarrow\mathrm{SO}(3)$ acts on $\mathbb{R}^3\simeq\mf{sp}(1)$ as $pxp^\dagger=\pi(p)x$. Thus, all rotations are of the form $pxp^\dagger$.

Suppose that $R$ is simple. Then there is some $x\in\mathbb{H}\setminus\{0\}$ such that $pxq^\dagger=x$. Then $p=\frac{xqx^\dagger}{|x|^2}$. Using the decomposition from above, $p_0=q_0$. Conversely, suppose that $p_0=q_0$. Let $x\in\mathrm{Sp}(1)$ such that $\vec{p}=x\vec{q}x^\dagger$. We see that
\begin{equation*}
pxq^\dagger=(q_0+x\vec{q}x^\dagger)x(q_0-\vec{q})=x.
\end{equation*}
Thus, the rotation is simple.

Suppose that $R$ is isoclinic. Then, for every $x\in\mathbb{H}$, we have that the angle between $x$ and $pxq^\dagger$ is some $\cos\alpha$. That is, $(x^\dagger pxq^\dagger)_0=|x|^2\cos\alpha$ for all $x\in\mathbb{H}$. Suppose that $p$ and $q$ are not real. Let $x\in\mathrm{Sp}(1)$ and $A>0$ such that $\vec{p}=Ax\vec{q}x^\dagger$. Then
\begin{equation*}
(x^\dagger pxq^\dagger)_0=((p_0+A\vec{q})(q_0-\vec{q}))_0=p_0q_0-A|\vec{q}|^2.
\end{equation*}
However, if we choose $y\in\mathrm{Sp}(1)$ such that $y^\dagger \vec{p}y$ and $\vec{q}$ are orthogonal, then $(y^\dagger \vec{p}y\vec{q})_0=0$. Hence, $(y^\dagger pyq^\dagger)_0=p_0q_0$. Thus, we see that $A|\vec{q}|^2=0$. But $A>0$, so $q\in\mathbb{R}$. But $q$ is not real, contradiction! Hence, $p$ or $q$ is real.

Conversely, suppose that $p$ is real. Then for all $x\in\mathbb{H}$, we have 
\begin{equation*}
(x^\dagger pxq^\dagger)_0=pq_0|x|^2.
\end{equation*}
Alternatively, if $q$ is real, then for all $x\in\mathbb{H}$, we have
\begin{equation*}
(x^\dagger pxq^\dagger)_0=qp_0|x|^2.
\end{equation*}
In either case, we have that the rotation is isoclinic.
\end{proof}

\subsection{Classifying Lie subalgebras of $\mf{sp}(2)$}\label{subsec:classifysubalg}

Before we classify the connected Lie subgroups up to conjugacy, we classify the Lie subalgebras of $\mathfrak{sp}(2)$.
\begin{lemma}
The non-trivial Lie subalgebras of $\mf{sp}(2)$ are $\mathbb{R}$, $\mathbb{R}\oplus\mathbb{R}$, $\mf{sp}(1)$, $\mf{sp}(1)\oplus \mathbb{R}$, and $\mf{sp}(1)\oplus\mf{sp}(1)$.\label{lemma:subalgebraclassify}
\end{lemma}

\begin{proof}
As $\mathrm{Sp}(2)$ is a simple, compact Lie group, its Lie algebra $\mf{sp}(2)$ is simple and compact. Subalgebras of compact Lie algebras are compact, hence reductive. Let $\mf{g}\subseteq\mf{sp}(2)$ be a Lie subalgebra. Then $\mf{g}=\mf{z}(\mf{g})\oplus[\mf{g},\mf{g}]$, where $\mf{z}(\mf{g})$ is abelian and $[\mf{g},\mf{g}]$ is semi-simple. 

Note that the rank of $\mf{sp}(2)$ is two. Thus, $\mathrm{rank}(\mf{g})\leq 2$, so $\mf{z}(\mf{g})$ is one of $0$, $\mathbb{R}$, or $\mathbb{R}\oplus\mathbb{R}$. If $\mf{z}(\mf{g})=\mathbb{R}\oplus\mathbb{R}$, then the rank of $[\mf{g},\mf{g}]$ is zero. As this Lie subalgebra is semi-simple, we have that $[\mf{g},\mf{g}]=0$. Hence, $\mf{g}=\mathbb{R}\oplus\mathbb{R}$. 

Suppose that $\mf{z}(\mf{g})=\mathbb{R}$. Then the rank of $[\mf{g},\mf{g}]$ is at most one. Just as above, if the rank of $[\mf{g},\mf{g}]$ is zero, then $\mf{g}=\mf{z}(\mf{g})=\mathbb{R}$. However, if the rank of $[\mf{g},\mf{g}]$ is one, then as it is semi-simple, $[\mf{g},\mf{g}]=\mf{sp}(1)$. Thus, $\mf{g}=\mf{sp}(1)\oplus\mathbb{R}$. 

Finally, suppose that $\mf{z}(\mf{g})=0$. That is, $\mf{g}$ is semi-simple. Then the rank of $\mf{g}$ is at most two. If the rank of $\mf{g}$ is zero, then $\mf{g}=0$. If the rank of $\mf{g}$ is one, then $\mf{g}=\mf{sp}(1)$. However, we have options should the rank of $\mf{g}$ be two. The only semi-simple Lie algebras of rank two are $\mf{sp}(2)$, $\mf{sp}(1)\oplus\mf{sp}(1)$, $\mf{su}(3)$, and $\mf{g}_2$. However, we note that $\mf{sp}(2)$ is 10-dimensional, so $\mf{g}\subseteq \mf{sp}(2)$ must be at most ten dimensional. As $\mf{g}_2$ is 14-dimensional, we can exclude this possibility. It remains to show that $\mf{su}(3)$ is not a Lie subalgebra of $\mf{sp}(2)$.

Indeed, suppose for the sake of contradiction that $\mf{su}(3)\simeq\mf{g}\subseteq\mf{sp}(2)$. As $\mf{g}$ and $\mf{sp}(2)$ both have rank two, a Cartan subalgebra of $\mf{g}$ is a Cartan subalgebra of $\mf{sp}(2)$. Thus, $\mf{g}$ is a regular Lie subalgebra of $\mf{sp}(2)$~\cite{dynkin_semisimple_1957}. Therefore, $\mf{g}$ has a closed root subsystem of the root system of $\mf{sp}(2)$. However, by looking at their root diagrams, we note that the root system of $\mf{su}(3)$ is not a closed root subsystem of $\mf{sp}(2)$. Contradiction! Therefore, $\mf{su}(3)$ is not a Lie subalgebra of $\mf{sp}(2)$. 
\end{proof}

Consider two Lie subgroups $G$ and $G'$ of $\mathrm{Sp}(2)$ with isomorphic Lie algebras. An instanton equivariant under $G$ does not necessarily correspond to an instanton equivariant under $G'$. To study the different subgroups, we have the following lemma.
\begin{lemma}
Let $x\in\mf{sp}(2)$. Then there is some $A\in\mathrm{Sp}(2)$ such that $AxA^\dagger=\mathrm{diag}(ai,bi)$. Moreover, $A$ can be chosen such that $a\geq b\geq 0$. \label{lemma:conjugatesp2}
\end{lemma}

\begin{proof}
The first claim follows from the fact that a maximal torus of $\mf{sp}(2)$ is given by $\{\mathrm{diag}(ai,bi)\mid a,b\in\mathbb{R}\}$. As $\mathrm{Sp}(2)$ is a matrix group, the adjoint action of $\mathrm{Sp}(2)$ on $\mf{sp}(2)$ is just conjugation. Then, as all adjoint orbits intersect maximal tori, we have the first claim. 

For the second claim, if $x=\mathrm{diag}(ai,bi)$, we can conjugate by $A=jI_2\in\mathrm{Sp}(2)$ to get $AxA^\dagger=\mathrm{diag}(-ai,-bi)$. Additionally, conjugating by $A=\begin{bmatrix}
0 & 1 \\ -1 & 0
\end{bmatrix}$, we have that $AxA^\dagger=\mathrm{diag}(bi,ai)$. Thus, we may assume that $a\geq b\geq 0$.
\end{proof}

Now we are equipped to study the continuous conformal subgroups that we use to study symmetric instantons. From Lemma~\ref{lemma:subalgebraclassify}, we only need to study Lie subgroups with Lie algebra of the form $\mathbb{R}$, $\mathbb{R}\oplus\mathbb{R}$, $\mf{sp}(1)$, $\mf{sp}(1)\oplus\mathbb{R}$, $\mf{sp}(1)\oplus\mf{sp}(1)$, and $\mf{sp}(2)$. 

\subsection{Classifying continuous conformal subgroups}

In this section, we classify the different connected Lie subgroups of $\mathrm{Sp}(2)$ up to conjugacy that we use to study symmetric instantons. 

Suppose we take two connected Lie subgroups whose Lie algebras are conjugate. This conjugacy lifts to the connected Lie subgroups, as the exponential map on $\mathrm{Sp}(2)$ is surjective. As we are dealing with Lie subalgebras of $\mathfrak{sp}(2)$, which naturally act on $\mathbb{H}^2$, these Lie algebras are conjugate if and only if they are isomorphic representations. Therefore, we need only classify Lie subalgebras up to conjugacy.

To start with, we consider subgroups with Lie algebra $\mathbb{R}$.
\begin{prop}
The Lie subgroups of $\mathrm{Sp}(2)$ with Lie algebra $\mathbb{R}$ are conjugate to some member of the family of subgroups $\{R_t\mid 0\leq t\leq 1\}$ with $R_t$ given by
\begin{equation}
R_t:=\left\{\begin{bmatrix}
e^{i\theta} & 0 \\ 0 & e^{ti\theta}
\end{bmatrix}\mid \theta\in\mathbb{R}\right\}.
\end{equation}
\end{prop}

\begin{note}
The subgroup $R_t$ is compact if and only if $t$ is rational. Thus, $R_t$ is isomorphic to $S^1$ when $t$ is rational and $\mathbb{R}$ otherwise. While $R_t$ is not always compact, it turns out that we do not need compactness in this case when investigating symmetric instantons~\cite[Proposition~1.2]{lang_moduli_2024}.
\end{note}

\begin{proof}
Let $G\subseteq \mathrm{Sp}(2)$ be a Lie subgroup with Lie algebra $\mathbb{R}$ generated by $x\in\mf{sp}(2)$. Thus, every element of $G$ is of the form $e^{x\theta}$ for some $\theta\in\mathbb{R}$. By Lemma~\ref{lemma:conjugatesp2}, we can take $x$ to be of the form $\mathrm{diag}(ai,bi)$ for some $a\geq b\geq 0$. If $a=0$, then $G=\{I\}$, which does not have a Lie algebra isomorphic to $\mathbb{R}$. Hence, $a\neq 0$. Let $t:=b/a$. Then $0\leq t\leq 1$. Finally, we see that the Lie algebra generated by $\mathrm{diag}(ai,bi)$ is the same as that generated by $\mathrm{diag}(i,ti)$. 
\end{proof}

Now we move on to subgroups with Lie algebra $\mathbb{R}\oplus\mathbb{R}$. 
\begin{prop}
The Lie subgroups of $\mathrm{Sp}(2)$ with Lie algebra $\mathbb{R}\oplus\mathbb{R}$ are conjugate to the Lie subgroup given by
\begin{equation}
\left\{\begin{bmatrix}
e^{i\theta} & 0 \\ 0 & e^{i\phi}
\end{bmatrix}\mid \theta,\phi\in\mathbb{R}\right\}.
\end{equation}
\end{prop}

\begin{note}
This Lie subgroup is compact and connected. Therefore, it is a maximal torus, isomorphic to $S^1\times S^1$.
\end{note}

\begin{proof}
Note that $\mathbb{R}\oplus\mathbb{R}\subseteq\mf{sp}(2)$ is a Cartan subalgebra. As all Cartan subalgebras are conjugate, we may take the Cartan subalgebra to be given by $\{\mathrm{diag}(i\theta,i\phi)\mid\theta,\phi\in\mathbb{R}\}$. Exponentiating, we get our group.
\end{proof}

Now that we move on to non-abelian Lie subalgebras, we can use representation theory to reduce the list of subgroups to investigate. Next, we study subgroups with Lie algebra $\mf{sp}(1)$.
\begin{prop}
The connected Lie subgroups of $\mathrm{Sp}(2)$ with Lie algebra $\mf{sp}(1)$ are conjugate to one of the unique connected Lie subgroups with the following Lie algebras:\label{prop:sp1subalgebras}
\begin{equation}
\begin{aligned}
\mf{h}_{3,1,1}&:=\left\langle\begin{bmatrix}
i/2 & 0 \\ 0 & i/2
\end{bmatrix},\begin{bmatrix}
j/2 & 0 \\ 0 & j/2
\end{bmatrix},\begin{bmatrix}
k/2 & 0 \\ 0 & k/2
\end{bmatrix}\right\rangle;\\
\mf{h}_{4,1}&:=\left\langle\begin{bmatrix}
i/2 & 0 \\ 0 & 0
\end{bmatrix},\begin{bmatrix}
j/2 & 0 \\ 0 & 0
\end{bmatrix},\begin{bmatrix}
k/2 & 0 \\ 0 & 0
\end{bmatrix}\right\rangle;\\
\mf{h}_5&:=\left\langle\begin{bmatrix}
i/2 & 0 \\ 0 & 3i/2
\end{bmatrix},\begin{bmatrix}
j & \sqrt{3}/2 \\ -\sqrt{3}/2 & 0
\end{bmatrix},\begin{bmatrix}
k & -\sqrt{3}i/2 \\ -\sqrt{3}i/2 & 0
\end{bmatrix}\right\rangle.
\end{aligned}\label{eq:sp1subalgebras}
\end{equation}
\end{prop}

\begin{note}
The subscripts of the above Lie algebras are related to the representations of $\mf{sp}(1)$ they induce under the isomorphism $\mathfrak{sp}(2)\simeq\mathfrak{so}(5)$. The first two Lie algebras give rise to Lie subgroups composed only of isometries. In particular, the former is composed of only isoclinic rotations and the latter only simple rotations. 
\end{note}

\begin{note}
As $\mf{sp}(1)$ is simple, all the Lie subgroups associated to the above Lie algebras are compact. In particular, all the Lie subgroups are isomorphic to $\mathrm{Sp}(1)$. 
\end{note}

\begin{proof}
Recall that $\mathrm{Sp}(2)$ acts on $S^4\subseteq\mathbb{R}^5$. By extending the action, the group acts on $\mathbb{R}^5$. This action is induced by the isomorphism $\mathrm{Sp}(2)\simeq \mathrm{Spin}(5)$, meaning $\mf{sp}(2)\simeq\mf{so}(5)$. If $\mf{g}\subseteq\mf{sp}(2)$ is a Lie algebra isomorphic to $\mf{sp}(1)$, then the isomorphism induces a real 5-representation of $\mf{sp}(1)$. 

We have that irreducible real representations of $\mathfrak{sp}(1)$, denoted by $(\mathbb{R}^k,\varrho_k)$, only exist in odd dimensions as well as dimensions divisible by four. Moreover, all irreducible real representations of the same dimension are isomorphic, that is they are conjugate. The only five-representations are $(\mathbb{R}^5,\varrho_5)$, $(\mathbb{R}^4,\varrho_4)\oplus (\mathbb{R},\varrho_1)$, $(\mathbb{R}^3,\varrho_3)\oplus (\mathbb{R},\varrho_1)^{\oplus 2}$, and $(\mathbb{R},\varrho_1)^{\oplus 5}$. Note that the final representation is the trivial representation, taking all elements $\upsilon\in\mf{sp}(1)$ to zero. As the real 5-representation we have comes from a Lie algebra isomorphism, the representation must be injective. Thus, it cannot be the final representation. As we only care about Lie algebras up to isomorphic representations, it now only remains to show that the Lie algebras in the statement correspond to the real representations above. 

A map sending the generators of the Chevalley bases of $\mathfrak{sp}(2)$ and $\mathfrak{so}(5)$ to each other provides an explicit Lie algebra isomorphism between the two Lie algebras. The Chevalley bases of $\mathfrak{sp}(2)$ and $\mathfrak{so}(5)$ are given explicitly in my thesis~\cite[Appendix~F]{lang_thesis_2024}. Under such a map, we have, in particular, the following relationships between generators of the Chevalley bases:
\begin{equation}
\begin{aligned}
\begin{bmatrix}
i & 0 \\
0 & -i
\end{bmatrix}\in\mf{sp}(2)&\leftrightarrow\begin{bmatrix}
0 & 0 & 0 & 0 & 0 \\
0 & 0 & 0 & 0 & 0 \\
0 & 0 & 0 & 2 & 0 \\
0 & 0 & -2 & 0 & 0 \\
0 & 0 & 0 & 0 & 0
\end{bmatrix}\in\mf{so}(5);\\
\begin{bmatrix}
0 & 0 \\ 0 & i
\end{bmatrix}\in\mf{sp}(2)&\leftrightarrow\begin{bmatrix}
0 & 1 & 0 & 0 & 0 \\
-1 & 0 & 0 & 0 & 0 \\
0 & 0 & 0 & -1 & 0 \\
0 & 0 & 1 & 0 & 0 \\
0 & 0 & 0 & 0 & 0
\end{bmatrix}\in\mf{so}(5).
\end{aligned}\label{eq:sp2so5}
\end{equation}

A representation $(V,\rho)$ of $\mf{sp}(1)$ is determined by $Y_1:=\rho(i/2)$, $Y_2:=\rho(j/2)$, and $Y_3:=\rho(k/2)$. Furthermore, the decomposition of the representation is determined by the eigenvalues of $Y_1\otimes i$. In each Lie algebra in the statement, $Y_1$ is the first matrix listed. Using \eqref{eq:sp2so5}, we view $Y_1\in\mf{so}(5)$ in order to see its eigenvalues. Note that the eigenvalues of $Y_1\otimes i$ for $(\mathbb{R}^k,\varrho_k)$ with $k$ odd are $\left\{\frac{k-1}{2},\frac{k-3}{2},\ldots,-\frac{k-1}{2}\right\}$, each with multiplicity one. When $k$ is divisible by four, the eigenvalues of $Y_1\otimes i$ for $(\mathbb{R}^k,\varrho_k)$ are $\left\{\frac{k-2}{4},\frac{k-6}{4},\ldots,-\frac{k-2}{4}\right\}$, each with multiplicity two.

In the case $\mf{h}_{3,1,1}$, we see that $Y_1\otimes i$ is given by
\begin{equation*}
Y_1\otimes i=\begin{bmatrix}
0 & i & 0 & 0 & 0 \\
-i & 0 & 0 & 0 & 0 \\
0 & 0 & 0 & 0 & 0 \\
0 & 0 & 0 & 0 & 0 \\
0 & 0 & 0 & 0 & 0
\end{bmatrix}.
\end{equation*}
The eigenvalues of $Y_1\otimes i$ in this case are $\pm 1$, with multiplicity one each, and $0$ with multiplicity three. Thus, the representation is $(\mathbb{R}^3,\varrho_3)\oplus (\mathbb{R},\varrho_1)^{\oplus 2}$. 

In the case $\mf{h}_{4,1}$, we see that $Y_1\otimes i$ is given by
\begin{equation*}
Y_1\otimes i=\frac{1}{2}\begin{bmatrix}
0 & i & 0 & 0 & 0 \\
-i & 0 & 0 & 0 & 0 \\
0 & 0 & 0 & i & 0 \\
0 & 0 & -i & 0 & 0 \\
0 & 0 & 0 & 0 & 0
\end{bmatrix}.
\end{equation*}
The eigenvalues of $Y_1\otimes i$ in this case are $\pm 1/2$, with multiplicity two each, and $0$ with multiplicity one. Thus, the representation is $(\mathbb{R}^4,\varrho_4)\oplus (\mathbb{R},\varrho_1)$.

Finally, in the case $\mf{h}_5$, we see that $Y_1\otimes i$ is given by 
\begin{equation*}
Y_1\otimes i=\begin{bmatrix}
0 & 2i & 0 & 0 & 0 \\
-2i & 0 & 0 & 0 & 0 \\
0 & 0 & 0 & -i & 0 \\
0 & 0 & i & 0 & 0 \\
0 & 0 & 0 & 0 & 0
\end{bmatrix}.
\end{equation*}
The eigenvalues of $Y_1\otimes i$ in this case are $\pm 2$, $\pm 1$, and $0$, each with multiplicity one. Thus, the representation is $(\mathbb{R}^5,\varrho_5)$.

Suppose that $G\subseteq \mathrm{Sp}(2)$ is a connected Lie group with Lie algebra $\mathfrak{sp}(1)\simeq\mathfrak{g}\subseteq\mathfrak{sp}(2)$. Then $\mathfrak{g}$ gives a representation of $\mathfrak{sp}(1)$ isomorphic to one of $\mathfrak{h}_{3,1,1}$, $\mathfrak{h}_{4,1}$, or $\mathfrak{h}_5$. As isomorphic representations are conjugate, we have that $G$ is conjugate to the unique connected Lie group with the corresponding Lie algebra.
\end{proof}

With these Lie subalgebras of $\mf{sp}(2)$ isomorphic to $\mf{sp}(1)$ in mind, we continue with the remaining Lie subgroups. Next we examine subgroups with Lie algebra $\mf{sp}(1)\oplus\mathbb{R}$. 
\begin{prop}
The Lie subgroups of $\mathrm{Sp}(2)$ with Lie algebra $\mf{sp}(1)\oplus\mathbb{R}$ are conjugate to one of the unique connected Lie subgroups with the following Lie algebras:\label{prop:sp1s1subalgebras}
\begin{equation}
\begin{aligned}
\mf{p}_{3,1,1}&:=\left\langle \begin{bmatrix}
i/2 & 0 \\ 0 & i/2 
\end{bmatrix},\begin{bmatrix}
j/2 & 0 \\ 0 & j/2
\end{bmatrix},\begin{bmatrix}
k/2 & 0 \\ 0 & k/2
\end{bmatrix},\begin{bmatrix}
0 & 1 \\ -1 & 0
\end{bmatrix}\right\rangle;\\
\mf{p}_{4,1}&:=\left\langle \begin{bmatrix}
i/2 & 0 \\ 0 & 0 
\end{bmatrix},\begin{bmatrix}
j/2 & 0 \\ 0 & 0
\end{bmatrix},\begin{bmatrix}
k/2 & 0 \\ 0 & 0
\end{bmatrix},\begin{bmatrix}
0 & 0 \\ 0 & i
\end{bmatrix}\right\rangle.
\end{aligned}\label{eq:sp1s1subalgebras}
\end{equation}
\end{prop}

\begin{note}
Both of the Lie subgroups associated to the above Lie algebras are compact and connected. In particular, the Lie subgroup associated with $\mf{p}_{4,1}$ is isomorphic to $\mathrm{Sp}(1)\times S^1$. The Lie subgroup associated with $\mf{p}_{3,1,1}$ is isomorphic to $(\mathrm{Sp}(1)\times S^1)/\{\pm (1,1)\}$.
\end{note}

\begin{proof}
Let $\mf{g}\subseteq\mf{sp}(2)$ be a Lie subalgebra isomorphic to $\mf{sp}(1)\oplus\mathbb{R}$. Thus, $\mf{g}\simeq\mf{g}_0\oplus \mf{g}'$, where $\mf{g}_0\subseteq\mf{sp}(2)$ is isomorphic to $\mf{sp}(1)$ and $\mf{g}'\subseteq\mf{sp}(2)$ is isomorphic to $\mathbb{R}$. Additionally, $\mf{g}_0$ and $\mf{g}'$ must commute. We know from Proposition~\ref{prop:sp1subalgebras}, we know that $\mf{g}_0$ is conjugate to one of $\mf{h}_{3,1,1}$, $\mf{h}_{4,1}$, or $\mf{h}_5$. Let $\mf{g}'$ be generated by $x\in\mf{sp}(2)$. Then $x$ must commute with all of $\mf{g}_0$.

Suppose that $\mf{g}_0$ is conjugate to $\mf{h}_5$. One can show that the only element of $\mf{sp}(2)$ commuting with $\mf{h}_5$ is zero. But then $\mf{g}\simeq\mf{sp}(1)$, contradiction! Thus, $\mf{g}_0$ cannot be isomorphic to $\mf{h}_5$.

Suppose that $\mf{g}_0$ is conjugate to $\mf{h}_{3,1,1}$. One can show that the only elements commuting with $\mf{h}_{3,1,1}$ are real multiples of $\begin{bmatrix}
0 & 1 \\ -1 & 0
\end{bmatrix}$. As choosing any non-zero multiple induces the same Lie algebra, we see that we get $\mf{p}_{3,1,1}$ as in the statement.

Finally, suppose that $\mf{g}_0$ is conjugate to $\mf{h}_{4,1}$. One can show that the only elements commuting with $\mf{h}_{4,1}$ are given by $\begin{bmatrix}
0 & 0 \\ 0 & \upsilon
\end{bmatrix}$, for some $\upsilon\in\mf{sp}(1)$. Up to conjugation, we may assume that $\upsilon$ is a real multiple of $i$. As the choice of non-zero multiple induces the same Lie algebra, we see that we get $\mf{p}_{4,1}$ as in the statement.

Just as in Proposition~\ref{prop:sp1subalgebras}, when dealing with connected Lie subgroups up to conjugacy, it is enough to investigate Lie subalgebras up to conjugacy.
\end{proof}

We finally consider the subgroups with Lie algebra $\mf{sp}(1)\oplus\mf{sp}(1)$, as the $\mf{sp}(2)$ case is trivial. Indeed, only $\mathrm{Sp}(2)$ has a Lie algebra $\mf{sp}(2)$. Note that $\mathrm{Sp}(2)$ is compact.
\begin{prop}
The Lie subgroups of $\mathrm{Sp}(2)$ with Lie algebra $\mf{sp}(1)\oplus\mf{sp}(1)$ are conjugate to $\mathrm{Sp}(1)\times\mathrm{Sp}(1)\subseteq\mathrm{Sp}(2)$, the subgroup comprised of diagonal matrices.
\end{prop}

\begin{note}
The subgroup $\mathrm{Sp}(1)\times \mathrm{Sp}(1)$ is compact.
\end{note}

\begin{proof}
Let $\mf{g}\subseteq\mf{sp}(2)$ be a Lie subalgebra isomorphic to $\mf{sp}(1)\oplus\mf{sp}(1)$. Thus, $\mf{g}\simeq \mf{g}_1\oplus \mf{g}_2$, where both $\mf{g}_1\subseteq\mf{sp}(2)$ and $\mf{g}_2\subseteq\mf{sp}(2)$ are isomorphic to $\mf{sp}(1)$. Additionally, $\mf{g}_1$ and $\mf{g}_2$ must commute. We know from Proposition~\ref{prop:sp1s1subalgebras} that no non-zero elements of $\mf{sp}(2)$ commute with $\mf{h}_5$. Hence, neither $\mf{g}_1$ nor $\mf{g}_2$ are conjugate to $\mf{h}_5$.

Suppose that $\mf{g}_1$ is conjugate to $\mf{h}_{3,1,1}$. We know that the only elements of $\mf{sp}(2)$ commuting with $\mf{h}_{3,1,1}$ are the real multiples of $\begin{bmatrix}
0 & 1 \\ -1 & 0
\end{bmatrix}$. Thus, $\mf{g}_2$ must be contained in the span of this matrix. However, $\mf{g}_2$ is three-dimensional, contradiction! 

Therefore, $\mf{g}_1$ is conjugate to $\mf{h}_{4,1}$. We know that the only elements of $\mf{sp}(2)$ commuting with $\mf{h}_{4,1}$ are the elements of the form $\begin{bmatrix}
0 & 0 \\ 0 & \upsilon
\end{bmatrix}$ for $\upsilon\in\mf{sp}(1)$. This space is three-dimensional, so $\mf{g}_2$ is exactly this space. Thus, we have that $\mf{g}$ is exactly given by 
\begin{equation}
\mf{g}=\left\{\begin{bmatrix}
\upsilon_1 & 0 \\ 0 & \upsilon_2
\end{bmatrix}\mid \upsilon_1,\upsilon_2\in\mf{sp}(1)\right\}.
\end{equation}
This space is exactly the Lie algebra of $\mathrm{Sp}(1)\times\mathrm{Sp}(1)\subseteq\mathrm{Sp}(2)$, the subgroup comprised of diagonal matrices.

Just as in the preceding propositions, when dealing with connected Lie subgroups up to conjugacy, it is enough to investigate Lie subalgebras up to conjugacy.
\end{proof}

Table~\ref{table:conformalsubgroups} summarizes the Lie subgroups mentioned in the preceding propositions in addition to $\mathrm{Sp}(2)$, lists their corresponding Lie algebras, and names the corresponding symmetries. The preceding propositions tell us that when studying instantons equivariant under continuous subgroups of $\mathrm{Sp}(2)$, it suffices to study these subgroups.
\newpage

\bibliographystyle{halpha}
\phantomsection
\addcontentsline{toc}{section}{\textbf{References}}
\bibliography{./Files/Bibliography/bibliography.bib}
\end{document}